\documentclass[10pt]{article}
\usepackage{lineno,hyperref}
\modulolinenumbers[5]
\usepackage{authblk}
\usepackage[table,xcdraw]{xcolor}
\usepackage[utf8]{inputenc}
\usepackage[T1]{fontenc}
\usepackage[colorinlistoftodos]{todonotes}

\usepackage{tikz}
  \usetikzlibrary{patterns}
\usepackage{standalone}     
  \standaloneconfig{mode=image} 

\usepackage{graphicx}
\usepackage{caption}           
\usepackage{subcaption}
\usepackage{fancybox}
\usepackage{float}      	     
\usepackage{algorithm}
\usepackage{algcompatible}
\usepackage{multirow}           
\usepackage{hhline}             
  \renewcommand{\algorithmicrequire}{\textbf{Input:}}
  \renewcommand{\algorithmicensure}{\textbf{Output:}} 
  \algblockdefx{FORALLP}{ENDFORALLP}[1]%
    {\textbf{for all }#1 \textbf{do in parallel}}%
    {\textbf{end for}}
  \algloopdefx{RETURN}[1][]{\textbf{return} #1}  
\usepackage{color, colortbl}
\usepackage{array}
\usepackage{amsmath,amsfonts,amssymb,amsthm,xspace}
\usepackage{ bbold }

 \newlength{\myMheight}
\newcommand{\II}{\mathbb{1}}
\newcommand{\OO}{\mathbb{0}}

\newcommand{\vvee}{\doubleSim{$\vee$}}
\newcommand{\wwedge}{\doubleSim{$\wedge$}}
\newcommand{\nneg}{\doubleSim{$\neg$}}
\newcommand{\redsquare}{\textcolor{red}{\blacksquare}}
\newcommand{\gsquare}{\textcolor{green}{\blacksquare}}
\newcommand{\bsquare}{\textcolor{blue}{\blacksquare}}

\newcommand{\NL}{\textbf{NL}\xspace}
\newcommand{\NC}{\textbf{NC}\xspace}
\newcommand{\Pt}{\textbf{P}\xspace}
\newcommand{\NP}{\textbf{NP}\xspace}
\newcommand{\PtC}{\textbf{P}-Complete\xspace}
\newcommand{\NPC}{\textbf{NP}-Complete\xspace}
\newcommand{\cO}{\mathcal{O}}
\newcommand{\AsyncStability}{\textsc{AsyncUnstability}\xspace}
\newcommand{\SAT}{\textsc{SAT}\xspace}
\newcommand{\CSAT}{\textsc{Circuit-SAT}\xspace}
\newcommand{\GCSAT}{\textsc{GridCircuit-SAT}\xspace}
\newcommand{\RGCSAT}{\textsc{Restricted-GridCircuit-SAT}\xspace}
\newcommand{\Robpred}{\textsc{Robust-Prediction}\xspace}
\newcommand{\stability}{\textsc{Unstability}\xspace}
\newcommand{\defproblemu}[3]{
  \vspace{1mm}
\noindent\fbox{
  \begin{minipage}{0.95\textwidth}
  #1 \\
  {\bf{Input:}} #2  \\
  {\bf{Question:}} #3
  \end{minipage}
  }
  \vspace{1mm}
}

\newcommand\doubleSim[1]{%
  \mathrel{
    \vbox{\offinterlineskip
      \ialign{##\cr
              #1\crcr
              \noalign{\vskip-0.6ex}
              #1\cr
      }
    }
  } 
} 

\newtheorem{theorem}{Theorem}
\newtheorem{lemma}{Lemma}
\newtheorem{proposition}{Proposition}
\newtheorem{definition}{Definition}

\newcommand{\PreserveBackslash}[1]{\let\temp=\\#1\let\\=\temp}
\newcolumntype{C}[1]{>{\PreserveBackslash\centering}p{#1}}
\newcolumntype{R}[1]{>{\PreserveBackslash\raggedleft}p{#1}}
\newcolumntype{L}[1]{>{\PreserveBackslash\raggedright}p{#1}}

\title{On the Complexity of Asynchronous Freezing Cellular Automata }

\providecommand{\keywords}[1]{\textbf{\textit{Keywords:}} #1}
\begin{document}
\author[1]{Eric Goles}
\author[1]{Diego Maldonado}
\author[1]{Pedro Montealegre}
\author[2]{Mart\'in R\'ios-Wilson}


\affil[1]{Facultad de Ingenieria y Ciencias, Universidad Adolfo Iba\~nez, Santiago, Chile.}
\affil[2]{Departamento de Ingenier\'ia Matem\'atica,  Facultad de Ciencias F\'isicas y Matem\'aticas, Universidad de Chile, Santiago, Chile.}
\maketitle

\begin{abstract}
In this paper we study the family of freezing cellular automata (FCA) in the context of asynchronous updating schemes. A cellular automaton is called \emph{freezing} if there exists an order of its states, and the transitions are only allowed to go from a lower to a higher state. A cellular automaton is asynchronous if at each time-step only one cell is updated. 
Given configuration, we say that a cell is \emph{unstable} if there exists a sequential updating scheme that changes its state. In this context, we define the problem \AsyncStability, which consists in deciding if a cell is unstable or not. In general \AsyncStability is in \NP, and we study in which cases we can solve the problem by a more efficient algorithm. 

We begin showing that  \AsyncStability is  in \NL for any one-dimensional FCA. Then we focus on the family of \emph{life-like freezing} CA (LFCA), which is a family of two-dimensional two-state FCA that generalize  the \emph{freezing version} of the game of life, known as \emph{life without death}. We study the complexity of  \AsyncStability for all LFCA in the triangular and square grids, showing that almost all of them can be solved in \NC, except for one rule for which the problem is \NPC.

\end{abstract}





\section{Introduction}\label{subsec: intro}
Cellular Automata (CA) are discrete mathematical models initially developed by Ulam and Von Neumann in the 1940's in order to study self-replicating systems. These models can be described as dynamical systems in which the state space is defined by different uniform units, called cells. Each of these cells has an associated state and the local interactions between them, that is to say, the interaction between the cells that are nearby (in a given neighborhood) will define the consecutive transitions of the dynamics of the CA according to some rule, that we will call the \emph{local rule} of the CA. Additionally, it is assumed that all these interactions take place simultaneously, so we say that the states are updated in a \emph{synchronous scheme}. In the last 70 years, CA have been broadly studied\cite{wolfram1983statistical,martin1984algebraic,wolfram1984universality,wolfram1984computation,wolfram1985undecidability} and there are several applications of these models in Biology \cite{adamatzky2015actin,lehotzky2019cellular,deveaux2019defining}, Sociology \cite{hoekstra2010simulating,torrens2001cellular,hegselmann1998understanding}, Computer Science \cite{10.1007/11786986_13,conf/focs/Banks70}, \cite{di2008computational} etc. In this context, the latter assumption on the interactions between all the cells being perfectly synchronized is very useful as it turns the model into a massive parallel computing environment \cite{cannataro1995parallel}. Nevertheless, for several applications this assumption may be unrealistic. This statement can be explained from two perspectives: first,  from the point of biology \cite{cornforth2003artificial} in which these perfectly simultaneous interactions between cells are fairly rare and second, from the point of view of a computation model, the synchronous update scheme implies the existence of an internal clock that synchronizes each processor which is computing the state of a cell. From the perspective of the hardware design for performing the calculations, this implies an increase in the complexity and it also implies a decreasing on the efficiency of simulating the model. Therefore, different alternative approaches to the synchronous update scheme, all of them generically called \emph{asynchronous update schemes},  have been proposed. Some of the most known types of these update schemes are the following:
\begin{enumerate}
 \item Fully asynchronous update scheme \cite{fates2006fully}: at each time step, the local rule of the CA is applied to a single cell that is uniformly randomly chosen.
 \item $\alpha$-Asynchronous update scheme \cite{fates2004experimental}: at each time step, each cell has a probability $\alpha$ to update its state and a probability $1-\alpha$ to stay in the same state. The parameter $\alpha$ is called the \emph{synchronization rate}.
 \item Fixed random sweep \cite{SCHONFISCH1999123,kitagawa1974cell}:  a randomly chosen permutation of the cells is fixed at initial time and from there, at each time, cells are updated according to this order.
 \item Random new sweep \cite{SCHONFISCH1999123,kitagawa1974cell}: at each time step, cells are updated according a randomly chosen permutation of the cells.
 \item Non-random update schemes\cite{robert2012discrete}: this class of update schemes includes the sequential update schemes, in which, cells are updated one by one, according to a given order. We will refer to these by simply update schemes, when the context is clear. In addition, it is also possible to not just update one cell at each time,  but to update a set of cells simultaneously. This is called a \emph{block sequential update scheme}.
\end{enumerate}

In this paper, we will consider a specific class of asynchronous update schemes, in which only a single cell is updated at the same time following a pre-set order. We will refer to the latter class as simply asynchronous update schemes whenever the context is clear. Many of the results presented in this article can be extended to other types of updating schemes, but always considering that only a single cell is updated in every time-step. Observe that when a cell is updated its state may or may not change according to the local rule and the states of its neighbors. In this regard, we distinguish the situation when the state of the cell changes after updating it, and we say that the cell is \emph{iterated} to emphasize this fact.

We focus in the asynchronous versions of freezing cellular automata (FCA). This model was  introduced  by Goles et al. \cite{goles:hal-01294144}, in order to study dynamics that are inherently irreversible. A cellular automaton is \emph{freezing} if there exists a partial order ($\leq$) of the state set, and a cell can only update its current state to a greater state according to this order. For example, one can define an order in which the set of states is partitioned into  \textit{active} and \textit{inactive} states, such that an inactive cell that is iterated to an active state will permanently stay in some active state. 

We call freezing asynchronous cellular automata to the FCA which its dynamics is defined by some asynchronous update scheme. A straightforward result in the context of the study of these CA is that any initial periodic configuration with $n$ cells eventually reaches a fixed point in at most $\mathcal{O}(n^2)$ steps. This is because after updating each cell, at least one of the $n$ cells is iterated (effectively change its state), and a given cell can change at most $|Q|$ times. Therefore, in at most $n^2(|Q|-1) $ time-steps, the dynamics reaches a fixed point.

In that context, we observe that, given an initial configuration, there might exist some cells that will always remain in their initial state, regardless of the chosen update scheme.  We call these cells \emph{stable cells}. Conversely, a cell is \emph{unstable} if there exists an updating scheme that iterates it.  We will call \AsyncStability to the problem of deciding, given an initial condition and a cell, if the given cell is unstable. 

 Our results classify FCA according to the computational complexity of {\AsyncStability}. The computational complexity of a decision problem is defined as the amount of resources (e.g., time, space, number of processors) required to give an answer, as a function of the size of input.  Our main aim is to understand what makes a FCA rule \emph{simple} or \emph{complex}. Roughly speaking, a rule that has a \emph{hard} complexity will require an algorithm solving  {\AsyncStability} to simulate the rule step by step on \emph{every possible} updating scheme. Conversely, if the complexity is \emph{low}, then there exist some properties of the dynamics that can be exploited algorithmically in order to give an efficient solution, better than simply simulating the dynamics exhaustively for each updating scheme.

In that context, we consider the following complexity classes, usually considered in the context of the theory of computational complexity: 
\begin{itemize}

\item {\NP} is the class of problems verifiable in polynomial time. In other words, the class {\NP} is the set of problems for which, given a polynomial-sized certificate, a yes-instances can be verified by a polynomial time-algorithm. Observe that {\AsyncStability} belongs to {\NP}, because here the updating scheme plays the role of the certificate, i.e. given \emph{the right} updating scheme, it is possible to verify is a cell is unstable if we simply simulate the dynamics according to the order given by the updating scheme. 

\item \Pt is the class of problems solvable in polynomial time by a deterministic Turing machine. The convention states that \Pt is the class of problems that can be \emph{efficiently solved} with respect to computation time. 

\item \NC is the class of problems solvable in poly-logarithmic time in a parallel machine (PRAM) using a polynomial number of processors. For this class, the convention states that $\NC$ is the class of problems that are \emph{efficiently paralelizable} \cite{Greenlaw:1995}. 

\item Finally, \NL is the  class of problems solvable by a non-deterministic logarithmic space Turing machine. If {\AsyncStability} belongs to {\NL} it means that the problem can be \emph{verified} with extreme efficiency in the use of memory, i.e., given the \emph{right certificate}, one can verify a \emph{yes-instance} of {\AsyncStability} using only \emph{logarithmic-space}. 

\end{itemize}

It is well-known (see for example \cite{arora2009computational}) that  $ \textbf{NL} \subseteq \textbf{NC} \subseteq \textbf{P} \subseteq \textbf{NP}$  and it has been conjectured that all inclusions are proper, meaning that,  there may be problems in $\textbf{NC}$ that do not belong to $\textbf{NL}$, problems in $\textbf{P}$ that are not in  $\textbf{NC}$, and problems in $\textbf{NP}$ that do not belong to $\textbf{P}$.  

As we know that \AsyncStability is in general in \NP, we would like to know whether there exist specific sets of FCA rules for which this problem is in a lower complexity class, that is to say, if there are some particular FCA rules for which we can exhibit algorithms that are more efficient than the exhaustive simulation approach. Conversely, if we are unable to show such algorithms, we would like to provide evidence that it is impossible to solve the \AsyncStability more efficiently. 

The problems in $\textbf{NP}$ that are the most likely to not belong to $\textbf{P}$, are the \textbf{NP}-Complete problems. A problem is  \textbf{NP}-Complete if any other problem in $\textbf{NP}$ can be reduced by a polynomial time reduction to it. One of the best-known is the \textsc{Boolean Satisfiability problem (SAT)} \cite{Cook:1971:CTP:800157.805047}, consisting in deciding if a given CNF formula can be satisfied. Roughly speaking, the fact that \textsc{SAT} is {\NP}-Complete means that any algorithm that solves the problem has to essentially try all possible combinations of the values of the input variables.  Analogously, the problems in $\textbf{P}$ that are the most likely to not belong to $\textbf{NC}$ are the $\textbf{P}$-complete problems, which are the problems to which every other problem in the class $\textbf{P}$ can be reduced by a function computable in logarithmic space~\cite{Greenlaw:1995}. 

In the context of the study of the dynamics of freezing cellular automata there exist one well-studied rule that is known as \textit{Life without dead}. This rule was introduced for the first time by Toffoli and Margolous in \cite{toffoli1987cellular}, who called it $\textit{Inkspot}.$ As the name Life without dead suggests, it is simply the freezing version of the well-known Conway's Game of Life  \cite{ConwaysLife,Durand1999, Berlekamp1982}. In this rule, the transitions are the same that in Game of Life (that is to say, a cell will born if exactly three of its neighbors are active) with the exception that if a cell borns it will remain alive for all time-step. Within this framework, Griffeath and Moore \cite{RePEc:wop:safiwp:97-05-044} have studied decision problems such as \textsc{Prediction}, which consist in predicting, for a finite amount of time, if a cell will change its state. In this work it is shown that the latter problem is $\Pt$-complete as a consequence of the capability of the rule to simulate boolean circuits. This result is based in the richness of the local interactions defined by the rule which produce the emergency of traveling patterns called \textit{ladders}. Additionally, it is shown that answering if a given finite configuration grows to infinity is $\Pt$-hard and in the case of a given initial condition with periodic background the question is undecidable.  However, all these studies have been focus in the synchronous case, leaving open the questions regarding how difficult are these problems in the asynchronous case.

On the other hand, the study of {\AsyncStability} has been previously addressed within other relevant contexts.  In fact,  in  \cite{StabilityMajority}  this problem is studied for the \emph{freezing majority cellular automaton} (FMCA), i.e., the two-state freezing cellular automaton (say with \emph{inactive} and \emph{active} cells) for which an inactive cell is updated to the state of the most represented cell in its neighborhood.  As it is shown in \cite{StabilityMajority}, for each configuration of inactive and active cells, an inactive cell is stable for the FMCA if and only if  the same cell remains inactive in the fixed point reached by the synchronous update of the FMCA. In other words, if a cell does not become active in the synchronous update of FMCA, it wont become active under any updating scheme, and vice-versa. 

Interestingly, the latter property does not only hold for the FMCA but for every monotone FCA. A CA $F$ is called \emph{monotone} if there exists a total order of the set of states of $F$, such that for every pair of configurations $x$ and $y$, if $x \leq y$ then $F(x) \leq F(y)$ (where last inequalities are coordinate-wise). As it is shown in \cite{StabilityMajority}, for every monotone FCA $F$, and every initial configuration, the stable cells are exactly the cells that remain in their initial configuration on the fixed point reached by $F$ updated synchronously (note that in \cite{StabilityMajority} the property is shown for two-state monotone CA, but the result is indeed easly extended to the case in which there are more states). The latter observation implies that  {\AsyncStability} is in \Pt for every monotone FCA. In fact, it suffices to simply simulate the local rule under the synchronous updating scheme until the dynamic reaches a fixed point.
\subsection{Our results}

We start studying the one-dimensional FCA. We show that, restricted to one-dimensional FCA,  \AsyncStability  belongs to the class {\NL}.  Since $\NL \subseteq \Pt$, this directly implies that we can solve \AsyncStability with a much more efficient algorithm rather than the exhaustive simulation of all possible updating schemes. This algorithm is an extension of a result of Goles et. al.  \cite{goles:hal-01294144} where it is shown that, for all one-dimensional FCA (i.e. a freezing CA where every cell is updated synchronously),  there is a non-deterministic logarithmic-space algorithm that, given a initial configuration, computes the state of a cell on any given time-step. Here we present an adaptation of the algorithm proposed in \cite{goles:hal-01294144} to solve \AsyncStability on all one-dimensional FCA. 

We remark that, unlike the \NP algorithm solving  \AsyncStability explained above,  in this context, the certificate of this given \NL algorithm is not the updating scheme that iterates the given input cell. 

Roughly, the certificate of the  algorithm presented in \cite{goles:hal-01294144} consists in the time-steps on which each cell is iterated. If the given local rule is freezing and has $Q$ states, the certificate of each cell can be stored in $\cO(|Q|\log n)$ bits of information. In order to verify the certificate in logarithmic-space, the algorithm uses the fact that the FCA is one-dimensional and sequentially reads the certificates in order, starting from the left-most cell and finishing with the rightmost one. The algorithm checks, given the information related to a cell and its neighbors (the time steps on which they change their state), whether it is valid, in the sense that, it describe consistently a valid iteration of the FCA. In order to solve \AsyncStability, we adapt the algorithm presented in the latter work to look through the options to decide whether the iterations are consistent with an asynchronous updating scheme. In order to do that, we show it suffices to check that each pair of adjacent cells are not updated at the same time. Fortunately, this can be done during the verification process of the latter algorithm.

As \AsyncStability restricted to one-dimensional FCA has a relatively \emph{low} complexity, we wonder to what extent we can show the same for two (or more) dimensional FCA. In this sense, we focus our study to \emph{life-like freezing} CA (LFCA). A FCA is called a LFCA if it has two states, namely \emph{inactive} and \emph{active}, and (2) the local transition function satisfies that there exist a pair of non-negative integers $k_1 \leq k_2$ such that, if an inactive cell has at least $k_1$ and at most $k_2$ active neighbors it becomes active, and otherwise it remains inactive (active cells never become inactive). We call $R_{k_1,k_2}$ the LFCA rule with parameters $k_1, k_2$. This family of rules includes many interesting cases. For example in two dimensions with Moore neighborhood, $R_{3,3}$ corresponds to the the \emph{life without death}. In two dimensions with von Neumann neighborhood $R_{2,4}$ is the freezing majority rule, and $R_{3,4}$ is the freezing strict majority rule.

We study the family of LFCA in the triangular and square grids with von Neumann neighborhood. In that context, we rename rule $R_{k_1,k_2}$  by $Tk_1k_2$ (respectively  $Sk_1k_2$) when rule $R_{k_1,k_2}$ is defined in the triangular grid (respectively in the square grid). Therefore, there exist $10$ different LFCA defined in the triangular grid, namely rules $T00$, $T01$, $T02$, $T03$, $T11$, $T12$, $T13$, $T22$, $T23$ and $T33$. Similarly, there exist $15$ different LFCA defined in the square grid, namely rules $S00$, $S01$, $S02$, $S03$, $S04$, $S11$, $S12$, $S13$, $S14$, $S22$, $S23$, $S24$, $S33$, $S34$ and $S44$. 

In the previous context, we show that {\AsyncStability} is in \NC when the problem is restricted to any LFCA defined in a triangular grid. Moreover, we show that the same is true for almost all LFCA defined over the square grid, with the exception rule $S22$. For the upper-bounds, unlike the case of one-dimensional FCA, we are not able to exhibit a single algorithm that solves {\AsyncStability} for all LFCA. Instead, we classify rules according in three groups, according to the complexity of {\AsyncStability} for the corresponding rule. 

\begin{enumerate}
  \item \emph{Trivial rules}: Rules for which we can decide if a cell is unstable by simply inspecting its neighborhood. 
  \item \emph{Infiltration Rules}: Rules where there exists a connected set $U$ of inactive cells that has the following property: in order to decide the unstability of the cells in $U$, one simply has to look to the boundary of $U$. 
  \item \emph{Monotone-like Rules}: Rules which dynamic can be related to a monotone rule, and therefore the algorithm solving \AsyncStability uses a result given in \cite{StabilityFTCA}, that presents a relation between the dynamics of a sequentially updated monotone freezing rule with respect to the synchronous dynamics of the same rule.
\end{enumerate}

As we said above, there is one LFCA rule defined over the square grid for which we are not able to provide efficient algorithms, namely the rule $S22$. In the last part of this article, we tackle the problem \AsyncStability restricted to this rule. Is relevant to notice that rule $S22$ is known to be \emph{hard} in the context of \textsc{(Synchronous)-Stability} (i.e. the problem of deciding if a given cell is stable for the synchronous update of a given initial configuration) \cite{StabilityFTCA}. More precisely, it is known that \textsc{Stability} restricted to any FCA is solvable in polynomial time, because the dynamics of any FCA reaches a fixed point in a polynomial number of time-steps. In \cite{StabilityFTCA} it shown that \textsc{Stability} restricted to rule $S22$ is \PtC, meaning that the existence of a (more) efficient algorithm solving the problem (say, by a \NL algoritm) it is extremely unlikely. 

We show that \AsyncStability restricted to rule $S22$ is \NPC. In order to prove our result, we show that any FCA capable of simulating a certain set of gadgets satisfies that \AsyncStability, restricted to that FCA, is \NPC. This  gadgets are, roughly, an \emph{conjunction gate}, an \emph{disjunction gate} and a \emph{selector}. A conjunction gate is a squared pattern that receives as input two signals (namely \emph{true} or \emph{false}) on the top and left edges, and outputs the conjunction of the two signals through the left and bottom edges. A \emph{disjunction gate} is defines analogously, except that the output values is the disjunction of the two input signals. Our result requires a \emph{robustness} property of conjunction and disjunction gates. This robustness means that (1) when these gates are supposed to output a true value (e.g. when there is two true signals entering through a conjunction gate) then there must exists \emph{at least} one updating scheme of the pattern on which the dynamics produces a true value on both output sides; and (2) when these gates are supposed to output a false value (e.g. when in a conjunction gate one input signal is false), then both outputs sides must be false, for \emph{every} updating scheme of the gadget. The third gadget is a \emph{selector}, which has no inputs and two output edges. The selector has the property that on any updating scheme \emph{at most} one output edge can send a \emph{true} signal, while the other must send a \emph{false} signal. Roughly speaking, an updating scheme of the selector gadget \emph{selects} (at most) one over two possible paths. We show that these gadgets can be used to simulate \textrm{SAT}. More precisely, disjunction gates are used to simulate the clauses, conjunction gates are used to simulate the conjunction of all clauses, and the selector is used to choose between truth values of the variables. Moreover, we show that the combination of these gadgets allows to construct a \emph{crossing gadget}, which allows to simulate a non-necessarily planar formula in a planar topology. The \NP-Completeness associated to rule $S22$ follows by simply exhibiting the grid patterns that simulate previous gadgets. 

\subsection*{Related Work}

Up to our knowledge, most previous research regarding freezing cellular automata consider only synchronous updating schemes. Perhaps the first work involving FCA and complexity was done by Moore in \cite{RePEc:wop:safiwp:97-05-044}, where it is studied a rule called the \emph{life without death}, which corresponds to the well-known \emph{game of life}, when living cells cannot die. In his paper, Moore shows that, for the life without death, \textsc{Prediction} is  \Pt-Complete. \textsc{Prediction} is the problem consisting in, given an initial configuration and a positive integer $t^*$,  computing the state of a cell after a given number of $t^*$ (synchronous) time-steps. Observe that for any FCA, \textsc{Prediction} can be solved in polynomial time by simply simulating the automaton, because for any updating scheme the dynamics reach a fixed-point in a polynomial number of time-steps. The fact that \textsc{Prediction} is  \Pt-Complete for this rule implies that, roughly, simple simulation is essentially \emph{the best algorithm} for solving \textsc{Prediction} for the life without death FCA.

Another example of a research involving freezing dynamics and complexity is given in \cite{goles:hal-00914603}, where it is studied the majority freezing automaton, also known as \emph{bootstrap percolation}  \cite{0022-3719-12-1-008}.  This rule is studied on an arbitrary topology, i.e., the cells are located the vertices of a given undirected graph. In this case, an inactive cell will become active, if the most represented cells in its neighborhood are active cells. In  \cite{goles:hal-00914603} it is shown that for the freezing majority rule  \textsc{Prediction} is in $\NC$ when the maximum degree of the input graph is at most $4$ (for example the two-dimensional grid) and in the family of graphs of degree at least $5$ the same problem is \PtC.

Previous examples consider specific freezing rules, which that have only two states. There exist also studies that study the complexity of families of FCA. In  \cite{goles:hal-01294144}  Goles et. al introduce the one-dimensional FCA (with an arbitrary number of states) and show that \textsc{Prediction} restricted this any rule in this family is in \NL.  More recently, in \cite{StabilityFTCA} is proposed a study of the complexity of the 32 different two-dimensional and two-state totalistic freezing CA (2FTCA) with von Neumann neighborhood. In that paper, the decision problem is \textsc{Stability}, which consists in deciding if a given cell is stable in the fixed-point reached by the automata given a finite periodic configuration.  The study of \cite{StabilityFTCA} shows that (unlike one-dimensional FCA) there is a two-state totalistic FCA that on which \textsc{Stability} is \PtC.
 
Regarding asynchronous CA and complexity, in \cite{StabilityMajority} it is studied the complexity of the \emph{asynchronous majority rule} (i.e. the majority rule updated asynchronously). The authors study the problem \textsc{Asynchronous-Prediction} (\textsc{AsyncPrediction})  restricted to this rule.  \textsc{AsyncPrediction} is the problem consisting in, given an initial configuration and a positive integer $t^*$ and a cell $v$, decide whether there exists an asynchronous updating scheme on which $v$ switches its state after updating $t^*$ cells. As we explained above, an interesting result of this work shows that {\AsyncStability} is solvable in polynomial time, for all monotone FCA.   Even though {\AsyncStability} and  \textsc{AsyncPrediction}  may look similar problems, in \cite{StabilityMajority} it is shown that their complexity can be very different. Indeed, restricted to freezing majority automaton in three-dimensions, the first problem is solvable in polynomial time (actually is \PtC \cite{RePEc:wop:safiwp:96-08-060}) while the second one is \NPC.  


\subsection*{Organization of this paper}

This paper is organized in the following way: first, in Section \ref{sec: definitions} we give the main formal definitions and previous results. Later, in Section \ref{sec: 1DF} we show that for all one-dimensional asynchronous FCA problem \AsyncStability is in \NL. In Section \ref{sec:upperbounds} we study the complexity upper-bounds for the problem {\AsyncStability} restricted to the family of LFCA defined over a triangular or square grid, showing that for all  rules except rule $S22$ the problem is solvable in \NC. Finally in Section \ref{sec:complexityofruleS22} we show that for rule $S22$ problem \AsyncStability is \NPC. Section \ref{sec:conclu} finishes the article with a discussion of the results and further questions.


\section{Preliminaries }\label{sec: definitions}
In this section we give the formal definitions used along the article. We begin defining the topologies that we consider, namely the one-dimensional grid, and the triangular and square two-dimensional grid. We continue defining freezing cellular automata, asynchronous updating schemes, and problem \textsc{Asynchronous-Unstability}. Then we define the family of life-like freezing cellular automata. Finally, we give the definitions of the complexity clases that will appear in our results. 

\subsection{Cellular automata and grids}
 Cellular automata are  discrete dynamical systems defined on a regular grid of cells, where each cell change its state by the action of a local function or automata rule, which depends on the state of the cell and the state of its neighbors. In this work we will consider three classes of grids:
\begin{itemize}
  \item the one-dimensional grid, where the cells are arranged in a line. The neighborhoods in this case are defined by the adjacent cells, that is to say, the left and right neighbors.
  \item the two-dimensional grid, where we consider two possible tessellations of plane: the tessellation by triangles and tessellation by squares.  In both cases, the neighbors of each cell are simply given by the nearest cells as it is shown in Figure \ref{fig: grids}. Thus, each cell has $3$ and $4$ neighbors respectively. This definition of neighborhood is known in literature as Von Neumann neighborhood.
\end{itemize}

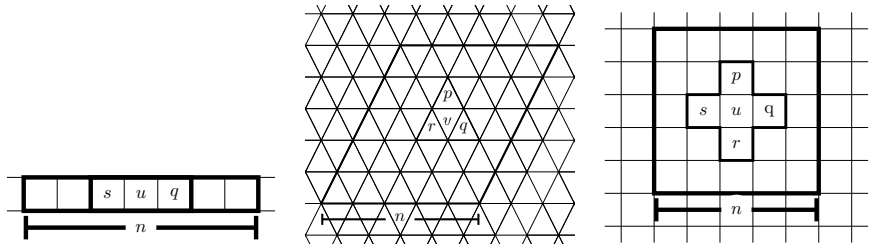
\begin{figure}
\centering
\resizebox{0.3\textwidth}{!}{%
  \begin{tikzpicture}[scale=0.8, every node/.style={scale=1.25}]
  \node at (0.5,0.5) {$u$};
  \node at (1.5,0.5) {$q$};
  \node at (-0.5,0.5) {$s$};
  \draw [line width=3pt,|-|] (-3,-0.5) -- (4,-0.5);
  \node[fill=white,circle] at (0.5,-0.5) {$n$};
  \draw [line width=3pt](-1,0) rectangle (2,1);
  \draw [line width=3pt](-3,0) rectangle (4,1);
  \draw (-2,1) -- (-2,0);
  \draw (-1,1) -- (-1,0);
  \draw (0,1) -- (0,0);
  \draw (1,1) -- (1,0);
  \draw (2,1) -- (2,0);
  \draw (3,1) -- (3,0);
  \draw (-3,1) -- (-3,0);
  \draw (4,1) -- (4,0);
  \draw (4.5,1) -- (-3.5,1);
  \draw (4.5,0) -- (-3.5,0);
  \end{tikzpicture}
}\hfil
\resizebox{0.3\textwidth}{!}{
\begin{tikzpicture}[scale=1.6, every node/.style={scale=2.5}]
\clip (11.5,-15.25) rectangle (20,-7.75);
\node[fill=white,circle]  at (16,-11.4) {$u$};
\node[fill=white,circle] at (16,-10.6) {$p$};
\node[fill=white,circle] at (16.5,-11.6) {$q$};
\node[fill=white,circle] at (15.5,-11.6) {$r$};
\draw [line width=3pt,|-|](12,-14.5) -- (17,-14.5);
\node[fill=white,circle] at (14.4921,-14.4478) {$n$};
\draw (10,-6)  -- (9.5,-7) -- (10.5,-7) -- (10,-6);
\draw (10.5,-7) -- (10,-6) -- (11,-6) -- (10.5,-7);
\draw (11,-6)  -- (10.5,-7) -- (11.5,-7) -- (11,-6);
\draw (11.5,-7) -- (11,-6) -- (12,-6) -- (11.5,-7);
\draw (12,-6)  -- (11.5,-7) -- (12.5,-7) -- (12,-6);
\draw (13,-6)  -- (12.5,-7) -- (13.5,-7) -- (13,-6);
\draw (14,-6)  -- (13.5,-7) -- (14.5,-7) -- (14,-6);
\draw (15,-6)  -- (14.5,-7) -- (15.5,-7) -- (15,-6);
\draw (16,-6)  -- (15.5,-7) -- (16.5,-7) -- (16,-6);
\draw (17,-6)  -- (16.5,-7) -- (17.5,-7) -- (17,-6);
\draw (18,-6)  -- (17.5,-7) -- (18.5,-7) -- (18,-6);
\draw (18.5,-7) -- (18,-6) node (v23) {} -- (19,-6) -- (18.5,-7);
\draw (19,-6)  -- (18.5,-7) -- (19.5,-7) -- (19,-6);
\draw (19.5,-7) -- (19,-6) -- (20,-6) -- (19.5,-7);
\draw (20,-6)  -- (19.5,-7) -- (20.5,-7) -- (20,-6);
\draw (20.5,-7) -- (20,-6) -- (21,-6) -- (20.5,-7);
\draw (21,-6)  -- (20.5,-7) -- (21.5,-7) -- (21,-6);
\draw (21.5,-7) -- (21,-6) -- (22,-6) -- (21.5,-7);
\draw (10,-16)  -- (9.5,-17) -- (10.5,-17) -- (10,-16);
\draw  (10.5,-17) -- (10,-16) -- (11,-16) -- (10.5,-17);
\draw (11,-16)  -- (10.5,-17) -- (11.5,-17) -- (11,-16);
\draw   (11.5,-17) -- (11,-16) -- (12,-16) -- (11.5,-17);
\draw (12,-16)  -- (11.5,-17) node (v25) {} -- (12.5,-17) -- (12,-16);
\draw (13,-16)  -- (12.5,-17) node (v24) {} -- (13.5,-17) -- (13,-16);
\draw (14,-16)  -- (13.5,-17) -- (14.5,-17) -- (14,-16);
\draw (15,-16)  -- (14.5,-17) -- (15.5,-17) -- (15,-16);
\draw (16,-16)  -- (15.5,-17) -- (16.5,-17) -- (16,-16);
\draw (17,-16)  -- (16.5,-17) -- (17.5,-17) -- (17,-16);
\draw (18,-16)  -- (17.5,-17) -- (18.5,-17) -- (18,-16);
\draw (18.5,-17) -- (18,-16) -- (19,-16) -- (18.5,-17);
\draw (19,-16)  -- (18.5,-17) -- (19.5,-17) -- (19,-16);
\draw (19.5,-17) -- (19,-16) -- (20,-16) -- (19.5,-17);
\draw (20,-16)  -- (19.5,-17) -- (20.5,-17) -- (20,-16);
\draw (20.5,-17) -- (20,-16) -- (21,-16) -- (20.5,-17);
\draw (21,-16)  -- (20.5,-17) -- (21.5,-17) -- (21,-16);
\draw (21.5,-17) -- (21,-16) -- (22,-16) -- (21.5,-17);
\draw (9.5,-15)  -- (9,-16) -- (10,-16) -- (9.5,-15);
\draw (10,-16) -- (9.5,-15) -- (10.5,-15) -- (10,-16);
\draw (10.5,-15)  -- (10,-16) -- (11,-16) -- (10.5,-15);
\draw (11,-16) -- (10.5,-15) -- (11.5,-15) -- (11,-16);
\draw (11.5,-15)  -- (11,-16) -- (12,-16) -- (11.5,-15);
\draw (12.5,-15)  -- (12,-16) -- (13,-16) -- (12.5,-15);
\draw (13,-16) -- (12.5,-15) -- (13.5,-15) -- (13,-16);
\draw (13.5,-15)  -- (13,-16) -- (14,-16) -- (13.5,-15);
\draw (14,-16) -- (13.5,-15) -- (14.5,-15) -- (14,-16);
\draw (14.5,-15)  -- (14,-16) -- (15,-16) -- (14.5,-15);
\draw (15,-16) -- (14.5,-15) -- (15.5,-15) -- (15,-16);
\draw (15.5,-15)  -- (15,-16) -- (16,-16) -- (15.5,-15);
\draw (16,-16) -- (15.5,-15) -- (16.5,-15) -- (16,-16);
\draw (16.5,-15)  -- (16,-16) -- (17,-16) -- (16.5,-15);
\draw (17,-16) -- (16.5,-15) -- (17.5,-15) -- (17,-16);
\draw (17.5,-15)  -- (17,-16) -- (18,-16) -- (17.5,-15);
\draw (18.5,-15)  -- (18,-16) -- (19,-16) -- (18.5,-15);
\draw (19,-16) -- (18.5,-15) -- (19.5,-15) -- (19,-16);
\draw (19.5,-15)  -- (19,-16) -- (20,-16) -- (19.5,-15);
\draw (20,-16) -- (19.5,-15) -- (20.5,-15) -- (20,-16);
\draw (20.5,-15)  -- (20,-16) -- (21,-16) -- (20.5,-15);
\draw (21,-16) -- (20.5,-15) -- (21.5,-15) -- (21,-16);
\draw (9.5,-7)  -- (9,-8) -- (10,-8) -- (9.5,-7);
\draw (10,-8) -- (9.5,-7) -- (10.5,-7) -- (10,-8);
\draw (10.5,-7)  -- (10,-8) -- (11,-8) -- (10.5,-7);
\draw (11,-8) -- (10.5,-7) -- (11.5,-7) -- (11,-8);
\draw (11.5,-7)  -- (11,-8) -- (12,-8) -- (11.5,-7);
\draw (12.5,-7)  -- (12,-8) -- (13,-8) -- (12.5,-7);
\draw (13,-8) -- (12.5,-7) -- (13.5,-7) -- (13,-8);
\draw (13.5,-7)  -- (13,-8) -- (14,-8) -- (13.5,-7);
\draw (14,-8) -- (13.5,-7) -- (14.5,-7) -- (14,-8);
\draw (14.5,-7)  -- (14,-8) -- (15,-8) node (v4) {} -- (14.5,-7);
\draw (15,-8) -- (14.5,-7) -- (15.5,-7) -- (15,-8);
\draw (15.5,-7)  -- (15,-8) -- (16,-8) -- (15.5,-7);
\draw (16,-8) -- (15.5,-7) -- (16.5,-7) -- (16,-8);
\draw (16.5,-7)  -- (16,-8) -- (17,-8) -- (16.5,-7);
\draw (17,-8) -- (16.5,-7) -- (17.5,-7) -- (17,-8);
\draw (17.5,-7)  -- (17,-8) -- (18,-8) -- (17.5,-7);
\draw (18.5,-7)  -- (18,-8) -- (19,-8) -- (18.5,-7);
\draw (19,-8) -- (18.5,-7) -- (19.5,-7) -- (19,-8);
\draw (19.5,-7)  -- (19,-8) -- (20,-8) node (v3) {} -- (19.5,-7);
\draw (20,-8) -- (19.5,-7) -- (20.5,-7) -- (20,-8);
\draw (20.5,-7)  -- (20,-8) -- (21,-8) -- (20.5,-7);
\draw (21,-8) -- (20.5,-7) -- (21.5,-7) -- (21,-8);
\draw (10,-8)  -- (9.5,-9) -- (10.5,-9) -- (10,-8);
\draw (10.5,-9) -- (10,-8) -- (11,-8) -- (10.5,-9);
\draw (11,-8)  -- (10.5,-9) -- (11.5,-9) -- (11,-8);
\draw (12,-8)  -- (11.5,-9) -- (12.5,-9) -- (12,-8);
\draw (12.5,-9) -- (12,-8) -- (13,-8) -- (12.5,-9);
\draw (13,-8)  -- (12.5,-9) -- (13.5,-9) -- (13,-8);
\draw (13.5,-9) -- (13,-8) -- (14,-8) -- (13.5,-9);
\draw (14,-8)  -- (13.5,-9) -- (14.5,-9) -- (14,-8);
\draw (14.5,-9) -- (14,-8) -- (15,-8) -- (14.5,-9);
\draw (15,-8)  -- (14.5,-9) -- (15.5,-9) -- (15,-8);
\draw (15.5,-9) -- (15,-8) -- (16,-8) -- (15.5,-9);
\draw (16,-8)  -- (15.5,-9) -- (16.5,-9) -- (16,-8);
\draw (16.5,-9) -- (16,-8) -- (17,-8) -- (16.5,-9);
\draw (17,-8)  -- (16.5,-9) -- (17.5,-9) -- (17,-8);
\draw  (17.5,-9) -- (17,-8) -- (18,-8) -- (17.5,-9);
\draw (18,-8)  -- (17.5,-9) -- (18.5,-9) -- (18,-8);
\draw (19,-8)  -- (18.5,-9) -- (19.5,-9) -- (19,-8);
\draw (19.5,-9) -- (19,-8) -- (20,-8) -- (19.5,-9);
\draw (20,-8)  -- (19.5,-9) -- (20.5,-9) -- (20,-8);
\draw (20.5,-9) -- (20,-8) -- (21,-8) -- (20.5,-9);
\draw (21,-8)  -- (20.5,-9) -- (21.5,-9) -- (21,-8);
\draw (21.5,-9) -- (21,-8) -- (22,-8) -- (21.5,-9);
\draw (9.5,-9)  -- (9,-10) -- (10,-10) -- (9.5,-9);
\draw (10,-10) -- (9.5,-9) -- (10.5,-9) -- (10,-10);
\draw (10.5,-9)  -- (10,-10) -- (11,-10) -- (10.5,-9);
\draw (11.5,-9)  -- (11,-10) -- (12,-10) -- (11.5,-9);
\draw (12,-10) -- (11.5,-9) -- (12.5,-9) -- (12,-10);
\draw (12.5,-9)  -- (12,-10) -- (13,-10) -- (12.5,-9);
\draw (13,-10) -- (12.5,-9) -- (13.5,-9) -- (13,-10);
\draw (13.5,-9)  -- (13,-10) -- (14,-10) -- (13.5,-9);
\draw (14,-10) -- (13.5,-9) -- (14.5,-9) -- (14,-10);
\draw (14.5,-9)  -- (14,-10) -- (15,-10) -- (14.5,-9);
\draw (15,-10) -- (14.5,-9) -- (15.5,-9) -- (15,-10);
\draw (15.5,-9)  -- (15,-10) -- (16,-10) -- (15.5,-9);
\draw (16,-10) -- (15.5,-9) -- (16.5,-9) -- (16,-10);
\draw (16.5,-9)  -- (16,-10) -- (17,-10) -- (16.5,-9);
\draw (17,-10) -- (16.5,-9) -- (17.5,-9) -- (17,-10);
\draw (17.5,-9)  -- (17,-10) -- (18,-10) -- (17.5,-9);
\draw  (18,-10) -- (17.5,-9) -- (18.5,-9) -- (18,-10);
\draw (18.5,-9)  -- (18,-10) -- (19,-10) -- (18.5,-9);
\draw (19.5,-9)  -- (19,-10) -- (20,-10) -- (19.5,-9);
\draw (20,-10) -- (19.5,-9) -- (20.5,-9) -- (20,-10);
\draw (20.5,-9)  -- (20,-10) -- (21,-10) -- (20.5,-9);
\draw (21,-10) -- (20.5,-9) -- (21.5,-9) -- (21,-10);
\draw (10,-10)  -- (9.5,-11) -- (10.5,-11) -- (10,-10);
\draw (11,-10)  -- (10.5,-11) -- (11.5,-11) -- (11,-10);
\draw (11.5,-11) -- (11,-10) -- (12,-10) -- (11.5,-11);
\draw (12,-10)  -- (11.5,-11) -- (12.5,-11) -- (12,-10);
\draw (12.5,-11) -- (12,-10) -- (13,-10) -- (12.5,-11);
\draw (13,-10)  -- (12.5,-11) -- (13.5,-11) -- (13,-10);
\draw (13.5,-11) -- (13,-10) -- (14,-10) -- (13.5,-11);
\draw (14,-10)  -- (13.5,-11) -- (14.5,-11) -- (14,-10);
\draw (14.5,-11) -- (14,-10) -- (15,-10) -- (14.5,-11);
\draw (15,-10)  -- (14.5,-11) -- (15.5,-11) -- (15,-10);
\draw (15.5,-11) -- (15,-10) -- (16,-10) -- (15.5,-11);
\draw (16,-10)  -- (15.5,-11) -- (16.5,-11) -- (16,-10);
\draw (16.5,-11) -- (16,-10) -- (17,-10) -- (16.5,-11);
\draw (17,-10)  -- (16.5,-11) -- (17.5,-11) -- (17,-10);
\draw (17.5,-11) -- (17,-10) -- (18,-10) -- (17.5,-11);
\draw (18,-10)  -- (17.5,-11) -- (18.5,-11) -- (18,-10);
\draw  (18.5,-11) -- (18,-10) -- (19,-10) -- (18.5,-11);
\draw (19,-10)  -- (18.5,-11) -- (19.5,-11) -- (19,-10);
\draw (20,-10)  -- (19.5,-11) -- (20.5,-11) -- (20,-10);
\draw (20.5,-11) -- (20,-10) -- (21,-10) -- (20.5,-11);
\draw (21,-10)  -- (20.5,-11) -- (21.5,-11) -- (21,-10);
\draw (21.5,-11) -- (21,-10) -- (22,-10) -- (21.5,-11);
\draw (9.5,-11)  -- (9,-12) -- (10,-12) -- (9.5,-11);
\draw (10.5,-11)  -- (10,-12) -- (11,-12) -- (10.5,-11);
\draw (11,-12) -- (10.5,-11) -- (11.5,-11) -- (11,-12);
\draw (11.5,-11)  -- (11,-12) -- (12,-12) -- (11.5,-11);
\draw (12,-12) -- (11.5,-11) -- (12.5,-11) -- (12,-12);
\draw (12.5,-11)  -- (12,-12) -- (13,-12) -- (12.5,-11);
\draw (13,-12) -- (12.5,-11) -- (13.5,-11) -- (13,-12);
\draw (13.5,-11)  -- (13,-12) -- (14,-12) -- (13.5,-11);
\draw (14,-12) -- (13.5,-11) -- (14.5,-11) -- (14,-12);
\draw (14.5,-11)  -- (14,-12) -- (15,-12) -- (14.5,-11);
\draw (15,-12) -- (14.5,-11) -- (15.5,-11) -- (15,-12);

\draw (15.5,-11)  -- (15,-12) -- (16,-12) -- (15.5,-11);
\draw (16,-12) -- (15.5,-11) -- (16.5,-11) -- (16,-12);
\draw (16.5,-11)  -- (16,-12) -- (17,-12) -- (16.5,-11);
\draw (17,-12) -- (16.5,-11) -- (17.5,-11) -- (17,-12);
\draw (17.5,-11)  -- (17,-12) -- (18,-12) -- (17.5,-11);
\draw (18,-12) -- (17.5,-11) -- (18.5,-11) -- (18,-12);
\draw (18.5,-11)  -- (18,-12) -- (19,-12) -- (18.5,-11);
\draw  (19,-12) -- (18.5,-11) -- (19.5,-11) -- (19,-12);
\draw (19.5,-11)  -- (19,-12) -- (20,-12) -- (19.5,-11);
\draw (20.5,-11)  -- (20,-12) -- (21,-12) -- (20.5,-11);
\draw (21,-12) -- (20.5,-11) -- (21.5,-11) -- (21,-12);
\draw (10,-12)  -- (9.5,-13) -- (10.5,-13) -- (10,-12);
\draw (11,-12)  -- (10.5,-13) -- (11.5,-13) -- (11,-12);
\draw (11.5,-13) -- (11,-12) -- (12,-12) -- (11.5,-13);
\draw (12,-12)  -- (11.5,-13) -- (12.5,-13) node (v1) {} -- (12,-12);
\draw (12.5,-13) -- (12,-12) -- (13,-12) -- (12.5,-13);
\draw (13,-12)  -- (12.5,-13) -- (13.5,-13) -- (13,-12);
\draw (13.5,-13) -- (13,-12) -- (14,-12) -- (13.5,-13);
\draw (14,-12)  -- (13.5,-13) -- (14.5,-13) -- (14,-12);
\draw (14.5,-13) -- (14,-12) -- (15,-12) -- (14.5,-13);
\draw (15,-12)  -- (14.5,-13) -- (15.5,-13) -- (15,-12);
\draw (15.5,-13) -- (15,-12) -- (16,-12) -- (15.5,-13);
\draw (16,-12)  -- (15.5,-13) -- (16.5,-13) -- (16,-12);
\draw (16.5,-13) -- (16,-12) -- (17,-12) -- (16.5,-13);
\draw (17,-12)  -- (16.5,-13) -- (17.5,-13) node (v2) {} -- (17,-12);
\draw (17.5,-13) -- (17,-12) -- (18,-12) -- (17.5,-13);
\draw (18,-12)  -- (17.5,-13) -- (18.5,-13) -- (18,-12);
\draw  (18.5,-13) -- (18,-12) -- (19,-12) -- (18.5,-13);
\draw (19,-12)  -- (18.5,-13) -- (19.5,-13) -- (19,-12);
\draw (20,-12)  -- (19.5,-13) -- (20.5,-13) -- (20,-12);
\draw (20.5,-13) -- (20,-12) -- (21,-12) -- (20.5,-13);
\draw (21,-12)  -- (20.5,-13) -- (21.5,-13) -- (21,-12);
\draw (21.5,-13) -- (21,-12) -- (22,-12) -- (21.5,-13);
\draw (9.5,-13)  -- (9,-14) -- (10,-14) -- (9.5,-13);
\draw (10,-14) -- (9.5,-13) -- (10.5,-13) -- (10,-14);
\draw (10.5,-13)  -- (10,-14) -- (11,-14) -- (10.5,-13);
\draw (11.5,-13)  -- (11,-14) -- (12,-14) -- (11.5,-13);
\draw (12,-14) -- (11.5,-13) -- (12.5,-13) -- (12,-14);
\draw (12.5,-13)  -- (12,-14) -- (13,-14) -- (12.5,-13);
\draw (13,-14) -- (12.5,-13) -- (13.5,-13) -- (13,-14);
\draw (13.5,-13)  -- (13,-14) -- (14,-14) -- (13.5,-13);
\draw (14,-14) -- (13.5,-13) -- (14.5,-13) -- (14,-14);
\draw (14.5,-13)  -- (14,-14) -- (15,-14) -- (14.5,-13);
\draw (15,-14) -- (14.5,-13) -- (15.5,-13) -- (15,-14);
\draw (15.5,-13)  -- (15,-14) -- (16,-14) -- (15.5,-13);
\draw (16,-14) -- (15.5,-13) -- (16.5,-13) -- (16,-14);
\draw (16.5,-13)  -- (16,-14) -- (17,-14) -- (16.5,-13);
\draw (17,-14) -- (16.5,-13) -- (17.5,-13) -- (17,-14);
\draw (17.5,-13)  -- (17,-14) -- (18,-14) -- (17.5,-13);
\draw  (18,-14) -- (17.5,-13) -- (18.5,-13) -- (18,-14);
\draw (18.5,-13)  -- (18,-14) -- (19,-14) -- (18.5,-13);
\draw (19.5,-13)  -- (19,-14) -- (20,-14) -- (19.5,-13);
\draw (20,-14) -- (19.5,-13) -- (20.5,-13) -- (20,-14);
\draw (20.5,-13)  -- (20,-14) -- (21,-14) -- (20.5,-13);
\draw (21,-14) -- (20.5,-13) -- (21.5,-13) -- (21,-14);
\draw (10,-14)  -- (9.5,-15) -- (10.5,-15) -- (10,-14);
\draw (10.5,-15) -- (10,-14) -- (11,-14) -- (10.5,-15);
\draw (11,-14)  -- (10.5,-15) -- (11.5,-15) -- (11,-14);
\draw (12,-14)  -- (11.5,-15) -- (12.5,-15) -- (12,-14);
\draw (12.5,-15) -- (12,-14) -- (13,-14) -- (12.5,-15);
\draw (13,-14)  -- (12.5,-15) -- (13.5,-15) -- (13,-14);
\draw (13.5,-15) -- (13,-14) -- (14,-14) -- (13.5,-15);
\draw (14,-14)  -- (13.5,-15) -- (14.5,-15) -- (14,-14);
\draw (14.5,-15) -- (14,-14) -- (15,-14) -- (14.5,-15);
\draw (15,-14)  -- (14.5,-15) -- (15.5,-15) -- (15,-14);
\draw (15.5,-15) -- (15,-14) -- (16,-14) -- (15.5,-15);
\draw (16,-14)  -- (15.5,-15) -- (16.5,-15) -- (16,-14);
\draw (16.5,-15) -- (16,-14) -- (17,-14) -- (16.5,-15);
\draw (17,-14)  -- (16.5,-15) -- (17.5,-15) -- (17,-14);
\draw  (17.5,-15) -- (17,-14) -- (18,-14) -- (17.5,-15);
\draw (18,-14)  -- (17.5,-15) -- (18.5,-15) -- (18,-14);
\draw (19,-14)  -- (18.5,-15) -- (19.5,-15) -- (19,-14);
\draw (19.5,-15) -- (19,-14) -- (20,-14) -- (19.5,-15);
\draw (20,-14)  -- (19.5,-15) -- (20.5,-15) -- (20,-14);
\draw (20.5,-15) -- (20,-14) -- (21,-14) -- (20.5,-15);
\draw (21,-14)  -- (20.5,-15) -- (21.5,-15) -- (21,-14);
\draw (21.5,-15) -- (21,-14) -- (22,-14) -- (21.5,-15);
\draw [line width=3pt](12,-14) -- (17,-14) -- (19.5,-9) -- (14.5,-9) -- (12,-14);
\draw[line width=2pt] (16,-10) -- (15,-12) -- (17,-12) -- (16,-10);
\end{tikzpicture}
}\hfil
\resizebox{0.3\textwidth}{!}{
\begin{tikzpicture}[scale=0.8, every node/.style={scale=1.25}]
\draw [line width=3pt](-2,-2) rectangle (3,3);
\node at (0.5,0.5) {$u$};
\node at (0.5,1.5) {$p$};
\node at (1.5,0.5) {$$q};
\node at (0.5,-0.5) {$r$};
\node at (-0.5,0.5) {$s$};
\draw [line width=3pt,|-|] (-2,-2.5) -- (3,-2.5);
\node[fill=white,circle] at (0.5,-2.5) {$n$};
\draw (-2,3.5) -- (-2,-3.5);
\draw (-1,3.5) -- (-1,-3.5);
\draw (0,3.5) -- (0,-3.5);
\draw (1,3.5) -- (1,-3.5);
\draw (2,3.5) -- (2,-3.5);
\draw (3,3.5) -- (3,-3.5);
\draw (-3,3.5) -- (-3,-3.5);
\draw (4,3.5) -- (4,-3.5);
\draw (-3.5,3) -- (4.5,3);
\draw (4.5,2) -- (-3.5,2);
\draw (4.5,1) -- (-3.5,1);
\draw (4.5,0) -- (-3.5,0);
\draw (4.5,-1) -- (-3.5,-1);
\draw (4.5,-2) -- (-3.5,-2);
\draw (4.5,-3) -- (-3.5,-3);
\draw[line width=2pt] (0,1) -- (-1,1) -- (-1,0) -- (0,0) -- (0,-1) -- (1,-1) -- (1,0) -- (2,0) -- (2,1) -- (1,1) -- (1,2) -- (0,2) -- (0,1);
\end{tikzpicture}
}
  \caption{ Left: $G(7)$ in the one dimensional grid. Middle:  $G(5)$ in the triangular grid; Right: $G(5)$ in the square grid. All figures are represented with the von Neumann neighborhood of cell $u$.}     
  \label{fig: grids}
\end{figure} 
In this model, each cell can only have a finite number of states. We call  $Q$  the set of states. If only two states are considered, they will be denoted by $0$ and $1$. 
We say that a site in state $1$ is \emph{active} and a site in state $0$ is \emph{inactive}. 
A \emph{configuration} of the grid is a function $x$ that assigns values in $Q$ to a region in the $d$ dimensional lattice for $d = 1,2$. For instance if $d=2$,  in the case of the triangular grid, this region is given by a rhomboid shaped area of $2n^2$ cells and for the square grid it is given by a square area of $n\times n$ cells (see Figure \ref{fig: grids}). In any case, we call this region $G(n)$ and we refer to it as the $n$-grid. In addition, the value of the cell $u$ in the configuration $x$ is denoted $x_u$.

In any case, if $u \in \mathbb{Z}^d$ with $d = 1,2$, we will refer as $N(u)$ to the neighborhood of the cell $u$ and $N(\cdot)$ the neighborhood of the cell at the origin. Depending on the dimensions, a configuration $x$ is considered to be defined over a cycle graph (one-dimensional) or a torus (two-dimensional), by identifying each cell in the boundary of $x$ as a neighbor of the cells placed in the opposite boundary of $x$. In addition, for a cell $u\in \mathbb{Z}^d$, we call $x_{N(u)}$ the restriction of $x$ to the neighborhood of $u$. 

Formally, a \emph{cellular automaton} (CA) with states $Q$ and \emph{local function} $f:Q^{N(\cdot)} \to Q$, is a map $F:Q^{G(n)}\to Q^{G(n)}$, such that $F(x)_u = f(x_{N(u)})$. We call $F$ the \emph{global function} or the \emph{global rule} of the CA.  The dynamics is defined by assigning to the configuration $x$ a new state given by the synchronous update of the local function on $x$.  Given a partial order ($\leq$) in the set $Q$, we say that a cellular automaton $F$ is  \emph{freezing}
\cite{goles:hal-01294144} (FCA) if $F$ satisfies that $x_u\leq F(x)_u$ for every $u\in G(n)$ and $x \in Q^{G(n)}$. 
An important remark is that every configuration of a FCA reaches a fixed point in at most  $\cO(|G(n)|)$ steps. 

\subsection{Asynchronous updates and Unstability}

\emph{Asynchronous cellular automata} (ACA)  are defined similarly to CA. In this case, the global function $F:Q^{G(n)}\to Q^{G(n)}$ is defined by assigning to the configuration $x$ a new state given by the sequential update (cell by cell) of the local function on $x$. More precisely, in each time step, we evaluate the local function in a given cell while the others remain in the same state. We chose in each time step the cell that is going to be updated by the local rule according to some given order.
Formally, the $t$-th sequential update of $x$ is given by, 

$$
    F^{\sigma(0)}(x)=x;\qquad
    F^{\sigma(t)}(x)_z =
    \begin{cases}
         f(F^{\sigma(t-1)}(x)_{N(z)}) &\text{ if } z =\sigma(t)\\
         x_{z} &\text{ otherwise.}
       \end{cases}
$$
Where $\sigma:\mathbb{N}\to G(n)$, is a function such that for every $k \geq 1$,  $\left.\sigma \right|_{[(k-1)|G(n)|,k|G(n)|]}$ is a permutation of the $\{1,...|G(n)|\}$ where $|G(n)|$ is the number of cells inside the $n$-grid. In other words, this means that in at most $t=n$ time steps each cell of the grid is updated. We call $\sigma$ a \emph{sequential or asyncrhonous update scheme}. 

On the other hand, if we re-define $\sigma$ to change simultaneously more than one cell per time step, we will say that $\sigma$ is a \emph{block-sequential update scheme}.

We are interested in the cells that always remain in the same state for any sequential updating scheme. We call such cells \emph{stable} cells.

\begin{definition}
Given a configuration $x \in Q^{G(n)}$ and a FCA $F$, we say that a site $v$ is \emph{stable} for $x$ if and only if $v$ remains in its initial state after the application of the rule under any updating scheme, i.e., $F^{\sigma(t)}(x)_v= x_v$ for all $t\geq 0$ and any updating scheme $\sigma$.
\end{definition}

From the previous definition, we consider the problem {\AsyncStability}, which consists in deciding if a cell on a periodic configuration $x$ is stable. 
Formally speaking, if $F$ is a FCA, then:

  \medskip

  \defproblemu{Asynchronous Unstability ({\AsyncStability})}{A natural number $n$, a finite configuration $x$ of $G(n)$ and a site $u \in G(n)$. }{Does there exist a sequential updating scheme $\sigma$ and $T>0$ such that $F^{\sigma(T)}(x)_u  \neq x_u$? }

  \medskip

We call the input cell of problem {\AsyncStability} as the \emph{decision cell}. Observe that the question of the latter problem is whether the decision cell is stable or no. We say that a cell is \emph{iterated} if it changes its state.

Let us note that for every FCA, \AsyncStability is in \NP. Indeed, given an arbitrary FCA and the input of the latter problem,  there exists an updating scheme that iterates $u$ if and only if there exists a non-deterministic Turing Machine that can decide \AsyncStability. The latter can be deduced by considering the update scheme as a non-deterministic way of simulating the FCA. In other words, the update scheme plays the role of the certificate of the polynomial verifier simulating a non-deterministic Turing machine.

Our approach is to study the complexity of the family of FCA by the computational complexity of \AsyncStability. 
We say that a FCA is \emph{hard} if, roughly, the best strategy for solving \AsyncStability is the \emph{trivial one}, i.e. the strategy consisting in computing the future state of a cell  by trying every possible updating scheme until the decision cell changes its state.

\subsection{Life-Like Freezing Cellular Automata}

We will be particularly interested in the family of \emph{life-like freezing cellular automata} (LFCA). 
In this family, the active cells remain active, because the rule is freezing, and the inactive cells become active depending to an interval to which belongs the number of active neighbors. More precisely, there exist two non-negative integers $k_1, k_2$ such that $k_1\leq k_2$, we call $F_{k_1,k_2}$ the rule defined by:
$$F_{k_1,k_2}(x)_u = \left\{ \begin{array}{cl} 1 & \textrm{if } (x_u = 1) \vee (k_1 \leq \sum_{v\in N(u)\setminus\{u\} }x_{v} \leq k_2), \\ 0 & \textrm{otherwise.}  \end{array} \right.$$
in other words, in $F_{k_1, k_2}$ an inactive cell becomes active if the number of active cells in its  neighborhood is at least $k_1$ and at most $k_2$. 

Let $F$ be a LFCA. We can identify $F$ with a set $\mathcal{I}_F\subseteq \{0, 1, 2, 3\}$ for the triangular grid and $\mathcal{I}_F\subseteq \{0, 1, 2, 3,4\}$ for the square grid such that, for every configuration $x$ and site $u$: 
$$F(x)_u = \left\{ \begin{array}{cl} 1 & \textrm{if } (x_u = 1) \vee (\sum_{v\in N(u)\setminus\{u\} }x_{v}\in \mathcal{I}_F), \\ 0 & \textrm{otherwise.}  \end{array} \right.$$

We name the LFCA with the letter $T$ or $S$ (depending if the rule is in defined over the triangular or square grid) concatenated with the minimum and maximum elements contained in $\mathcal{I}_F$. For example, if $Maj$ is the freezing majority vote CA, where an inactive cell becomes active if the majority of its neighbors is active. Then in this notation $Maj$ is rule $T23$ for triangular grid and for square grid is the rule $S34$ or  $S24$ (depending if we consider strict or non-strict majority).

\subsection{Complexity classes}

We finish this section defining the main background concepts in computational complexity required in this article. For a more complete and formal presentation we refer to the books of Arora and Barak  \cite{arora2009computational} and  Greenlaw et al.  \cite{Greenlaw:1995}.  We assume that the reader is familiar with the basics concepts dealing with computational complexity.  As we mentioned in the introduction, in this paper we will  only consider complexity classes into which we classify the prediction problems.  {\Pt} is the class of problems solvable in a  Turing machine that runs in polynomial time in the size of the input. More formally, if $n$ is the size of the input, then a problem is polynomial time solvable if it can be solved in time $n^{\cO(1)}$ in a deterministic Turing machine. 

A logarithmic-space Turing machine consists in a Turing machine with three tapes: a read-only \emph{input tape}, a write-only \emph{output-tape} and a read-write \emph{work-tape}. The Turing machine is allowed to move as much as it likes on the input tape, but can only use $\cO(\log n)$ cells of the work-tape (where $n$ is the size of the input). Moreover, once the machine writes something in the output-tape, it moves to the next cell and can not return. 

A non-deterministic Turing machine is a Turing machine whose transition function does not necessarily output a single state but one over a set of possible states. A computation of the non-deterministic Turing machine considers all possible outcomes of the transition function. The machine is required to stop  in every possible computation thread, and we say that the machine \emph{accepts} if at least one thread finishes on an accepting state. A non-deterministic Turing machine is said to run in \emph{polynomial-time} if every computation thread stops in a number of steps that is polynomial in the size of the input.  {\NP} is the class of problems solvable in polynomial time in a non-deterministic Turing machine. A non-deterministic Turing machine is said to run in \emph{logarithmic-space} if every thread of the machine uses only logarithmic space in the work-tape. {\NL} is the class solvable in logarithmic space in a non-deterministic Turing machine. Both non-deterministic complexity classes can be characterized as the classes that can be efficiently \emph{verified}. More precisely, a verifier  of a problem $\mathcal{L}$  is a Turing machine that receives an input of $\mathcal{L}$ together with a polynomial-sized \emph{certificate}. The verified machine is required to accept at least one certificate on yes-instances, and reject all possible certificates on no-instances of $\mathcal{L}$. In that context, \NP is equivalent to the class of problems having a verifier that runs in polynomial time. Similarly, \NL can be characterized as the set of problems having verifier that is a logarithmic-space Turing machine, with the additional requirement that the certificate is written in a specific read-once tape. 

The other complexity class that we consider is {\bf NC}, which is the class of problems solvable in  \emph{polylogarithmic} time in a PRAM (a  parallel RAM) with a polynomial number of processors. An {\bf NC} algorithm is called a \emph{fast-parallel-algorithm}. Formally, a problem is in the class {\bf NC} if an instance of size $n$ can be solved in time $(\log(n))^{\cO(1)}$ using $n^{\cO(1)}$ processors on a  PRAM. In literature, there are several distinctions of PRAMs depending on how processors are allowed to access memory to read or write. Unless stated differently, we will consider the \emph{concurrent-read-exclusive-write} (CREW) PRAM, where two or more processors can read the same portion of the memory, and have reserved exclusive places in memory to write.  All these distinctions have no impact in the definition of class \NC. For more details we refer to \cite{Greenlaw:1995,JaJa:1992:IPA:133889}.

It is known that $\NL \subseteq \NC \subseteq \Pt \subseteq \NP$. The inclusions $\Pt \subseteq \NP$ and $\NC \subseteq \Pt$ are quite simple: any problem solvable in polynomial time on a deterministic Turing machine can be also solved in polynomial time on a non-deterministic Turing machine. Similarly, any problem solvable by a fast-parallel algorithm can be solved in polynomial time on a deterministic Turing machine, by simply sequentially simulating the computation of each processor.  The inclusion $\NL \subseteq \NC$ is a bit more technical, and is given by the fact that there is an \NL-Complete problem called \textsc{Reachability}, that can be solved by an \NC algorithm (see \cite{arora2009computational}  for more details).

The problems in $\Pt$ that are the most likely to not belong to ${\bf NC}$ are the \PtC problems.  A problem $\mathcal{L}$ is \PtC if it belongs to $\Pt$ and any other problem in $\Pt$ can be reduced to $\mathcal{L}$ via a logarithmic-space reduction. \PtC problems are the \emph{hardest} problems in \Pt, in the sense that If {\PtC} problem belongs to {\NC}, then all problems in \Pt would also belong to \NC, i.e. it would imply that  ${\Pt} ={\NC}$. Similarly, the problems in $\NP$ that are the most likely to not belong to $\Pt$ are the \NPC problems. A problem $\mathcal{L}$ is \NPC if every problem in \NP can be reduced to $\mathcal{L}$ by a polynomial-time reduction. \NPC problems are the \emph{hardest} problems in \NP, in the sense that If an \NPC complete belongs to \Pt, then $\NP=\Pt$. 

\section{The complexity of one dimensional asyncrhonous FCA}\label{sec: 1DF}

One-dimensional cellular automata are one the simplest class of cellular automata. 
In fact, properties such as reversibility or surjectivity are undecidable in dimension two or higher \cite{kari1994reversibility}, these properties or any other that can be expressed in  first order logic become decidable in one dimension \cite{amoroso1972decision,Sutner2009ModelCO}. 
Another advantage of one-dimensional cellular automata is that it is possible to represent the time as an extra dimension, obtaining a two-dimensional diagram called \emph{time-space diagram}. 

\begin{definition}
Let $F$ be a one-dimensional CA with states in $Q$. We define the \emph{space-time diagram of $F$} for a configuration $x$, $D_x \in Q^{\mathbb{Z}\times \mathbb{N}}$, as 
$D_x(z,t)=F^t(x)_z$. Similarly, if $F$ is a one-dimensional ACA with updating scheme $\sigma$, then the \emph{space-time diagram of $F$ and $\sigma$} is $D^{\sigma}_x \in Q^{\mathbb{Z}\times \mathbb{N}}$ such that $D^{\sigma}_x(z,t)=F^{\sigma(t)}(x)_z$.
\end{definition}

In other words, $D_x(\cdot,t)$ is a function that assigns a cell the $t$-th updating of $F$.\\

Observing that in FCA with $k$ states we can store a column of the time-space diagram in logarithmic space,
Goles et al \cite{goles:hal-01294144} show that for any freezing CA, any finite initial configuration $x$ and cell $u$, it is possible to compute $D_x(u,t)$ in \NC. 
Indeed, it is enough to store the initial value of the cell and the time-step in which this cell changes and its new value, obtaining at most a vector with $ k $ elements, containing $ k $ time-state pairs.  
Thus, if the initial configuration has $n$ elements, then the time-space diagram can by stored using $\cO(2k\log n)=\cO(\log n)$ bits in memory. In fact, in \cite{goles:hal-01294144} the authors show an algorithm running in non-deterministic logarithmic space. 

Without loss of generality, we consider only one-dimensional cellular automata  $F$ with radius one. Note that we can always consider that the radius is $1$ because in every other case we can simulate the previous one by considering groups or blocks of cells \cite{Delorme20113866,Delorme20113881}. We note that this simulation is still freezing by considering the lexicographic order induced by the original order in state set $Q$. Let $f$ be the local function of $F$. Let $C_{i}$ be the $i$-th column of the time-space diagram, where $i$ is counted modulo $n-1$ ($n$ is the size of the input configuration). Let $C_{i}[t]$ be the state of the cell $i$ at time $t$ in the time-space diagram. 
The  algorithm given in \cite{goles:hal-01294144} to solve the prediction problem of $F$  is given in Algorithm \ref{alg: 1D FCA}.

  \begin{algorithm}
  \caption{\NL algorithm for 1D FCA}\label{alg: 1D FCA}
  \begin{algorithmic}[1]
  \REQUIRE $x\in \{0,1\}^{[0,n]}$  a site $u \in [0,n]$, a time $T$, and state $q\in Q$    
  \STATE  guess columns $C_{n-1}$,  $C_{0}$ and $C_u$. 
  \IF{ $C_{0}[0] \neq x_{0} \vee C_{n-1}[0] \neq x_{n-1} \vee C_{u}[0] \neq x_{u}$ }
  \RETURN \emph{Reject}
  \ENDIF
  \FOR{ $i = 0,...,n-1$}
    \STATE  guess column $C_{i+1}$ and store it in memory.
    \FOR{ $t = 1,...,T-1$}
      \IF{ $f(C_{i-1}[t],\ C_{i}[t],\ C_{i+1}[t])\not=C_{i}[t+1]$}
        \RETURN \emph{Reject}
      \ENDIF
    \ENDFOR
     \IF{ $(i-1) \neq 0$, $(i-1) \neq n-1$ and  $(i-1) \neq u$}
    \STATE remove from the memory column $C_{i-1}$
    \ENDIF
  \ENDFOR
  \IF{ $q=C_{u}[T]$}
    \RETURN \emph{Accept} 
  \ENDIF  
  \RETURN \emph{Reject}
  \end{algorithmic}
  \end{algorithm}

Algorithm \ref{alg: 1D FCA} can be summarized as follows: 

\begin{enumerate}
\item The algorithm starts (nondeterministically) guessing columns $C_0$, $C_{n-1}$ and $C_u$.  The algorithm checks whether these columns are compatible with $x$ (i.e., the starting point of each column corresponds to the value of the cell in $x$). If it is not the case, the algorithm \emph{rejects}. Otherwise, these columns are going to be stored in memory during all the execution of the algorithm. 

\item For each cell $i$, starting from cell $0$ to cell $n-1$, the algorithm guesses the value of column $C_{i+1}$ (at this point we assume that the columns $C_{i-1}$ and $C_{i}$ are already in memory) and checks for each time-step whether the columns are \emph{valid}, i.e. the transitions of the FCA are coherent  with the local values given in the columns. If it is not the case, the algorithm \emph{rejects}.  Then, the algorithm removes from memory column $C_{i-1}$ (unless this column is the one that corresponds to cell $0$, $n$ or $u$ ) and continues with the next cell.
\item Once all columns were verified, if at this point the algorithm has not rejected, the algorithm accepts if cell $u$ reach state $q$ after $T$ time-steps by checking $C_u[T]$. 
\end{enumerate}

Observe that Algorithm \ref{alg: 1D FCA}  runs in (nondeterministic) logarithmic space, since at any time-step at most six columns are stored, and each column can be encoded in $\cO(\log n)$ space. 

We show that Algorithm \ref{alg: 1D FCA}  can be adapted for the  \emph{asynchronous} case, In order to achieve this, we need to show that it is possible to verify that the given time-space diagram is \emph{valid} in this context. We say that an element of $D$ of $Q^{\mathbb{Z}\times\mathbb{N}}$ is \emph{valid} for $F$, $x$ and if there exists a sequential update scheme $\sigma$ such that $D = D^{\sigma}_x$ In this regard, we show in the following lemma that it is not necessary to globally verify the asynchronicity of the update scheme but only checking a local condition.

\begin{lemma}
\label{LMA:Desynch} 
Let $\sigma$ be an update scheme (not necessarily sequential) that satisfies the following condition: for each $t \in \mathbb{N}$ and for each cell $i \in G(n)$ if $i \in \sigma(t)$ then $\forall j \in N(i), \text{ } j \not \in \sigma(t)$. Then, there is a sequential update scheme $\tilde{\sigma}$ that reaches the same fixed point as $\sigma$.
\end{lemma}

\begin{proof}
 Let $\sigma$ be an update scheme satisfying the conditions of the lemma. Let $i,j$ two different cells that do not have any common neighbor.  Let us suppose that there exist $t \in \mathbb{N}$ such that $\{i,j\} \subseteq \sigma(t)$. In other words, $i$ and $j$ are iterated simultaneously at time $t$. Then, as they do not share any neighbor, the configuration $F^{\sigma(t)}(x)$ can be obtained at $t+1$ by other update scheme $\tilde{\sigma}$ such that:
 \begin{itemize}
 \item $\tilde{\sigma}(s) = \sigma(s)$ for every $s < t$, 
  \item $\tilde{\sigma}(t) = \{i\}$ and $\tilde{\sigma}(t+1) = \sigma(t) \setminus \{i\}$,
  \item  $\tilde{\sigma}(s+1) = \sigma(s)$ for every $s>t$.
  \end{itemize}
  In other words, $\tilde{\sigma}$ is the update scheme that iterates $i$ at time $t$ and leaves the other cells in its current state, then at time $t+1$ it iterates as same as $\sigma$ in time $t$, and in other time-step it follows the update scheme given by $\sigma$. Finally, by repeating the previous argument for each pair $i,j$ that are iterated simultaneously,  we obtain an asynchronous update scheme $\tilde{\sigma}$ that satisfies the desired result.
\end{proof}

Lemma \ref{LMA:Desynch} implies that it does not matter if two cells that are not in the same neighborhood are iterated at the same time, because it is always possible to ``desynchronize'' these iterations. This latter observation reduce the problem of validating a given space-time diagram to check if two cells that share the same neighborhood are iterated at the same time. We deduce that with slight modifications Algorithm \ref{alg: 1D FCA} can be used to solve  \AsyncStability.

\begin{theorem}\label{thm: 1D FTACA}
For all one-dimensional freezing cellular automata  \AsyncStability is in \NL .
\end{theorem}

\begin{proof}

The proof of this theorem is a modification of Algorithm \ref{alg: 1D FCA} where, at the moment of validating three columns, the algorithm also checks if two adjacent vertices are updated simultaneously. More precisely, in  {\bf Step 8}   instead of checking that $$f(C_{i-1}[t],\ C_{i}[t],\ C_{i+1}[t])\not=C_{i}[t+1]$$ the algorithm  checks whether  $$\left( f(C_{i-1}[t],\ C_{i}[t],\ C_{i+1}[t])\not=C_{i}[t+1] \right) \wedge \left( C_{i}[t]\not=C_{i}[t+1] \right).$$

Note that for any input \emph{accepted}  by the algorithm there exists a time-space diagram with an underlying updating scheme that can be ``desynchronized''  according to Lemma \ref{LMA:Desynch}. Therefore, from any accepted input there exists a sequential (asynchronous) updating scheme that iterates $u$.
  Reciprocally, if there is no asynchronous updating scheme that iterates  $u$, then any time-space diagram contains two adjacent cells that are updated at the same time (it is not sequential) or it is not valid.
\end{proof}

\section{Two-dimensional asynchronous FCA: Complexity Upper-bounds}\label{sec:upperbounds}

Now that we have shown that {\AsyncStability} for all one-dimensional FCA are in \NL, we extend the study of this problem to the two-dimensional case.  We study the complexity of \AsyncStability for the LFCA on the triangular grid and square grids. As a result, we obtain  that  the problem is in \NC for all these rules in the triangular grid, and for the majority of them in the square grid. 

Recall that a LFCA rule $F$ can be represented by a set $\mathcal{I}_F$, which is a subset of $\{1,2,3\}$ for rules defined in the triangular grid, and a subset of $\{1,2,3,4\}$ for rules defined over the square grid. 

In order to show our results, we group the set of FTCA on the triangular grid according the algorithm used to study its complexity: 
\begin{itemize}

\item Trivial rules:  The rules on which the stability can be solved in constant time on a sequential machine. The list of rules in this group are: $T00$, $S00$, $T33$, $S44$,  $T03$, $S04$, $T13$, $S14$. 
\item Infiltration rules:  Are rules for which $1$ belongs to $\mathcal{I}_F$ (except trival rules). The list of rules in this group are $T01$, $T02$, $T11$, $T12$, $S01$, $S02$, $S03$, $S11$, $S12$, $S13$.
\item Monotone-like rules: Are the rules that are not trivial or infiltration, and contain the value $|N(\cdot)|-1$, which equals $2$ in the triangular grid and $3$ in the square grid. These rules are monotone, or can be related with a monotone rule. The list of rules is $T22$, $T23$, $S23$, $S24$, $S33$, $S34$. 
\item Hard rules: this group contains only rule $S22$, which is going to be studied in the next section. 
\end{itemize}

In Table \ref{table:rulegroups} we summarize this classification. 

\begin{table}[h]
\centering
\begin{tabular}{|c|c|l|c|c|}
\cline{1-2} \cline{4-5}
\cellcolor[HTML]{C0C0C0}Rule & \cellcolor[HTML]{C0C0C0}Group &  & \cellcolor[HTML]{C0C0C0}Rule & \cellcolor[HTML]{C0C0C0}Group \\ \cline{1-2} \cline{4-5} 
$T00$                        & Trivial                       &  & $T12$                        & Infiltration                  \\ \cline{1-2} \cline{4-5} 
$T01$                        & Infiltration                  &  & $T13$                        & Trivial                       \\ \cline{1-2} \cline{4-5} 
$T02$                        & Infiltration                  &  & $T22$                        & Monotone-like                 \\ \cline{1-2} \cline{4-5} 
$T03$                        & Trivial                       &  & $T23$                        & Monotone-like                 \\ \cline{1-2} \cline{4-5} 
$T11$                        & Infiltration                  &  & $T33$                        & Trivial                       \\ \cline{1-2} \cline{4-5} 
\end{tabular}

\bigskip

\begin{tabular}{|c|c|l|c|c|l|c|c|}
\cline{1-2} \cline{4-5} \cline{7-8}
\cellcolor[HTML]{C0C0C0}Rule & \cellcolor[HTML]{C0C0C0}Group &  & \cellcolor[HTML]{C0C0C0}Rule & \cellcolor[HTML]{C0C0C0}Group &  & \cellcolor[HTML]{C0C0C0}Rule & \cellcolor[HTML]{C0C0C0}Group \\ \cline{1-2} \cline{4-5} \cline{7-8} 
$S00$                        & Trivial                       &  & $S11$                        & Infiltration                  &  & $S23$                        & Monotone-like                 \\ \cline{1-2} \cline{4-5} \cline{7-8} 
$S01$                        & Infiltration                  &  & $S12$                        & Infiltration                  &  & $S24$                        & Monotone-like                 \\ \cline{1-2} \cline{4-5} \cline{7-8} 
$S02$                        & Infiltration                  &  & $S13$                        & Infiltration                  &  & $S33$                        & Monotone-like                 \\ \cline{1-2} \cline{4-5} \cline{7-8} 
$S03$                        & Infiltration                  &  & $S14$                        & Trivial                       &  & $S34$                        & Monotone-like                 \\ \cline{1-2} \cline{4-5} \cline{7-8} 
$S04$                        & Trivial                       &  & $S22$                        & {\bf Hard}                   &  & $S44$                        & Trivial                       \\ \cline{1-2} \cline{4-5} \cline{7-8} 
\end{tabular}
\caption{Classification of two-dimensional LFCA in four groups: Trivial, Infiltration, Monotone-like and Hard. In this section we show that for all rules in groups different than Hard the problem \AsyncStability  can be solved in \NC.  Rule $S22$ is the only Hard rule, and will be studied in the next section. } 
\label{table:rulegroups}
\end{table}
\subsection{Trivial rules}

We begin explaining the algorithms for Trivial rules: on rules $T00$ and $S00$ (respectively $T33$ and $S44$) an algorithm only has to check the neighborhood of the decision cell. In fact, if every neighbor is initially inactive (respectively active) then the cell will become active in one time-step. Otherwise, it is impossible that a neighbor cell that is not initially inactive (respectively active) becomes active in a future time-step. Thus, simulating only one time-step of the decision cell is enough to decide  \AsyncStability  for rules $T00$, $S00$, $T33$ or $S44$. For rules $T03$ $S04$ all input of \AsyncStability is a Yes-instance. Finally, for rules $T13$ and $S14$ it is enough to have at least one active cell in the initial configuration in order to eventually activate any other cell (in particular the decision cell). Thus, an algorithm for \AsyncStability for rule $T13$ or $S14$ only checks whether there exists an active cell in the initial configuration, taking $\cO(\log n)$ space.

The arguments used to design \NC algorithms deciding \AsyncStability for the other rules are not that simple. In the next subsections, we elaborate on this matter. 

 \subsection{Infiltration rules}\label{subsec: rules with 1}
  In this subsection we study the rules where an infiltration approach can be implemented to efficiently solve \AsyncStability. Observe that these rules are characterized by the fact that inactive cells become active when the cardinality of the set of active neighbors is $1$.
  
  We start by defining, for an infiltration rule $F$, an initial configuration $x$, the set $V_{+}$  given by: 
  $$V_{+}=\left\{v \in G(n): \left( x_v=0 \right) \ \wedge \ \left( \sum_{w\in N(v)}x_w  \notin \mathcal{I}_F \right) \ \wedge \ \left( \sum_{w\in N(v)}x_w +1 \in \mathcal{I}_F \right)\right\}, $$
in other words, $V_{+}$ is the set containing all the cells in $G(n)$ that, with respect to the initial configuration,  \emph{become active with exactly one more active neighbor} according to rule $F$.

\begin{figure}[H]
  \centering
  \resizebox{0.3\textwidth}{!}{
\begin{tikzpicture}[scale=1, every node/.style={scale=1.25}]
\clip (10,7.5) rectangle (3.5,1.5);
\draw [pattern=north east lines](5.0,3.46410161514) -- (5.5,2.59807621135) -- (6,3.5) -- (6.5,2.59807621135) -- (7.5,4.33012701892) -- (8.0,3.46410161514) -- (8.5,4.33012701892) -- (8,5.19615242271)  -- (8,5.1962) -- (7.5,6.06217782649) -- (7,5.19615242271) -- (6.0,5.19615242271) -- (5.0,3.46410161514);
\draw [pattern=checkerboard,pattern color=black!30] (5.0,5.19615242271)
 -- (5.5,4.33012701892)
 -- (6,5.19615242271);
\draw [pattern=checkerboard,pattern color=black!30](6,5.19615242271)
 -- (6.5,6.06217782649)
 -- (7.5,6.06217782649)
 -- (7,5.19615242271)
; 
 \draw [pattern=checkerboard,pattern color=black!30](8.5,6.05)
 -- (7.5,6.05)
 -- (8,5.134);
 \draw [pattern=checkerboard,pattern color=black!30](9,5.216)
 -- (8,5.216)
 -- (8.5,4.35);
 \draw [pattern=checkerboard,pattern color=black!30](5.5,4.33012701892)
 -- (4.5,4.33012701892)
 -- (5.0,3.46410161514);
\draw [pattern=checkerboard,pattern color=black!30](6,3.466)
 -- (5.5,2.6)
 -- (6.5,2.6);
\draw [pattern=checkerboard,pattern color=black!30](7,3.466)
 -- (6.5,2.6)
 -- (7.5,2.6);
 \draw [pattern=checkerboard,pattern color=black!30](7.5,4.316)
 -- (7,3.45)
 -- (8,3.45);
 \draw [pattern=checkerboard,pattern color=black!30](8.5,4.316)
 -- (8,3.45)
 -- (9,3.45);
 \draw [pattern=checkerboard,pattern color=black!30](5.0,3.46410161514)
 -- (4.5,2.59807621135)
 -- (5.5,2.59807621135);
;
\draw [pattern=checkerboard,pattern color=black!30] (11,7) rectangle (11,7);
\fill[black] (6.5,6.06217782649) -- (7.0,6.92820323028) -- (7.5,6.06217782649) -- (6.5,6.06217782649);

\fill[black] (7.5,6.06217782649) -- (8.0,6.92820323028) -- (8.5,6.06217782649) node (v8) {} -- (7.5,6.06217782649) node (v9) {};
\fill[black] (6.0,5.19615242271) -- (5.5,6.06217782649) -- (6.5,6.06217782649) -- (6.0,5.19615242271) node (v10) {};
\fill[black] (4.5,4.33012701892) -- (5.0,5.19615242271) -- (5.5,4.33012701892) -- (4.5,4.33012701892);
\fill[black] (8.5,4.33012701892) -- (9.0,5.19615242271) node (v6) {} -- (9.5,4.33012701892) -- (8.5,4.33012701892);
\fill[black] (4.0,3.46410161514) -- (4.5,4.33012701892) -- (5.0,3.46410161514) node (v11) {} -- (4.0,3.46410161514);
\fill[black] (9.0,3.46410161514) -- (8.5,4.33012701892) node (v5) {} -- (9.5,4.33012701892) -- (9.0,3.46410161514);
\fill[black] (7.5,2.59807621135) -- (7.0,3.46410161514) node (v3) {} -- (8.0,3.46410161514) node (v4) {} -- (7.5,2.59807621135);
\fill[black] (5.0,1.73205080757) -- (4.5,2.59807621135) -- (5.5,2.59807621135) -- (5.0,1.73205080757);
\fill[black] (6.0,1.73205080757) -- (5.5,2.59807621135) node (v1) {} -- (6.5,2.59807621135) node (v2) {} -- (6.0,1.73205080757);

\draw[color=gray, thick] (1.5,0.866025403784) -- (1.5,0.866025403784);
\draw[color=gray, thick] (2.5,0.866025403784) -- (2.0,1.73205080757);
\draw[color=gray, thick] (3.5,0.866025403784) -- (2.5,2.59807621135);
\draw[color=gray, thick] (4.5,0.866025403784) -- (3.0,3.46410161514);
\draw[color=gray, thick] (5.5,0.866025403784) -- (3.5,4.33012701892);
\draw[color=gray, thick] (6.5,0.866025403784) -- (4.0,5.19615242271);
\draw[color=gray, thick] (7.5,0.866025403784) -- (4.5,6.06217782649);
\draw[color=gray, thick] (8.5,0.866025403784) -- (5.0,6.92820323028);
\draw[color=gray, thick] (9.0,1.73205080757) -- (5.5,7.79422863406);
\draw[color=gray, thick] (9.5,2.59807621135) -- (6.5,7.79422863406);
\draw[color=gray, thick] (10.0,3.46410161514) -- (7.5,7.79422863406);
\draw[color=gray, thick] (10.5,4.33012701892) -- (8.5,7.79422863406);
\draw[color=gray, thick] (11.0,5.19615242271) -- (9.5,7.79422863406);
\draw[color=gray, thick] (11.5,6.06217782649) -- (10.5,7.79422863406);
\draw[color=gray, thick] (12.0,6.92820323028) -- (11.5,7.79422863406);
\draw[color=gray, thick] (12.5,7.79422863406) -- (12.5,7.79422863406);

\draw[color=gray, thick] (-0.5,0.866025403784) -- (12.5,0.866025403784);
\draw[color=gray, thick] (2.22044604925e-16,1.73205080757) -- (12.0,1.73205080757);
\draw[color=gray, thick] (-0.5,2.59807621135) -- (12.5,2.59807621135);
\draw[color=gray, thick] (4.4408920985e-16,3.46410161514) -- (12.0,3.46410161514);
\draw[color=gray, thick] (-0.5,4.33012701892) -- (12.5,4.33012701892);
\draw[color=gray, thick] (8.881784197e-16,5.19615242271) -- (12.0,5.19615242271);
\draw[color=gray, thick] (-0.5,6.06217782649) -- (12.5,6.06217782649);
\draw[color=gray, thick] (8.881784197e-16,6.92820323028) -- (12.0,6.92820323028);
\draw[color=gray, thick] (-0.5,0.866025403784) -- (3.5,7.79422863406);
\draw[color=gray, thick] (2.22044604925e-16,1.73205080757) -- (3.5,7.79422863406);
\draw[color=gray, thick] (-0.5,2.59807621135) -- (2.5,7.79422863406);
\draw[color=gray, thick] (4.4408920985e-16,3.46410161514) -- (2.5,7.79422863406);
\draw[color=gray, thick] (-0.5,4.33012701892) -- (1.5,7.79422863406);
\draw[color=gray, thick] (8.881784197e-16,5.19615242271) -- (1.5,7.79422863406);
\draw[color=gray, thick] (-0.5,6.06217782649) -- (0.5,7.79422863406);
\draw[color=gray, thick] (8.881784197e-16,6.92820323028) -- (0.5,7.79422863406);
\draw[color=gray, thick] (12.5,0.866025403784) -- (12.5,0.866025403784);
\draw[color=gray, thick] (11.5,0.866025403784) -- (12.0,1.73205080757);
\draw[color=gray, thick] (11.5,0.866025403784) -- (12.5,2.59807621135);
\draw[color=gray, thick] (10.5,0.866025403784) -- (12.0,3.46410161514);
\draw[color=gray, thick] (10.5,0.866025403784) -- (12.5,4.33012701892);
\draw[color=gray, thick] (9.5,0.866025403784) -- (12.0,5.19615242271);
\draw[color=gray, thick] (9.5,0.866025403784) -- (12.5,6.06217782649);
\draw[color=gray, thick] (8.5,0.866025403784) -- (12.0,6.92820323028);
\draw[color=gray, thick] (0.5,0.866025403784) -- (4.5,7.79422863406);
\draw[color=gray, thick] (1.5,0.866025403784) -- (5.5,7.79422863406);
\draw[color=gray, thick] (2.5,0.866025403784) -- (6.5,7.79422863406);
\draw[color=gray, thick] (3.5,0.866025403784) -- (7.5,7.79422863406);
\draw[color=gray, thick] (4.5,0.866025403784) -- (8.5,7.79422863406);
\draw[color=gray, thick] (5.5,0.866025403784) -- (9.5,7.79422863406);
\draw[color=gray, thick] (6.5,0.866025403784) -- (10.5,7.79422863406);
\draw[color=gray, thick] (7.5,0.866025403784) -- (11.5,7.79422863406);
\draw[color=gray, thick] (1.5,0.866025403784) -- (4.4408920985e-16,3.46410161514);
\draw[color=gray, thick] (2.0,1.73205080757) -- (8.881784197e-16,5.19615242271);
\draw[color=gray, thick] (2.5,2.59807621135) -- (8.881784197e-16,6.92820323028);
\draw[color=gray, thick] (3.0,3.46410161514) -- (0.5,7.79422863406);
\draw[color=gray, thick] (3.5,4.33012701892) -- (1.5,7.79422863406);
\draw[color=gray, thick] (4.0,5.19615242271) -- (2.5,7.79422863406);
\draw[color=gray, thick] (4.5,6.06217782649) -- (3.5,7.79422863406);
\draw[color=gray, thick] (5.0,6.92820323028) -- (4.5,7.79422863406);
\draw[color=gray, thick] (0.5,0.866025403784) -- (2.22044604925e-16,1.73205080757);
\draw[color=gray, thick] (8.5,0.866025403784) -- (8.5,0.866025403784);
\draw[color=gray, thick] (9.0,1.73205080757) -- (9.5,0.866025403784);
\draw[color=gray, thick] (9.5,2.59807621135) -- (10.5,0.866025403784);
\draw[color=gray, thick] (10.0,3.46410161514) -- (11.5,0.866025403784);
\draw[color=gray, thick] (10.5,4.33012701892) -- (12.5,0.866025403784);
\draw[color=gray, thick] (11.0,5.19615242271) -- (12.5,2.59807621135);
\draw[color=gray, thick] (11.5,6.06217782649) -- (12.5,4.33012701892);
\draw[color=gray, thick] (12.0,6.92820323028) -- (12.5,6.06217782649);
\node at (7,4.5) {};
\node[circle,fill=white,inner sep=0pt,minimum size=3pt] at (5,3) {{\bf b}};
\node[circle,fill=white,inner sep=0pt,minimum size=3pt]  at (5,4) {{\bf a}};
\end{tikzpicture} 
  } 
  \hfil  
  \resizebox{0.3\textwidth}{!}{
\begin{tikzpicture}[scale=1.0,every node/.style={scale=1.25}]
\fill (-2.5,2) rectangle (-3,2.5);
\fill (-2.0,1) rectangle (-2.5,1.5);
\fill (-2,0)  rectangle (-2.5,0.5) ;
\fill (-1,-0.5) rectangle (-1.5,-0);
\fill (-0,-1) rectangle (-0.5,-0.5);
\fill (1,0)   rectangle (0.5,0.5)  ;
\fill (1.5,0.5)   rectangle (1,1)  ; 
\fill (1,2)   rectangle (0.5,2.5)  ;
\fill (0.5,3.5)   rectangle (0,4)  ;
\fill (0,3.5)  rectangle (-0.5,4) ;
\fill (-1.5,3)  rectangle (-2,3.5) ;
\fill (-1,3)  rectangle (-1.5,3.5) ;
\fill (-2.5,0.5)  rectangle (-3,1) ;
\draw [pattern=checkerboard,pattern color=black!30](-2,2.5) -- (-2.5,2.5)-- (-2.5,1.5) -- (-2,1.5) node (v1) {}-- (-2,1)-- (-2.5,1) -- (-2.5,0.5) -- (-2,0.5) -- (-2,0) -- (-1,0) -- (-1,-0.5) -- (0,-0.5) -- (0,0) -- (0.5,0) -- (0.5,0.5) -- (1,1) -- (1.5,1) -- (1.5,1.5) -- (1,1.5) -- (1,2) -- (0.5,2) -- (0.5,2.5) -- (1,2.5) -- (1,3) -- (0.5,3) -- (0.5,3.5) -- (-1,3.5) -- (-1,3) -- (-2,3) -- (-2,2.5);
\draw [preaction={fill, white}, pattern=north east lines](-1.5,2.5) -- (-2,2.5) -- (-2,0.5) -- (-1,0.5) -- (-1,0) -- (0,0) -- (0,0.5) -- (1,0.5) -- (1,1.5) -- (0.5,1.5) --(0.5,2) --  (0,2) -- (0,2.5) -- (0.5,2.5) -- (0.5,3) -- (-1,3) -- (-1,2.5) -- (-1.5,2.5);
\fill [pattern=checkerboard,pattern color=black!30,white] (v1) rectangle (-1.5,1);
\draw [pattern=checkerboard,pattern color=black!30] (v1) rectangle (-1.5,1);
\draw [help lines, step=0.5cm] (-3.75,-1.75) grid (2.75,4.75);
\node[fill=white,circle,inner sep=0pt,style={scale=0.99}] at (-2.25,1.75) {b};
\node[fill=white,circle,inner sep=0pt,,style={scale=0.99}] at (-2.25,0.75) {a};
\end{tikzpicture}
  }
  \caption[]{Example of  $V_{+1}$ and  $B_{+1}$ for rule $T11$ and $S11$.
  {\protect
    \resizebox{!}{\myMheight}{
    \begin{tikzpicture}
    \draw [preaction={fill, white}, pattern=north east lines] (0,0) rectangle (0.25,0.25);
    \end{tikzpicture}
    }
  }: Cells in $V_{+1}$.
  {\protect
    \resizebox{!}{\myMheight}{
    \begin{tikzpicture}
    \draw [pattern=checkerboard,pattern color=black!30] (0,0) rectangle (0.5,0.5);
    \end{tikzpicture}
    }
  }: Cells in $B_{+1}$.
  {\protect
    \resizebox{!}{\myMheight}{
    \begin{tikzpicture}
    \draw [fill] (0,0) rectangle (0.25,0.25);
    \end{tikzpicture}
    }
  }: Active cells .
  The cell (a) is not in $V_{+1}$ because it has more than one active neighbor.
  The cell (b) is not in $V_{+1}$ because it will become active in one time-step. 
}
  \label{fig: asinc path}
\end{figure}

Let $v$ be an inactive cell. We say that a cell $v $ \emph{infiltrates} if  $\sum_{w \in N(v)} x_w \in \mathcal{I}_F$, i.e. when updating $v$ first it becomes active.  
 
\begin{lemma}\label{lem:infil}
Let $v$ be a cell that does not belong to $V_+$. Then $v$ is either stable or it infiltrates. 
\end{lemma}

\begin{proof}
Let $v$ be a cell that does not belong to $V+$, and suppose first that all neighbors of $v$ are active. Since active cells never become inactive, then the neighborhood of $v$ would not change over any updating scheme. As $4$ does not belong to $\mathcal{I}$ for all infiltration rules, we deduce that $v$ is either stable.

Suppose now that $v$ has zero active neighbors. As we are assuming that $v$ does not belong to $V_+$, this must imply that $0$ is contained in $\mathcal{I}_F$ (i.e. $F$ is rule $T01$, $T02$, $S01$, $S02$ or $S03$). Therefore, $v$ infiltrates. 

Suppose now that that $v$ has at last one inactive  neighbor and at least one active neighbor. If $v$ has only one active neighbor, we have that $v$ infiltrates (by definition of infiltration rule). We will analyze different cases regarding the number of active neighbors of $v$ and the infiltration rule $F$. 

Suppose first that cell $v$ has exactly two active neighbors. 
\begin{itemize}
\item If $\mathcal{I}_F$ does not contain $2$ (i.e. $F$ is $T01$, $T11$, $S01$ or $S11$), then $v$ is stable because it can only become active with exactly one neighbor, it already has two active neighbors and active cells never become inactive. 
\item If $\mathcal{I}_F$ contains $2$ (i.e. $F$ is $T02$, $T12$, $S02$, $S03$, $S12$ or $S13$) we have that $v$ infiltrates.  
\end{itemize}
Suppose now  that cell $v$ has exactly three active neighbors. 
\begin{itemize}
\item If $\mathcal{I}_F$ does not contain $3$ (i.e. $F$ is $T01$, $T02$, $T11$, $T12$, $S01$, $S02$, $S11$ or $S12$) then $v$ is stable because it can only become active with exactly one neighbor, it already has two active neighbors and active cells never become inactive. 
\item If $\mathcal{I}_F$ contains $3$ (i.e. $F$ is $S03$ or $S13$) we have that $v$ infiltrates.  
\item $F$ can not be rule $S124$ or $S134$, because in that case $v$ would belong to $V_+$.
\end{itemize}

In any case, we deduce that $v$ is either stable or it infiltrates.  
\end{proof} 

Observe that Lemma \ref{lem:infil} allows to efficiently solve \AsyncStability when the decision cell does not belong to $V_+$. Indeed, an algorithm simply consists in computing in parallel $F(x)$ (the parallel update of $F$ on configuration $x$) and look for $F(x)_u$. If $F(x)_u =1$ it means that $u$ infiltrates, and the algorithm \emph{accepts}. If $F(x)_u=0$ it means that $u$ is stable, and the algorithm \emph{rejects}.  

We suppose now that $u$ belongs to $V_+$. We define $G_{+}$ as the graph induced by $V_{+}$, and ket us call $G_+[u]$ the connected component of $G_{+}$ containing $u$, and $V_{+}[u]$ the vertex set of $G_+[u]$.  We also define  $B[u]$, called \emph{boundary} of $G_{+}[u]$, as the cells in the complement of $V_{+}[u]$ with at least one neighbor in $V_{+}[u]$, formally:
  $$B_{+}[u]=\{v \not\in V_{+1}[u]: V_{+}[u] \cap N(v)\not= \emptyset \}. $$

As it is shown in the following lemma, to know whether the decision cell $u$ is unstable it is enough to decide the existence of a cell in $B[u]$ that infiltrates. 

\begin{lemma}\label{lem:boundary}
Cell $u$ is unstable if and only if there is a cell in $B[u]$ that infiltrates.
\end{lemma}
\begin{proof}
  Suppose first that no cell of $B[u]$ infiltrates. From Lemma \ref{lem:infil} we know that a cell in $B[u]$ that does not infiltrates is necessarily stable.  Since all vertices of $V_+[u]$ need at least one more active neighbor to become active, and all the vertices on the border of $V_+[u]$ are stable, we deduce that all vertices in $V_+[u]$ are stable, in particular $u$.  
  
 Let $v \in B[u]$ a cell that infiltrates. Since $G_{+}[u]$ is connected, there exists a $(v,u)$-path $P$ in the grid, such that all internal vertices belong to $V_+[u]$. Over all possible paths, we take an induced path $(v,u)$, for example a shortest one. Name $v_1, \hdots, v_k$ be the cells of such a path, where $v_1 = v$ and $v_k = u$, and consider the updating scheme $\sigma$ that updates the cells in the order given by their names (i.e. first update $v_1$, then $v_2$, and so on until it updates $v_k$ (the updating scheme can update the rest of the cells in an arbitrary order). We claim that $\sigma$ iterates all the cells of the path in one updating. Indeed, observe that $v$ infiltrates, so the claim is true for the first cells of the path. Inductively, if we consider that the claim is true for the first $i$ vertices of the path, for $i \in [k]$, then the claim is also true for the $i+1$ fist vertices of the path. Indeed, since the path $P$ is induced, at the moment on which we update vertex $v_{i+1}$, the state of its neighbors is invariant with respect to the initial configuration, with the exception of $v_i$, which now is active (by the induction hypothesis). Since $v_{i+1}$ belongs to $V_+$, we deduce that $\sigma$ iterates $v_{i+1}$. We conclude then that $\sigma$ iterates $u$, and then $u$ is unstable. 
\end{proof}
  
The latter lemma allow us to solve \AsyncStability when $u$ belongs to $V_+$.  Indeed, given this result, we only need identify the cells in $B[u]$ that infiltrate.   More precisely, our algorithm has to perform two tasks: (1) compute the set $B[u]$ and (2) decide if any cell in $B[u]$ infiltrate. Task (2) can be done quite efficiently, computing in parallel $F(x)$. Task (1) is slightly more complicated.  It consists in (1.1) compute the set $V_+$; (1.2) compute the set of edges of $G+$, (1.3) compute the connected components of $G_+$; (1.4) Compute $V_+[u]$; and (1.5) Compute $B[u]$. 
  
Fortunately, there exist a fast-parallel algorithm that computes the connected components of a graph. 

\begin{proposition}[\cite{JaJa:1992:IPA:133889}]\label{prop:conjaja}
There is an algorithm that, given the adjacency matrix of an undirected  $n$-vertex graph $G$, computes the connected components of $G$ in time $\mathcal{O}(\log^2n)$ using $\mathcal{O}(n{^2}) $ processors. The output of the algorithm is a vector $C$ of length $n$, such that, for each pair of indices $i,j \in [n]$ , when $C_i = C_j$ it means that the $i$-th vertex and the $j$-th vertex of $G$ are in the same connected component. 
\end{proposition}

In our algorithm solving \AsyncStability we are going to need another subroutine, which is called a \emph{prefix-sum algorithm}. Given an integer vector $Vec$, the \emph{prefix-sum} of $Vec$ is an integer $sum$ corresponding to the sum of the elements of $Vec$ (actually the prefix-sum is vector with the sum of all prefixes of $Vec$, but we are going to need only the sum of all elements)

\begin{proposition}[\cite{JaJa:1992:IPA:133889}]\label{prop:jajaprefixsum}
There is an algorithm that, given a vector $Vec$ of length $n$, computes the sum of all elements of $Vec$ in time $\cO(\log n)$ using $\cO(n)$ processors. 
\end{proposition}

We are now ready to prove the main result of this subsection.
  
\begin{theorem}\label{thm: infiltration tri}
For every infiltration rule \AsyncStability  is in \NC. 
\end{theorem}
\begin{proof}
Let $(x, u)$ be an input of {\AsyncStability}, i.e. $x$ is a finite configuration of $G(n)$, and $u$ is a vertex of $G(u)$. 
Our algorithm for {\AsyncStability} first computes $F(x)$. Then, it computes the set $V_+$ as it is defined above. The algorithm then checks whether $u$ belongs to $V_+$. If it is not contained in $V_+$, the algorithm checks if $u$ infiltrates by looking to $F(x)_u$. If $F(x)_u=1$ it means that $u$ infiltrates and the algorithm \emph{accepts}. Otherwise, from Lemma \ref{lem:infil} we know that $u$ is stable, so the algorithm \emph{rejects}. In the case when $u$ belongs to $V_+$, the algorithm computes the set $B[u]$ and looks for a vertex in that set that infiltrates. From Lemma \ref{lem:boundary} we know that, if some vertex infiltrates $u$ is unstable, and otherwise, if no vertex of $B[u]$ infiltrates, then $u$ is stable. 
Therefore algorithm outputs \emph{accept} some vertex in $B[u]$ infiltrates and \emph{rejects} otherwise. %
The details of our algorithm are given in Algorithm \ref{algo:NCinfil}.
  
\begin{algorithm}[H]

\caption{\NC algorithm for infiltration rules}\label{alg: AsyncStability 1}
  \label{algo:NCinfil}
  \begin{algorithmic}[1]  
  \REQUIRE $x\in \{0,1\}^{G(n)}$  and $u \in G(n)$ such that $x_u =0$.   
  \STATE Compute $F(x)$.
  \STATE Compute the $V_{+}$.
    \IF{$u$ belongs to $V_+$ } 
    \IF{$F(x)_u = 1$}
    \RETURN \emph{accept}.
    \ELSE
    \RETURN \emph{reject}.
    \ENDIF
    \ELSE
  \STATE Compute the adjacency matrix of graph $G_{+}$.
  \STATE Compute the connected components of $G_{+}$.
    \STATE Compute $V_{+}[u]$.
  \STATE Compute the $B[u]=\{v \not\in V_{+}[u]: V_{+}[u] \cap N(v)\not= \emptyset \} $.
  \FORALLP{ $v \in B[u]$}
    \IF{$F(x)_v=1$ } 
      \RETURN \emph{Accept}
    \ENDIF  
  \ENDFORALLP
  \RETURN  \emph{Reject} 
  \ENDIF
  \end{algorithmic}
\end{algorithm}

We now analyze the complexity of the algorithm. Let us call $N = |G(n)|$.

\begin{description}
  \item[ { Step 1}] can be computed in  $\cO(\log N)$ time using $\cO(N)$ processors. First, define a vector $VecS$ of length $N$. One processor is assigned to each cell. The processor assigned to cell $i$ computes $\sum_{w\in N(v)}x_w$ in $\cO(\log n)$ time and saves the output value in the $i$-th coordinate of $VecS$. Then, define a vector $VecF$ of length $N$. The processor assigned to cell $i$ computes $F(x)_i$ according on the table of rule $F$, it and saves the value in the $i$-th coordinate of $VecF$.
  
  \item[{Step 2}] can be computed in time $\cO(\log N)$ using $\cO(N)$ processors. First, define an empty vector $VecV$ of length $N$. One processor is assigned to each cell. The processor assigned to cell $i$ looks to the $i$-th coordinate of vector $VecS$ computed in {\bf Step 1} and decides if $\sum_{w\in N(v)}x_w \notin \mathcal{I}_F$ and $\sum_{w\in N(v)}x_w +1 \in \mathcal{I}_F$. If the answer is affirmative, it writes a $1$ in the $i$-th coordinate of $VecV$. If the answer is negative, it writes a $0$ in the same coordinate.  
   
    \item[{Steps 3-9}] can be computed in $\cO(\log N)$ time with one processor. The processor looks to the coordinate corresponding to $u$ in $VecV$ and $VecF$.
   
     \item[{Step 10}] can be computed in $\cO(\log N)$ time using $\cO(N^2)$ processors. First, define an empty matrix $N \times N$ called  $MatG$. Then, one processor is assigned to each pair of cells $i$ and $j$. If $VecV_i=1$ and $VecV_j=1$ and $i$ and $j$ are adjacent, then the processor writes a $1$ in the coordinate $(i,j)$ of $MatG$, and otherwise it writes a $0$ in the same coordinate. 
     
          \item[{Step 11}] can be computed in $\cO(\log^2 N)$ time using $\cO(N^2)$ processors using the algorithm given by Proposition \ref{prop:conjaja}. 
          
          \item[{Step 12}] can be computed in $\cO(\log N)$ time using $\cO(N)$ processors. First, define an empty vector $VecVu$ of length $N$. Then, one processor is assigned to each cell. The processor assigned to cell $i$ checks if $C_i = C_u$ and if $VecV_i=1$. If the answer is affirmative, it writes a $1$ in the $i$-th coordinate of $VecVu$ . Otherwise, it writes a $0$ in the same coordinate.
          
          \item[{Step 13}] can be computed in $\cO(\log n)$ time using $\cO(N)$ processors. First, define am empty vector $VecB$ of length $N$: Then, one processor is assigned to each cell. The processor assigned to cell $i$ checks if $VecVu_i = 0$ and if $i$ has a neighbor $j$ such that  $VecVu_j=1$. If the answer is affirmative the processor writes a $1$ in the $i$-th coordinate of $VecB$ and otherwise it writes a $0$ in the same coordinate. 
          
\item[{Steps 14-20}] can be computed in $\cO(\log N)$ using $\cO(N)$ processors. First, define an empty vector $VecInf$ of length $N$. Then, one processor is assigned to each cell. The processor assigned to cell $i$ checks if $VecB_i=1$ and $VecF_i=1$. If the answer is affirmative, the processor writes a $1$ in the $i$-th coordinate of $VecInf$. Otherwise, it writes a $0$ in the same coordinate. After each processor finishes, run the prefix-sum algorithm of Proposition \ref{prop:jajaprefixsum} on input $VecInf$, and call $sum$ the output. If $sum=1$ the algorithm accepts, and if $sum=0$ the algorithm rejects. 
\end{description}

We deduce that our algorithm runs in time $\cO(\log^2 N)$ using $\cO(N^2)$ processors, and therefore our \AsyncStability is in \NC.
\end{proof}

\subsection{Monotone-like rules} \label{subsec: monotone rules}

We now study monotone-like rules.  In \cite{StabilityMajority} it is shown that all FCA that are  monotone have a particular property:  From a given initial configuration, they reach the same fixed point under any updating scheme.

\begin{proposition}[\cite{StabilityMajority}]\label{prp: freezing and monotone}
  If $F$ is a freezing and monotone CA, then starting from any configuration,  any updating scheme of $F$ reaches the same fixed point than the synchronous update. .
\end{proposition}
 We will call \stability to the version of \AsyncStability where the updating scheme is synchronous, i.e. the problem is to decide is a cell is stable updating every cell at same time. 
Thus in freezing monotone CA is equivalent solve \stability or \AsyncStability. This proves that $\AsyncStability$ is solvable in polynomial time for every $FCA$. 

\begin{proposition}[\cite{StabilityMajority}]
For every monotone FCA, \AsyncStability is in \Pt.
\end{proposition}

Observe that the two-dimensional monotone $LFCA$ rules are rule $T23$ in the triangular grid and  rules $S24$ and $S34$ in the square grid. Interestingly, for these rules there exist better algorithm solving \AsyncStability. Let us call $maj$ the \emph{freezing majority rule}, which consists in the freezing  rule where every cell takes the state of the majority of its neighbors. In case of a tie (same number of inactive and active neighbors), the cell becomes active. Similarly, we call $Maj$ the \emph{freezing strict majority rule}, which is defined similarly, except that in tie case the inactive cells remain inactive. Observe that rule $S24$ corresponds to $maj$, while $S34$ corresponds to $Maj$ in the square grid. Moreover, since the triangular grid is a topology of odd degree,  $T23$ corresponds to $maj$ and $Maj$.
 
In \cite{StabilityFTCA} it is shown that for $S24$ (i.e. $maj$ in the squared grid) problem \stability is in \NC.

\begin{proposition}[\cite{StabilityFTCA}]\label{prop:diego1}
For rule $S24$ \stability is in \NC.
\end{proposition}

On the other hand, in \cite{goles:hal-00914603} it is shown that for $Maj$ in any topology of degree at most $4$, problem \stability is in \NC. 

\begin{proposition}[\cite{goles:hal-00914603}]\label{prop:otra}
For rule $Maj$ restricted to a topology of degree at most $4$, problem \stability is in \NC.
\end{proposition}

Observe that the triangular grid is a topology of degree $3$ and the square grid is a  topology of degree $4$. We deduce that for rules $T23$ and $S34$ the problem \stability is in \NC. This observations pipelined with Proposition \ref{prp: freezing and monotone} prove the following lemma.

  \begin{lemma}\label{lem:monotoneasync} 
 For rule $T23$, $S24$ or $S34$, \AsyncStability is in \NC.
  \end{lemma}

We are now study rules $T22$, $S23$, $S33$. Observe that these rules are \emph{almost} monotone, except in the case when the sum of the neighborhood of a cell equals $|N(\cdot)|$ (which is $3$ in the triangular grid or $4$ in the square grid). Let $F$ be a LFCA  and let $F^*$ be the rule such that $\mathcal{I}_{F^*} = \mathcal{I}_F \setminus \{|N(\cdot)|\}$ (i.e. $\mathcal{I}_{F^*} = \mathcal{I}_F \setminus \{4\}$ if $F$ is defined over the square grid and $\mathcal{I}_{F^*} = \mathcal{I}_F \setminus \{3\}$ if $F$ is defined over the triangular grid). The following lemma relates the unstability of a cell on rules $F$ and $F^*$.

  \begin{lemma}\label{lem:complexitysinmax}
 Let $x$ be a configuration, $u$ an inactive cell, and suppose that $$\exists k \in \mathcal{I}_{F^*} \textrm{ such that } \sum_{w\in N(u)}x_w \leq k.$$  Then $u$ is stable for $x$ in rule and  $F$ if and only if $u$ is stable for $x$ and rule $F^*$.
  \end{lemma}

\begin{proof}
We suppose that $|N(\cdot)|$ belongs to $\mathcal{I}_F$, because otherwise the lemma is trivially true. 
Suppose first that $u$ is stable for $F$. Then necessarily $u$ is stable for $F^*$, because an updating scheme that iterates $u$ under rule $F^*$ will also iterate $u$ under rule $F$, since $\mathcal{I}_{F^*} \subseteq \mathcal{I}_F$.

Suppose now that there is a vertex $u$ that is stable for rule $F^*$ but not for $F$. Call $\sigma$ the updating scheme that iterates $u$ under rule $F$ and let $t$ be such that $F^{\sigma(t)}(x)_u=1$. Over all possible updating schemes, we choose the one such that $t$ is minimum (i.e. $\sigma$ iterates $u$ in the minimum possible number of time-steps). This implies in all time-steps smaller than $t$, no vertex with exactly $|N(\cdot)|$ neighbors are updated, because otherwise $t$ would not be the minimum. This implies that, until time-step $t$, the dynamics over rule $F$ and $F^*$ coincide. Now, we are supposing that $u$ is stable for $F^*$, and it implies that $ \tilde{F}^{\sigma(t)}(x)_u=0$. Therefore in configuration  $F^{\sigma(t-1)}(x)$ cell $u$ has $|N(\cdot)|$ active neighbors.  
As our hypothesis is that $ \sum_{w\in N(u)}x_w \leq k$ for some $k\in \mathcal{I}_{F*}$, and in time-step $t$ $u$ has $N(\cdot)>k$ active neighbors, necessarily there exists a time-step $t^* < t$ on which $u$ has exactly $k$ active neighbors. Let $\sigma*$ be the updating scheme such that, until time-step $t^*$ is equal to $\sigma$, but in time-step $t^*+1$ it updates $u$. We obtain that on rule $F^*$ updated according to $\sigma*$ cell $u$ becomes active in time-step $t^*+1$. This contradicts the fact that $u$ was stable for $F^*$. We deduce that, under the conditions of the lemma, if $u$ is stable for $F^*$ it is necessarily stable for $F$. 
\end{proof}

We are now ready to show the main result of this subsection.

  \begin{theorem}\label{thm: stabilitymonotone} 
 For every monotone-like rule \AsyncStability is in \NC.
  \end{theorem}

\begin{proof}

Let us suppose first that $F$ is rule $T23$, $S24$ or $S34$. Then, the statement of the theorem holds from Lemma \ref{lem:monotoneasync}. 

Now suppose that $F$ is rule $T22$ (respectively $S23$ or $S33$). Let $(x,u)$ be an instance of \AsyncStability, i.e. $x$ is a configuration of active and inactive cells and $u$ is the decision cell. Lemma \ref{lem:complexitysinmax} indicates that on the condition that $\sum_{w\in N(u)}x_w \leq 2$ (respectively $\sum_{w\in N(u)}x_w \leq 2$ or $\sum_{w\in N(u)}x_w \leq 3$), then the complexity of the problem \AsyncStability for rule $T22$ (respectively $S23$ or $S33$) coincides with the complexity of \AsyncStability for rule $T23$ (respectively $S24$ or $S34$). 
Moreover, if previous condition is not satisfied, i.e.   $\sum_{w\in N(u)}x_w > 2$ (respectively $\sum_{w\in N(u)}x_w > 3$ for $S23$ or $S33$) then $u$ is stable for rule $F$.

Therefore, the algorithm solving \AsyncStability for rule $T22$ (respectively $S23$ or $S33$) consists in the following steps: First, it computes the answer of \AsyncStability for rule $T23$ (respectively $S24$, $S34$) with the \NC algorithm given by Lemma \ref{lem:monotoneasync}. If the answer is that $u$ is stable, then the algorithm rejects. Otherwise, the algorithm checks whether $\sum_{w\in N(u)}x_w > 2$ (respectively $\sum_{w\in N(u)}x_w > 3$ for rules $S24$ or $S34$). If the answer is yes, the algorithm rejects, otherwise it accepts. Clearly previous algorithm shows that for all these rules \AsyncStability is in \NC, because checking the condition $\sum_{w\in N(u)}x_w > 2$ (respectively $\sum_{w\in N(u)}x_w > 3$ for rules $S24$ or $S34$) can be done in $\cO(\log n)$ time on a sequential machine.   
\end{proof}

\section{Lower-Bounds}\label{sec:complexityofruleS22}

In last section we had studied all LFCA defined over a triangular grid, showing that for all of them problem \AsyncStability is in \NC. Moreover, we obtained a similar result regarding the square grid, except for one rule, namely rule $S22$. We remember that in this rule inactive cells become active only if they have exactly two active neighbors.  In this section, we show that for this rule,  \AsyncStability is \NPC. In order to achieve this task, we reduce  \textsc{Boolean Satisfiability (SAT)} to this problem. 

A \emph{Boolean variable} is a variable $x \in \{0,1\}$ (it can be either \emph{true} or \emph{false}). An $n$-variable \emph{Boolean formula} $\phi: \{0,1\}^n \to \{0,1\}$ is a function  defined over a set of Boolean variables $x_1, \dots, x_n$, and which value is computed as a combination of conjunctions ($\wedge$), disjunctions ($\vee$) and negation ($\neg$) of the variables.  A Boolean formula is called \emph{satisfiable} if there exists a truth-assignment of the variables of $\phi$, i.e. if there exists $u \in \{0,1\}^n$ such that $\phi(u)=1$. Problem \SAT consists in deciding whether a given Boolean formula $\phi$ on $n$ variables is satisfiable.  The famous Cook-Levin Theorem states that \SAT is \NPC   \cite{arora2009computational}. 

A \emph{Boolean circuit} $C$ is a labeled directed acyclic graph where each node has in-degree (fan-in) at most $2$. Every other node is called \emph{gate} and it is labelled with one Boolean operator: $\wedge$ (conjunction), $\vee$ (disjunction) or $\neg$ (negation). Nodes labeled $\neg$ have in-degree $1$. The gates with in-degree are called \emph{inputs} and vertices with out-degree $0$ are called \emph{outputs}. All other node has in-degree and out-degree $2$. Every node has a unique Boolean value assigned by the evaluation of its correspondent boolean operator in the values given by its incoming neighbors. Thus, the circuit is evaluated by layers, defined as the distance of a gate to an input node. Problem \CSAT is the problem of, given a Boolean circuit and an output gate $g$, decide if there exists a truth-assignment of the inputs of $G$ for which $g$ is satisfiable (evaluated \emph{true}). Clearly every Boolean formula can be represented as a Boolean circuit, where there is only one output gate. Therefore \CSAT is also \NPC. 

Our reduction of \SAT to \AsyncStability for rule $S22$ consists in, roughly, a representation of an arbitrary Boolean Circuit $C$ into a pattern of a two dimensional grid.  This representation considers that each gate is associated with specific cell of the grid, in a way that, this specific cell is unstable if and only if the corresponding gate is satisfiable. In the following, we give some intermediate results that allow to build this representation.

\subsection{Grid-embedded Boolean Circuits}

Our simulation framework considers a specific representation of a Boolean circuit, that we call \emph{south-east-output, north-west input grid-embedded Boolean circuits}, or simply \emph{grid-embedded Boolean circuits}. This representation is inspired in another one given by Goldschlager in \cite{Goldschlager1977}.  A grid-embedded Boolean circuit is a Boolean circuit embedded in a two-dimensional square grid. Each cell of the grid contains an operator, that we call in this context a \emph{block}, representing the parts of the circuit (i.e. gates and edges).  Each cell of the grid is occupied by some block, which communicates with its four neighbors (north, east, south, west), in such a way that each one implements a Boolean function $g$ with 2 inputs and and 2 outputs. Each block receives inputs $s_N$ and $s_W$ from the north and west blocks, and  sends signals $g_E(s_N, s_W)$ and $g_S(s_N,s_W)$ to the east and south blocks. 
 
In order to define a grid-embedded circuit, we first need to define a \emph{directed acyclic square grid}. A directed acyclic square grid is a directed graph on $n^2$ vertices embedded in the two-dimensional grid $G(n)$, where each cell contains a vertex.  Let us fix coordinates $(0,0)$ on an arbitrary cell. The incoming neighbors of vertex placed in cell $(i,j)$ are the vertices placed in cells $(i-1,j)$ and $(i,j-1)$, and the outgoing neighbors of $(i,j)$ are the vertices placed in cells $(i+1,j)$ and $(i,j+1)$ (all values are computed modulo $n$, remember that $G(n)$ is defined as a torus).  Hence, each vertex has in-degree  $2$ and out-degree $2$.  We consider that each vertex is labeled with one Boolean operator, called \emph{block}. Let us define the following blocks:
\begin{enumerate}
\item \emph{Conjunction block}: it is represented by the symbol $\wedge$ and it outputs the conjunction of the two signals through the east and south sides. Formally, for this gate $g^{}_E(s_N, s_W) = g^{}_S(s_N, s_W) = s_N \wedge s_W$.

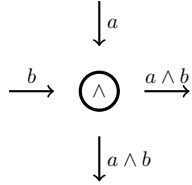
\begin{center}
  \begin{tikzpicture}[scale=0.8,every node/.style={draw,scale=0.8}]  
    \node[very thick, circle] (v) at (0,0){$\wedge$};
    \draw[thick, ->] (0,1.5) --node[draw=none,right, midway](){$a$} (0,0.75); 
    \draw[thick, ->] (-1.5,0) --node[draw=none,above, midway](){$b$} (-0.75,0) ;
    \draw[thick, ->] (0.75,0) --node[draw=none, above, midway](){$a\wedge b$} (1.5,0);
    \draw[thick, ->] (0,-0.75) --node[draw=none,right, midway](){$a\wedge b$} (0,-1.5);
  \end{tikzpicture}
\end{center}

\item \emph{Disjunciton block}: it is represented by the symbol $\vee$ and it outputs the disjunction of the two signals through the east and south sides. Formally, for this gate $g_E(s_N, s_W) = g_S(s_N, s_W) = s_N \vee s_W$.

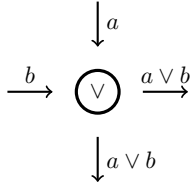
\begin{center}
  \begin{tikzpicture}[scale=0.8,every node/.style={draw,scale=0.9}]    
    \node[very thick, circle] (v) at (0,0){$\vee$};  
     \draw[thick, ->] (0,1.5) --node[draw=none,right, midway](){$a$} (0,0.75); 
    \draw[thick, ->] (-1.5,0) --node[draw=none,above, midway](){$b$} (-0.75,0);
    \draw[thick, ->] (0.75,0) --node[draw=none, above, midway](){$a\vee b$} (1.5,0);
    \draw[thick, ->] (0,-0.75) --node[draw=none,right, midway](){$a\vee b$} (0,-1.5);
  \end{tikzpicture}
\end{center}

\item\emph{Crossing block}: it is represented by the letter $C$. This gadget crosses the north input through the south output, and the west input through the east output. Formally, for this gate $g^{}_E(s_N, s_W)  = s_W$ and  $g^{}_S(s_N, s_W) = s_N$.
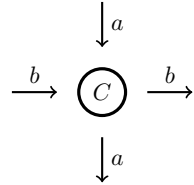
\begin{center}
\begin{tikzpicture}[scale=0.8,every node/.style={draw,scale=0.9}]  
  \node[very thick, circle] (v) at (0,0){$C$};  
  \draw[thick, ->] (-1.5,0) --node[draw=none,above, midway](){$b$} (-0.75,0) ;
  \draw[thick, ->] (0.75,0) --node[draw=none,above, midway](){$b$} (1.5,0);
  \draw[thick, ->] (0,1.5) --node[draw=none,right, midway](){$a$} (0,0.75);
  \draw[thick, ->] (0,-0.75) --node[draw=none,right, midway](){$a$} (0,-1.5);
\end{tikzpicture}
\end{center}

\item {\it Fixed-value block}: it is represented by a $0$. This block   ignores the values of its input (the associated functions $g$ are constant functions), and outputs a fixed value \emph{false} through the east and south blocks.

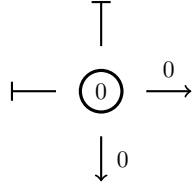
\begin{center}
\begin{tikzpicture}[scale=0.8,every node/.style={circle, draw,scale=0.9}]
  \node[very thick] (v) at (0,0){$0$};  
  \draw[thick, |-] (-1.5,0) --node[draw=none,above, midway](){} (-0.75,0) ;
  \draw[thick, ->] (0.75,0) --node[draw=none,above, midway](){$0$} (1.5,0);
  \draw[thick, |-] (0,1.5) --node[draw=none,right, midway](){} (0,0.75);
  \draw[thick, ->] (0,-0.75) --node[draw=none,right, midway](){$0$} (0,-1.5);
\end{tikzpicture}
\end{center}

\medskip 

\item \emph{Signal-multiplier block}: this gadget is represented by a symbol $M$, which only reads one input,  either from the north or east side, and outputs that value through the east and south sides.  Formally, for this gate $g_E(s_N, s_W)  = g_S(s_N, s_W) = s_N$ (if the block is a north-signal multiplier) or $g_E(s_N, s_W)  = g_S(s_N, s_W) = s_W$ (if the block is a west-signal multiplier).

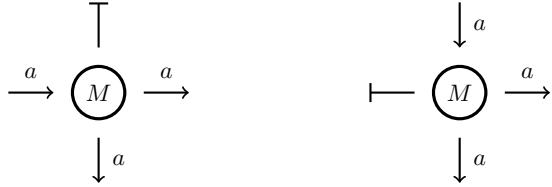
\begin{center}
\begin{tikzpicture}[scale=0.8,every node/.style={circle, draw,scale=0.9}]
  
  \node[very thick] (v) at (-2,0){$M$};
  
  \draw[thick, ->] (-3.5,0) --node[draw=none,above, midway](){$a$} (-2.75,0);
  \draw[thick, ->] (-1.25,0) --node[draw=none,above, midway](){$a$} (-0.5,0);
 \phantom{ \draw[thick, ->] (0,1.5) -- (0,0.75);}
  \draw[thick, ->] (-2,-0.75) --node[draw=none,right, midway](){$a$}  (-2,-1.5);
    \draw[thick, |-] (-2,1.5) -- (-2,0.75);
    
   \node[very thick] (v) at (4,0){$M$};
  
   \draw[thick, ->] (4,1.5) --node[draw=none,right, midway](){$a$} (4,0.75);
  \draw[thick, ->] (4.75,0) --node[draw=none,above, midway](){$a$} (5.5,0);
 \draw[thick, |-] (2.5,0) -- (3.25,0);
  \draw[thick, ->] (4,-0.75) --node[draw=none,right, midway](){$a$}  (4,-1.5);
\end{tikzpicture}
\end{center}

\item \emph{Selector block}: this gadget is represented by the symbol $S$. It is a block that receives no input (the functions $g$ are constant functions), and outputs a truth value $x$ through the east side and $\neg x$ through the south side. Formally, for this gate $g_E(s_N, s_W)  = 1$ and $g_S(s_N, s_W)  = 0$ (if the block is a east-selector) or $g_E(s_N, s_W)  = 0$ and $g_S(s_N, s_W) = 1$ (if the block is a south-selector).

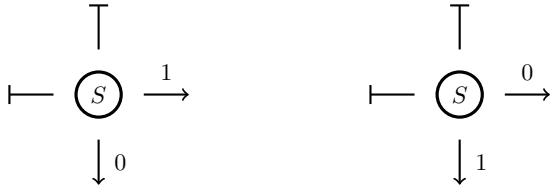
\begin{center}
\begin{tikzpicture}
  [scale=0.8,every node/.style={circle, draw,scale=0.9}]
  
   \node[very thick] (v) at (-2,0){$S$};
  \draw[thick, |-] (-3.5,0) --node[draw=none,above, midway](){} (-2.75,0);
  \draw[thick, ->] (-1.25,0) --node[draw=none,above, midway](){$1$} (-0.5,0);
 \phantom{ \draw[thick, ->] (0,1.5) -- (0,0.75);}
  \draw[thick, ->] (-2,-0.75) --node[draw=none,right, midway](){$0$}  (-2,-1.5);
    \draw[thick, |-] (-2,1.5) -- (-2,0.75);
    
   \node[very thick] (v) at (4,0){$S$};
  
   \draw[thick, |-] (4,1.5) --node[draw=none,right, midway](){} (4,0.75);
  \draw[thick, ->] (4.75,0) --node[draw=none,above, midway](){$0$} (5.5,0);
 \draw[thick, |-] (2.5,0) -- (3.25,0);
  \draw[thick, ->] (4,-0.75) --node[draw=none,right, midway](){$1$}  (4,-1.5);
\end{tikzpicture}
\end{center}

\end{enumerate}

A \emph{grid-embedded Boolean circuit} $C$ is a labeled directed acyclic square grid $G(n)$ where each  vertex is labeled as some block, where:
\begin{itemize}
\item[(1)] each vertex placed in with coordinates $(i,j)$ such that $ij = 0$ (i.e. either $i$ or $j$ are $0$) is labeled as a fixed-value block. 
\item[(2)] and every other vertex is labeled as a one of the blocks defined latter: conjunction, disjunction, crossing, fixed-value, signal multiplier or selector block.  
\end{itemize}
Observe that condition (1) (and the fact that $G(n)$ is a torus) implies that the perimeter of the grid $G(n)$ consist in fixed-value blocks.  Let $v_1, \dots, v_s$ be the selector blocks of $C$. A truth-assignment of $C$ is defined as a vector $u \in \{0,1\}^s$ such that, if $u_i = 1$ then $v_i$ is defined as an east-selector block, and if $u_i=0$ then $v_i$ is defined as a south-selector block. A truth-assignment defines a set of \emph{values} of each vertex of $C$, called $C(u)$, which is computed evaluating the function of each block, in the following order: starting from $k=0$, compute the outputs of all cells $(i,j)$ such that $i+j = k$. Then, sequentially repeat the procedure for $k=1, \dots, 2(n-1)$ (in other words, we evaluate vertex in cell $(0,0)$, then $(1,0)$ and $(0,1)$, then $(2,0), (1,1), (0,2)$, and so on until reaching cell in coordinates $(n-1,n-1)$. For a vertex $v$, we call $C(u)_v$ the value of $v$ on truth-assignment $u$. We say that $v$ is satisfiable if there exists a truth-assignment $u$ of $C$ such that $C(u)_v=1$. Note that, $C$ represents a boolean function $\{0,1\}^s \to \{0,1\}^{n^2}$.

Given grid-embedded Boolean circuit $C$ and a vertex $v$ of $C$, the problem \textsc{Grid-Embedded Circuit Satisfiability (\GCSAT)} consists in deciding whether there exist a truth-assignment $u$ of $C$ such that $C(u)_v = 1$. Note that \GCSAT can be verified in polynomial time (the certificate is simply a truth-assignment of $C$ satisfying $v$), thus the problem is in \NP. The following result shows that this problem is as hard as \SAT. As a consequence of this together with the latter observation we have that \GCSAT is \NPC.

\begin{theorem}
\textsc{Grid-Embedded Circuit Satisfiability} is \NPC.
\label{teo:GCSATNP}
\end{theorem}

\begin{proof}
   We reduce \SAT to \GCSAT. Let $\varphi$ be an instance of \SAT, i.e. $\varphi$ is a CNF formula on $n$ variables. We represent $\varphi$ as a Boolean circuit $C$, where all negations belong to the first layer, i.e. their incoming neighbors are input gates, and every gate has in-degree at most $2$.   In the literature $C$ is called a \emph{normalized circuit}, and observe that $C$ can be computed in polynomial time in the size of $\varphi$   \cite{arora2009computational}.

Without loss of generality, we assume that $C$  satisfies that: 
\begin{itemize}
\item[(1)] every gate that is not an output has out-degree $2$,
\item[(2)] each input gate is connected to exactly one negation gate, 
\item[(3)] each gate that is in the first layer (i.e. has an input gate as incoming neighbor) has in-degree $1$.

\end{itemize}
Indeed,  we can make assumptions (1), (2) and (3) by observing that we can simulate a \emph{signal multiplier} using a $\vee$ gate with only one input. Now suppose that there exists a gate $g$ with out-degree $d$ greater than $2$. We can replace that gate by a directed tree of $\cO(\log(d))$ signal multipliers, all of out-degree $2$. In the same way assume that each input gate has out-degree at most $2$ and it is connected to at most one negation. If an output gate has degree fewer than $2$, we add a new dummy output gate, labeled $\vee$, if the other neighbor is a negation, or $\neg$ otherwise. Finally, if a gate of the first layer has in-degree $2$, replace it by a signal multiplier. 

Let $n$ be the number of gates of $C$.  We define a grid-embedded circuit $\tilde{C}$ over a directed acyclic square grid $G(n+1)$ that simulates $C$.  This simulation satisfies that: \begin{enumerate}
  \item $\tilde{C}$ can be computed in polynomial time given a representation of $C$ (which in turn represents $\varphi$).
  \item there exists a vertex $v$ of $\tilde{C}$ such that $\varphi$ is satisfiable if and only if $v$ is satisfiable. 
\end{enumerate}

In order to build $\tilde{C}$, we have to define which block we are going to assign to each vertex. The algorithm starts assigning unique identifiers in $\{1, \dots, n\}$ to the gates of $C$ in concordance to a lexicographic order, i.e. if gate $g_1$ is an incoming neighbor of gate $g_2$, then the identifier of $g_1$ is strictly smaller than the numbering of $g_2$. In the following we do not distinguish a gate from its identifier.  Given this latter enumeration for each gate, we assign to each of them a unique position in the two-dimensional integer lattice: for each gate $g \in C$ we assign to it a cell in the grid with coordinates $(g,g)$. In addition, we assign label to each cell in order to define the gates of $\tilde{C}$:
\begin{itemize}
  \item For each gate $g \in C$ that is not an input nor a negation gate label it by a conjunction or disjunction gadget (block), depending on the type of gate $g$ in the circuit $C$.
  \item  If $g$ is a negation gate, label it as a disjunction block.
  \item  Finally, if $g$ is an input gate, label it as selector block. 
\end{itemize}
 In all cases, if $v$ is the vertex of $\tilde{C}$ in $(g,g)$, we say that $v$ \emph{represents} $g$, and we denote it by $v(g)$. 

In the vertices of the top row, that is, vertices in cells $(i,0)$ for $i=0,\dots,n$, assign a fixed-value block with value $0$. In the leftmost column, that is, vertices in cells with coordinates $(0,i)$ for $i=1,\dots,n$ assign a fixed-value block with value $0$, except in rows $j \in \{1, \dots n\}$  where gate numbered $j$ is a $\wedge$ gate with in-degree $1$. In those cases, label the vertex in cell $(0,j)$ with a fixed-value block in with value $1$. We remark that it is also possible to replace the latter gate $\wedge$ with a gate $\vee$ without loss of generality as it has only one input (see Figure \ref{fig:ejemploreduccionGCSAT}).

Let $g$ be a block in the first layer of $C$. Remember that we assume that $g$ has in-degree $1$. Call $g_1$ the incoming neighbor of $g$ (which is a input gate since). If $g$ is a negation gate, then we define the block of $v(g)$ in $\tilde{C}$ in way such it receives the signal from $g_1$ through the west block. In order to achieve that, we assign to cell $(g, g_1)$ a north-signal multiplier. On the other hand, if $g$ is not a negation gate, we define the block of $v(g)$ in a way such that it receives a signal from $g_1$ through the north block. Therefore, in this case we assign to cell $(g_1,g)$ an east-signal multiplier. 

 Now suppose that $g$ is not in the first layer, nor an input gate. Suppose that the in-degree of $g$ equals $1$. Then, we repeat the latter assignation so that the incoming input signal of the block of $v(g)$ comes through the north cell. Therefore, if $g_1$ is the incoming neighbor of $g$, we label the vertex in cell $(g_1, g)$ as a west-signal multiplier. Suppose now that $g$ has in-degree $2$, and let $g_1$ and $g_2$ be its two inputs, with $g_1<g_2$. Then, in $\tilde{C}$, we consider that the signal sent from $g_1$ will arrive to the block assigned to $v(g)$ through north input, while the signal sent from $g_2$ will be received through the west input. Therefore, we label the vertex in cell $(g_1, g)$ by a west-signal multiplier block, and the vertex in cell $(g, g_2)$ by a north-signal multiplier block. 

Finally, every block that remains unlabeled is labeled  as a crossing block.

 In Figure \ref{fig:ejemploreduccionGCSAT} is an example of this construction. 

 Let $g$ be a gate of $C$ . We now show that $g$ is satisfiable in $C$ if and only if $v_g$ is satisfiable in $\tilde{G}$. Let $u$ be a  truth-assignment of the inputs of $C$ and consider the truth-assignment $\tilde{u}$ of the selectors of $\tilde{C}$ such that, if $g$ is an input gate of $C$, then $\tilde{u}_{v(g)}=u_g$.  We claim that, for every gate $g$ of $C$, $C(u)_g = \tilde{C}(\tilde{u})_{v(g)}$.
Indeed, let $g$ be a gate and let $v_g$ be the corresponding vertex of $\tilde{C}$, which is placed in cell $(g,g)$. 

Suppose that $g$ has in-degree $2$, and call $g_1$ and $g_2$ the two inputs of $g$ with $g_1<g_2$.  Observe that the construction of $\tilde{C}$ defines that all cells in coordinates $\{(g_1,k) : k \in \{g_1+1, \dots, g-1\}$ are either crossing or west-signal multiplier blocks. Similarly, all vertices in cells $\{(k, g) : k \in \{g_1+1, \dots, g-1\}$ are either crossing or north-signal multiplier blocks. Finally, the cell in coordinates $(g_1, g)$ is a west-signal multiplier block. Then, by definition of the blocks, the east output of $v(g_1)$ reaches the north input of $v(g)$.  Similarly, the south output of $v(g_2)$ reaches the west input of $v(g)$.  Therefore the block $v(g)$ outputs the same value than $g$ through both of its outputs. 

Suppose now that $g$ is a gate in the first layer of $C$. Call $g_1$ the incoming neighbor of $g$. If $g$ is a negation gate, it receives the signal from $g_1$ through its west input. With an analogous argument than the one that we discussed before, the south output of the block in $v(g_1)$ reaches the west input of the block of $v(g)$. On the other hand, all cells in $\{(k, g) : k \in \{1, \dots, g-1\}$ are crossing gadgets. Since the vertex in $(0,g)$ is labeled with a  fixed-value block, there is a \emph{false} signal that reaches the north input of $v(g)$. Therefore, the output $v(g)$ will be exactly the value that $v(g_1)$ sends through its south output. Since $v(g_1)$ is a selector, this value will be \emph{false} if $\tilde{u}_{v(g_1)} = 1$ and \emph{true} otherwise. Therefore, $v(g)$ outputs the negation of $g_1$.  An analogous argument follows for all gates with in-degree $1$, including gates in the first layer that are no negations. We deduce that for every gate $g$ of $C$, $C(u)_g = \tilde{C}(\tilde{u})_{v(g)}$. Therefore, a gate $g$ of $C$ is satisfiable if and only if $v(g)$ is satisfiable in $\tilde{C}$. We deduce that \GCSAT is \NP-hard. As we have already observed that \GCSAT is in \NP, we conclude that this problem is \NPC.\\
\end{proof}

\begin{figure}[h!]
\begin{center}

    \begin{minipage}{.45\linewidth}
      \tikzset{every picture/.style={line width=0.75pt}} 
      
      \begin{tikzpicture}[x=0.5pt,y=0.5pt,yscale=-1,xscale=1]
      
      \draw   (369.4,167.74) .. controls (369.4,160.16) and (376.17,154.03) .. (384.52,154.03) .. controls (392.87,154.03) and (399.64,160.16) .. (399.64,167.74) .. controls (399.64,175.31) and (392.87,181.44) .. (384.52,181.44) .. controls (376.17,181.44) and (369.4,175.31) .. (369.4,167.74) -- cycle ;
      \draw    (248.77,52.9) -- (220.42,91.26) ;
      \draw [shift={(219.23,92.87)}, rotate = 306.48] [fill={rgb, 255:red, 0; green, 0; blue, 0 }  ][line width=0.75]  [draw opacity=0] (8.93,-4.29) -- (0,0) -- (8.93,4.29) -- cycle    ;
      
      \draw    (248.77,52.9) -- (275.33,91.85) ;
      \draw [shift={(276.46,93.5)}, rotate = 235.71] [fill={rgb, 255:red, 0; green, 0; blue, 0 }  ][line width=0.75]  [draw opacity=0] (8.93,-4.29) -- (0,0) -- (8.93,4.29) -- cycle    ;
      
      \draw   (233.3,166.47) .. controls (233.3,158.9) and (240.07,152.76) .. (248.42,152.76) .. controls (256.78,152.76) and (263.55,158.9) .. (263.55,166.47) .. controls (263.55,174.04) and (256.78,180.18) .. (248.42,180.18) .. controls (240.07,180.18) and (233.3,174.04) .. (233.3,166.47) -- cycle ;
      
      \draw    (384.87,52.9) -- (356.51,91.26) ;
      \draw [shift={(355.32,92.87)}, rotate = 306.48] [fill={rgb, 255:red, 0; green, 0; blue, 0 }  ][line width=0.75]  [draw opacity=0] (8.93,-4.29) -- (0,0) -- (8.93,4.29) -- cycle    ;
      
      \draw    (384.87,52.9) -- (411.42,91.85) ;
      \draw [shift={(412.55,93.5)}, rotate = 235.71] [fill={rgb, 255:red, 0; green, 0; blue, 0 }  ][line width=0.75]  [draw opacity=0] (8.93,-4.29) -- (0,0) -- (8.93,4.29) -- cycle    ;
      
      \draw    (217.19,107.16) -- (240.61,150.58) ;
      \draw [shift={(241.56,152.34)}, rotate = 241.66] [fill={rgb, 255:red, 0; green, 0; blue, 0 }  ][line width=0.75]  [draw opacity=0] (8.93,-4.29) -- (0,0) -- (8.93,4.29) -- cycle    ;
      
      \draw    (281.92,107.63) -- (255.19,150.02) ;
      \draw [shift={(254.12,151.71)}, rotate = 302.24] [fill={rgb, 255:red, 0; green, 0; blue, 0 }  ][line width=0.75]  [draw opacity=0] (8.93,-4.29) -- (0,0) -- (8.93,4.29) -- cycle    ;
      
      \draw  [fill={rgb, 255:red, 255; green, 255; blue, 255 }  ,fill opacity=1 ] (202.59,107.16) .. controls (202.59,99.85) and (209.13,93.92) .. (217.19,93.92) .. controls (225.25,93.92) and (231.79,99.85) .. (231.79,107.16) .. controls (231.79,114.47) and (225.25,120.39) .. (217.19,120.39) .. controls (209.13,120.39) and (202.59,114.47) .. (202.59,107.16) -- cycle ;
      
      \draw    (353.98,108.42) -- (377.4,151.85) ;
      \draw [shift={(378.35,153.61)}, rotate = 241.66] [fill={rgb, 255:red, 0; green, 0; blue, 0 }  ][line width=0.75]  [draw opacity=0] (8.93,-4.29) -- (0,0) -- (8.93,4.29) -- cycle    ;
      
      \draw    (418.72,108.9) -- (391.98,151.28) ;
      \draw [shift={(390.92,152.97)}, rotate = 302.24] [fill={rgb, 255:red, 0; green, 0; blue, 0 }  ][line width=0.75]  [draw opacity=0] (8.93,-4.29) -- (0,0) -- (8.93,4.29) -- cycle    ;
      
      \draw  [fill={rgb, 255:red, 255; green, 255; blue, 255 }  ,fill opacity=1 ] (402.9,107.63) .. controls (402.9,100.06) and (409.67,93.92) .. (418.02,93.92) .. controls (426.37,93.92) and (433.14,100.06) .. (433.14,107.63) .. controls (433.14,115.2) and (426.37,121.34) .. (418.02,121.34) .. controls (409.67,121.34) and (402.9,115.2) .. (402.9,107.63) -- cycle ;
      
      \draw    (281.92,107.63) -- (308.26,153.14) ;
      \draw [shift={(309.26,154.87)}, rotate = 239.94] [fill={rgb, 255:red, 0; green, 0; blue, 0 }  ][line width=0.75]  [draw opacity=0] (8.93,-4.29) -- (0,0) -- (8.93,4.29) -- cycle    ;
      
      \draw  [fill={rgb, 255:red, 255; green, 255; blue, 255 }  ,fill opacity=1 ] (266.8,107.63) .. controls (266.8,100.06) and (273.57,93.92) .. (281.92,93.92) .. controls (290.28,93.92) and (297.05,100.06) .. (297.05,107.63) .. controls (297.05,115.2) and (290.28,121.34) .. (281.92,121.34) .. controls (273.57,121.34) and (266.8,115.2) .. (266.8,107.63) -- cycle ;
      
      \draw    (349.1,110.32) -- (326.18,147.91) -- (322.5,152.8) ;
      \draw [shift={(321.3,154.4)}, rotate = 307] [fill={rgb, 255:red, 0; green, 0; blue, 0 }  ][line width=0.75]  [draw opacity=0] (8.93,-4.29) -- (0,0) -- (8.93,4.29) -- cycle    ;
      
      \draw  [fill={rgb, 255:red, 255; green, 255; blue, 255 }  ,fill opacity=1 ] (338.69,107.16) .. controls (338.69,99.85) and (345.22,93.92) .. (353.29,93.92) .. controls (361.35,93.92) and (367.88,99.85) .. (367.88,107.16) .. controls (367.88,114.47) and (361.35,120.39) .. (353.29,120.39) .. controls (345.22,120.39) and (338.69,114.47) .. (338.69,107.16) -- cycle ;
      
      \draw   (300.3,169.63) .. controls (300.3,162.06) and (307.07,155.93) .. (315.42,155.93) .. controls (323.78,155.93) and (330.55,162.06) .. (330.55,169.63) .. controls (330.55,177.2) and (323.78,183.34) .. (315.42,183.34) .. controls (307.07,183.34) and (300.3,177.2) .. (300.3,169.63) -- cycle ;
      \draw  [fill={rgb, 255:red, 0; green, 0; blue, 0 }  ,fill opacity=1 ] (243.22,52.9) .. controls (243.22,50.12) and (245.71,47.87) .. (248.77,47.87) .. controls (251.84,47.87) and (254.33,50.12) .. (254.33,52.9) .. controls (254.33,55.69) and (251.84,57.94) .. (248.77,57.94) .. controls (245.71,57.94) and (243.22,55.69) .. (243.22,52.9) -- cycle ;
      \draw  [fill={rgb, 255:red, 0; green, 0; blue, 0 }  ,fill opacity=1 ] (379.31,52.9) .. controls (379.31,50.12) and (381.8,47.87) .. (384.87,47.87) .. controls (387.93,47.87) and (390.42,50.12) .. (390.42,52.9) .. controls (390.42,55.69) and (387.93,57.94) .. (384.87,57.94) .. controls (381.8,57.94) and (379.31,55.69) .. (379.31,52.9) -- cycle ;
      \draw  [fill={rgb, 255:red, 255; green, 255; blue, 255 }  ,fill opacity=1 ] (240.64,22.89) -- (259.93,22.89) -- (259.93,41.19) -- (240.64,41.19) -- cycle ;
      \draw  [fill={rgb, 255:red, 255; green, 255; blue, 255 }  ,fill opacity=1 ] (373.03,22.89) -- (392.32,22.89) -- (392.32,41.19) -- (373.03,41.19) -- cycle ;
      \draw  [fill={rgb, 255:red, 255; green, 255; blue, 255 }  ,fill opacity=1 ] (194.67,71.03) -- (213.96,71.03) -- (213.96,89.33) -- (194.67,89.33) -- cycle ;
      \draw  [fill={rgb, 255:red, 255; green, 255; blue, 255 }  ,fill opacity=1 ] (283.45,69.52) -- (302.74,69.52) -- (302.74,87.82) -- (283.45,87.82) -- cycle ;
      \draw  [fill={rgb, 255:red, 255; green, 255; blue, 255 }  ,fill opacity=1 ] (332.6,68.77) -- (351.89,68.77) -- (351.89,87.07) -- (332.6,87.07) -- cycle ;
      \draw  [fill={rgb, 255:red, 255; green, 255; blue, 255 }  ,fill opacity=1 ] (418.21,70.28) -- (437.5,70.28) -- (437.5,88.58) -- (418.21,88.58) -- cycle ;
      \draw  [fill={rgb, 255:red, 255; green, 255; blue, 255 }  ,fill opacity=1 ] (237.47,186.85) -- (256.76,186.85) -- (256.76,205.15) -- (237.47,205.15) -- cycle ;
      \draw  [fill={rgb, 255:red, 255; green, 255; blue, 255 }  ,fill opacity=1 ] (305.65,188.36) -- (324.94,188.36) -- (324.94,206.66) -- (305.65,206.66) -- cycle ;
      \draw  [fill={rgb, 255:red, 255; green, 255; blue, 255 }  ,fill opacity=1 ] (376.2,186.85) -- (395.49,186.85) -- (395.49,205.15) -- (376.2,205.15) -- cycle ;
      
      \draw (217.72-2,104.47+2) node  [align=left] {$\displaystyle \neg $};
      \draw (283.32-2,105.1+2) node  [align=left] {$\displaystyle \lor $};
      \draw (249.82-2,163.94+2) node  [align=left] {$\displaystyle \land $};
      \draw (353.81-2,104.47+2) node  [align=left] {$\displaystyle \neg $};
      \draw (419.41-2,105.1+2) node  [align=left] {$\displaystyle \lor $};
      \draw (384.6,165.82+2) node  [align=left] {$\displaystyle \land $};
      \draw (314.88,168.04+2) node  [align=left] {$\displaystyle \lor $};
      \draw (249.84,31.17) node  [align=left] {$\displaystyle 1$};
      \draw (382.67,31.04) node  [align=left] {$\displaystyle 2$};
      \draw (204.31,80.18) node  [align=left] {$\displaystyle 3$};
      \draw (293.09,78.67) node  [align=left] {$\displaystyle 4$};
      \draw (342.24,77.92) node  [align=left] {$\displaystyle 5$};
      \draw (427.86,78.43) node  [align=left] {$\displaystyle 6$};
      \draw (247.12,196) node  [align=left] {$\displaystyle 7$};
      \draw (315.29,197.51) node  [align=left] {$\displaystyle 8$};
      \draw (385.84,196) node  [align=left] {$\displaystyle 9$};
      \end{tikzpicture}
    \end{minipage}
    \begin{minipage}{.49\linewidth}
      \centering
      \resizebox{\textwidth}{!}{%
        \begin{tabular}{lllllllllll}
          & $0$                      & $1$                           & $2$                           & $3$                         & $4$                           & $5$                           & $6$                           & $7$                           & $8$                         & $9$                           \\ \cline{2-11} 
          \multicolumn{1}{l|}{$0$} & \multicolumn{1}{l|}{$0$} & \multicolumn{1}{l|}{$0$}      & \multicolumn{1}{l|}{$0$}      & \multicolumn{1}{l|}{$0$}    & \multicolumn{1}{l|}{$0$}      & \multicolumn{1}{l|}{$0$}      & \multicolumn{1}{l|}{$0$}      & \multicolumn{1}{l|}{$0$}      & \multicolumn{1}{l|}{$0$}    & \multicolumn{1}{l|}{$0$}      \\ \cline{2-11} 
          \multicolumn{1}{l|}{$1$} & \multicolumn{1}{l|}{$0$} & \multicolumn{1}{l|}{$S$}      & \multicolumn{1}{l|}{$C$}      & \multicolumn{1}{l|}{$C$}    & \multicolumn{1}{l|}{$\top$}   & \multicolumn{1}{l|}{$C$}      & \multicolumn{1}{l|}{$\top$}   & \multicolumn{1}{l|}{$C$}      & \multicolumn{1}{l|}{$C$}    & \multicolumn{1}{l|}{$C$}      \\ \cline{2-11} 
          \multicolumn{1}{l|}{$2$} & \multicolumn{1}{l|}{$0$} & \multicolumn{1}{l|}{$C$}      & \multicolumn{1}{l|}{$S$}      & \multicolumn{1}{l|}{$C$}    & \multicolumn{1}{l|}{$C$}      & \multicolumn{1}{l|}{$C$}      & \multicolumn{1}{l|}{$C$}      & \multicolumn{1}{l|}{$C$}      & \multicolumn{1}{l|}{$C$}    & \multicolumn{1}{l|}{$C$}      \\ \cline{2-11} 
          \multicolumn{1}{l|}{$3$} & \multicolumn{1}{l|}{$0$} & \multicolumn{1}{l|}{$\vdash$} & \multicolumn{1}{l|}{$C$}      & \multicolumn{1}{l|}{$\vee$} & \multicolumn{1}{l|}{$C$}      & \multicolumn{1}{l|}{$C$}      & \multicolumn{1}{l|}{$C$}      & \multicolumn{1}{l|}{$\top$}   & \multicolumn{1}{l|}{$C$}    & \multicolumn{1}{l|}{$C$}      \\ \cline{2-11} 
          \multicolumn{1}{l|}{$4$} & \multicolumn{1}{l|}{$0$} & \multicolumn{1}{l|}{$C$}      & \multicolumn{1}{l|}{$C$}      & \multicolumn{1}{l|}{$C$}    & \multicolumn{1}{l|}{$\vee$}   & \multicolumn{1}{l|}{$C$}      & \multicolumn{1}{l|}{$C$}      & \multicolumn{1}{l|}{$C$}      & \multicolumn{1}{l|}{$\top$} & \multicolumn{1}{l|}{$C$}      \\ \cline{2-11} 
          \multicolumn{1}{l|}{$5$} & \multicolumn{1}{l|}{$0$} & \multicolumn{1}{l|}{$C$}      & \multicolumn{1}{l|}{$\vdash$} & \multicolumn{1}{l|}{$C$}    & \multicolumn{1}{l|}{$C$}      & \multicolumn{1}{l|}{$\vee$}   & \multicolumn{1}{l|}{$C$}      & \multicolumn{1}{l|}{$C$}      & \multicolumn{1}{l|}{$C$}    & \multicolumn{1}{l|}{$\top$}   \\ \cline{2-11} 
          \multicolumn{1}{l|}{$6$} & \multicolumn{1}{l|}{$0$} & \multicolumn{1}{l|}{$C$}      & \multicolumn{1}{l|}{$C$}      & \multicolumn{1}{l|}{$C$}    & \multicolumn{1}{l|}{$C$}      & \multicolumn{1}{l|}{$C$}      & \multicolumn{1}{l|}{$\vee$}   & \multicolumn{1}{l|}{$C$}      & \multicolumn{1}{l|}{$C$}    & \multicolumn{1}{l|}{$C$}      \\ \cline{2-11} 
          \multicolumn{1}{l|}{$7$} & \multicolumn{1}{l|}{$0$} & \multicolumn{1}{l|}{$C$}      & \multicolumn{1}{l|}{$C$}      & \multicolumn{1}{l|}{$C$}    & \multicolumn{1}{l|}{$\vdash$} & \multicolumn{1}{l|}{$C$}      & \multicolumn{1}{l|}{$C$}      & \multicolumn{1}{l|}{$\wedge$} & \multicolumn{1}{l|}{$C$}    & \multicolumn{1}{l|}{$C$}      \\ \cline{2-11} 
          \multicolumn{1}{l|}{$8$} & \multicolumn{1}{l|}{$0$} & \multicolumn{1}{l|}{$C$}      & \multicolumn{1}{l|}{$C$}      & \multicolumn{1}{l|}{$C$}    & \multicolumn{1}{l|}{$C$}      & \multicolumn{1}{l|}{$\vdash$} & \multicolumn{1}{l|}{$C$}      & \multicolumn{1}{l|}{$C$}      & \multicolumn{1}{l|}{$\vee$} & \multicolumn{1}{l|}{$C$}      \\ \cline{2-11} 
          \multicolumn{1}{l|}{$9$} & \multicolumn{1}{l|}{$0$} & \multicolumn{1}{l|}{$C$}      & \multicolumn{1}{l|}{$C$}      & \multicolumn{1}{l|}{$C$}    & \multicolumn{1}{l|}{$C$}      & \multicolumn{1}{l|}{$C$}      & \multicolumn{1}{l|}{$\vdash$} & \multicolumn{1}{l|}{$C$}      & \multicolumn{1}{l|}{$C$}    & \multicolumn{1}{l|}{$\wedge$} \\ \cline{2-11} 
        \end{tabular}%
      }
    \end{minipage}
\caption{Example of the embedding of a circuit in the grid used for the proof of Theorem \ref{teo:GCSATNP}. In the left-hand side it is shown the original circuit $C$ and in the right-hand side it is shown the correspondent circuit $\tilde{C}$ embedded in the grid. In the left-hand side: the symbol $\bullet$ represents the inputs of the circuit. In the right-hand side: $\vdash$ and $\top$ are the signal multipliers, the symbol $S$ represents the selector gadget and $C$ represents the crossing gadget. }
\label{fig:ejemploreduccionGCSAT}
\end{center}
\end{figure}
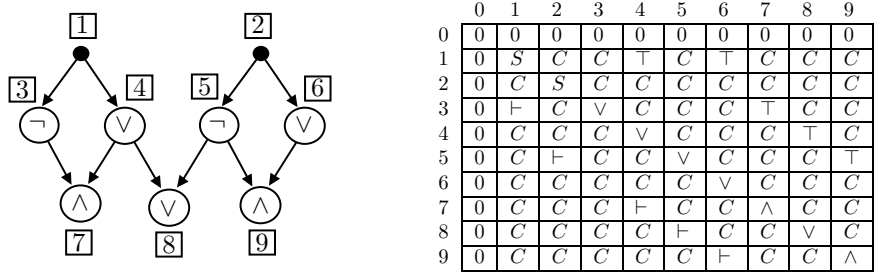

\subsection{Restricted Grid-embedded Boolean circuits}
 
In the definition of grid-embedded Boolean circuits, we consider that each cell of the grid contains a block, where the set of blocks is given by: disjunction, conjunction, crossing, fixed-value, signal-multiplier and selector. In this subsection we show that \GCSAT is \NPC even when we restrict the set blocks to only conjunction, disjunction, fixed-value and selector blocks. In other words, we show that it is possible to simulate crossing and signal-multiplier blocks using a combination of the other blocks. 
 
 In order to achieve our task, we are going to simulate each block by a \emph{meta-block}. A meta-block is a square grid of blocks of dimensions $\Delta \times \Delta$, where each block is a conjunction, disjunction, fixed-value or selector block. Each meta-block will contain in its north and east borders two designated blocks called \emph{input blocks}, and in its south and east borders two designated blocks called \emph{output blocks}. 
 
 Our goal is to simulate each one of the blocks defined in the previous section with one (or, in fact, a set of) meta-blocks, defining  disjunction, conjunction, crossing, fixed-value, signal-multiplier and selector meta-blocks. Then, we will simulate any grid-embedded Boolean circuit by a restricted one, taking a larger grid of blocks, and replacing each block by a meta-block.
 
 Roughly, all the gadgets are straightforwardly deduced from the previous block with exception of the crossing gadget. Thus, we first show the construction of a structure that will allow us to simulate the latter gadget. In fact, we will start by the construcction of the gadget exhibited in Figure \ref{fig:crossingsim} that we call \emph{crossing gadget} which \textit{a posteriori} will be a subpart of our meta-block gadget.
 
 \begin{figure}[h]
    \centering
   
\begin{minipage}{.45\linewidth}
  \resizebox{\textwidth}{!}{%
    \begin{tabular}{lllllllll}
      & $0$                         & $1$                         & $2$                           & $3$                           & $4$                           & $5$                         & $6$                           & $7$                         \\ \cline{2-9} 
      \multicolumn{1}{l|}{$0$} & \multicolumn{1}{l|}{$0$}    & \multicolumn{1}{l|}{$0$}    & \multicolumn{1}{l|}{$0$}      & \multicolumn{1}{l|}{$\vee$}   & \multicolumn{1}{l|}{$0$}      & \multicolumn{1}{l|}{$0$}    & \multicolumn{1}{l|}{$0$}      & \multicolumn{1}{l|}{$0$}    \\ \cline{2-9} 
      \multicolumn{1}{l|}{$1$} & \multicolumn{1}{l|}{$0$}    & \multicolumn{1}{l|}{$0$}    & \multicolumn{1}{l|}{$0$}      & \multicolumn{1}{l|}{$\vee$}   & \multicolumn{1}{l|}{$S$}      & \multicolumn{1}{l|}{$\vee$} & \multicolumn{1}{l|}{$\vee$}   & \multicolumn{1}{l|}{$0$}    \\ \cline{2-9} 
      \multicolumn{1}{l|}{$2$} & \multicolumn{1}{l|}{$0$}    & \multicolumn{1}{l|}{$0$}    & \multicolumn{1}{l|}{$0$}      & \multicolumn{1}{l|}{$\vee$}   & \multicolumn{1}{l|}{$\wedge$} & \multicolumn{1}{l|}{$0$}    & \multicolumn{1}{l|}{$\vee$}   & \multicolumn{1}{l|}{$0$}    \\ \cline{2-9} 
      \multicolumn{1}{l|}{$3$} & \multicolumn{1}{l|}{$\vee$} & \multicolumn{1}{l|}{$\vee$} & \multicolumn{1}{l|}{$\vee$}   & \multicolumn{1}{l|}{$\wedge$} & \multicolumn{1}{l|}{$\vee$}   & \multicolumn{1}{l|}{$0$}    & \multicolumn{1}{l|}{$\vee$}   & \multicolumn{1}{l|}{$0$}    \\ \cline{2-9} 
      \multicolumn{1}{l|}{$4$} & \multicolumn{1}{l|}{$0$}    & \multicolumn{1}{l|}{$S$}    & \multicolumn{1}{l|}{$\wedge$} & \multicolumn{1}{l|}{$\vee$}   & \multicolumn{1}{l|}{$\vee$}   & \multicolumn{1}{l|}{$\vee$} & \multicolumn{1}{l|}{$\wedge$} & \multicolumn{1}{l|}{$\vee$} \\ \cline{2-9} 
      \multicolumn{1}{l|}{$5$} & \multicolumn{1}{l|}{$0$}    & \multicolumn{1}{l|}{$\vee$} & \multicolumn{1}{l|}{$0$}      & \multicolumn{1}{l|}{$0$}      & \multicolumn{1}{l|}{$\vee$}   & \multicolumn{1}{l|}{$0$}    & \multicolumn{1}{l|}{$0$}      & \multicolumn{1}{l|}{$0$}    \\ \cline{2-9} 
      \multicolumn{1}{l|}{$6$} & \multicolumn{1}{l|}{$0$}    & \multicolumn{1}{l|}{$\vee$} & \multicolumn{1}{l|}{$\vee$}   & \multicolumn{1}{l|}{$\vee$}   & \multicolumn{1}{l|}{$\wedge$} & \multicolumn{1}{l|}{$0$}    & \multicolumn{1}{l|}{$0$}      & \multicolumn{1}{l|}{$0$}    \\ \cline{2-9} 
      \multicolumn{1}{l|}{$7$} & \multicolumn{1}{l|}{$0$}    & \multicolumn{1}{l|}{$0$}    & \multicolumn{1}{l|}{$0$}      & \multicolumn{1}{l|}{$0$}      & \multicolumn{1}{l|}{$\vee$}   & \multicolumn{1}{l|}{$0$}    & \multicolumn{1}{l|}{$0$}      & \multicolumn{1}{l|}{$0$}    \\ \cline{2-9} 
    \end{tabular}%
  }

\end{minipage}\hspace{1cm}
\begin{minipage}{.45\linewidth}   
  
  \begin{tikzpicture}
  \node[circle,draw,minimum size=0.5cm,inner sep=0pt] (v1) at (-0.5,0) {$a$};
  \node[circle,draw,minimum size=0.5cm,inner sep=0pt] (v3) at (1.5,0) {$b$};
  \node[circle,draw,minimum size=0.5cm,inner sep=0pt] (v4) at (-1.5,-1) {$\wedge$};
  \node[circle,draw,minimum size=0.5cm,inner sep=0pt] (v5) at (2.5,-1) {$\wedge$};
  \node[circle,draw,minimum size=0.5cm,inner sep=0pt] (v2) at (0.5,-1) {$\wedge$};
  \node[circle,draw,minimum size=0.5cm] (v7) at (-1.5,-2.5) {$s$};
  \node[circle,draw,minimum size=0.5cm,inner sep=0pt] (v6) at (-0.5,-2) {$\vee$};
  \node[circle,draw,minimum size=0.5cm] (v9) at (2.5,-2.5) {$s$};
  \node[circle,draw,minimum size=0.5cm,inner sep=0pt] (v8) at (-1.5,-4) {$\wedge$};
  \node[circle,draw,minimum size=0.5cm,inner sep=0pt] (v10) at (2.5,-4) {$\wedge$};
  \node[circle,draw,minimum size=0.5cm,inner sep=0pt] (v11) at (-0.5,-5) {$b'$};
  \node[circle,draw,minimum size=0.5cm,inner sep=0pt] (v12) at (1.5,-5) {$a'$};
  \node[circle,draw,minimum size=0.5cm,inner sep=0pt] (v13) at (1.5,-2) {$\vee$};
  \node[circle,draw,minimum size=0.5cm,inner sep=0pt] (v14) at (0.5,-3) {$\vee$};
  
  \node[scale=0.85,xshift=-0.0cm,yshift=0.5cm] at (v1)  {$(3,0)$};
  \node[scale=0.85,xshift=-0.0cm,yshift=0.5cm,] at (v3)  {$(0,3)$};
  \node[scale=0.85,xshift=-0.0cm,yshift=0.5cm,] at (v4)  {$(2,4)$};
  \node[scale=0.85,xshift=-0.0cm,yshift=0.5cm,] at (v5)  {$(4,2)$};
  \node[scale=0.85,xshift=-0.0cm,yshift=0.5cm,] at (v2)  {$(3,3)$};
  \node[scale=0.85,xshift=-0.5cm,yshift=0.5cm,] at (v7)  {$(4,1)$};
  \node[scale=0.85,xshift=-0.0cm,yshift=0.5cm,] at (v6)  {$(4,3)$};
  \node[scale=0.85,xshift=0.5cm,yshift=0.5cm,] at (v9)  {$(1,4)$};
  \node[scale=0.85,xshift=-0.5cm,yshift=0.5cm,] at (v8)  {$(6,4)$};
  \node[scale=0.85,xshift=0.5cm,yshift=0.5cm,] at (v10)  {$(4,6)$};
  \node[scale=0.85,xshift=-0.0cm,yshift=0.5cm,] at (v11)  {$(7,4)$};
  \node[scale=0.85,xshift=-0.0cm,yshift=0.5cm,] at (v12)  {$(4,7)$};
  \node[scale=0.85,xshift=-0.0cm,yshift=0.5cm,] at (v13)  {$(3,4)$};
  \node[scale=0.85,xshift=-0.0cm,yshift=0.5cm,] at (v14)  {$(4,4)$};

  \draw [->](v1) -- (v2);
  \draw [->](v3) -- (v2);
  \draw [->](v1) -- (v4);
  \draw [->](v3) -- (v5);
  
  \draw [->](v4) -- (v6);
  \draw [->](v4);
  \draw [->](v7) -- (v4);
  \draw [->](v7) -- (v8);
  \draw [->](v9) -- (v5);
  \draw [->](v9) -- (v10);
  \draw [->](v14) -- (v8);
  \draw [->](v14) -- (v10);
  \draw [->](v2) -- (v6);
  \draw [->](v8) -- (v11);
  \draw [->](v10) -- (v12);
  
  \draw[->] (v2) -- (v13);
  
  \draw [->](v5) -- (v13);
  
  \draw [->](v6) -- (v14);
  \draw [->](v13) -- (v14);
  \end{tikzpicture}

\end{minipage}
      \caption{Left: A crossing gadget. Right: A directed graph representing a crossing gadget.}
       \label{fig:crossingsim} 
\end{figure}
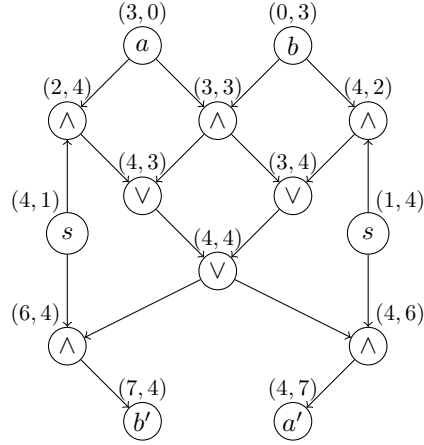 

\begin{lemma}\label{lem:selectormetablock}
Let $a, b\in \{0,1\}$ two input values of the input blocks of the  crossing gadget, and let $a', b'$ the values of the output blocks. \begin{itemize}
\item[a)] For every  truth-assignment of $C$,  $b'\leq b$ and $a'\leq a$. 
\item[b)] There exists a truth-assignment of $C$ such that $a' = a$ and $b'=b$.
\end{itemize}
\end{lemma}

\begin{proof}
  Let $a=b=1$. A truth assignment $u$ of $C$ has two coordinates, one for each selector. Let us call $s_1$ the selector of coordinates $(4,1)$ (vertex numbered $6$ in the graph of Figure~\ref{fig:crossingsim}), and $s_2$ the selector of coordinates  coordinates $(1,4)$ (vertex numbered $8$ in the graph of Figure~\ref{fig:crossingsim}).  Without loss of generality,  the first coordinate of $u$ corresponds to the choice of $s_1$ and the second coordinate the choice of $s_2$. More precisely, for $i\in \{1,2\}$, if $u_i = 1$ then $s_i$ is an east-selector, and if $u_i=0$ then $s_i$ is a south-selector. 
  
  Consider now the vertices in coordinates $(4,2)$ and $(6, 4)$, which are conjunction blocks.  Call them $v_1$ and $v_2$. Observe that the construction satisfies that at least one of these two blocks has to output \emph{false} signals. Indeed, since the selector $s_1$ can only output a true-signal through either the east or the south sides, at least one of these blocks has to receive a false signal, and therefore output false.  More precisely, if $s_1$ is a east-selector ($u_1=1$) then $v_2$ always outputs a false-signal, and if if $s_1$  is a south-selector ($u_1=0$) then $v_1$ always outputs false.  Similarly, we call $v_3$ and $v_4$ the blocks with coordinates $(2, 4)$ and $(4, 6)$. If $s_2$ is a south-selector ($u_2=0$) then $v_4$ always outputs a false-signal, and if $s_2$ is a east-selector ($u_2=1$) $v_3$ always outputs a false-signal. 
  
  Then, there are four options for $u$.  When $u=(0,0)$ we have that $v_1$ and $v_4$ always output a false-signal, so $a'=0$ and $b'= (a\wedge b) \vee b$.  When $u=(0,1)$ we have that $v_1$ and $v_3$ always output a false-signal, then $a' = b'= a\wedge b$.  When $u=(1,0)$ we have that $v_2$ and $v_4$ output false-signals, then $a' = b' = 0$.  Finally, when $u=(1,1)$ we have that $v_2$ and $v_3$ output false-signals, then $a'= (a\wedge b) \wedge a = a$ and $b = 0$. In all cases $a' \leq a$ and $b' \leq b$. Moreover, if we choose $u = (\neg b,a)$ we obtain that $a'=a$ and $b'=b$. 
   
  \begin{table}[h]
    \centering
    $    
     \begin{array}{|c|c|c|c||c|c|}
     \hline
       a & b & s_1 & s_2 & a'& b'\\
      \hline
       0 & 0 & \downarrow & \downarrow & 0 & 0\\
      \hline 
       0 & 0 & \downarrow & \rightarrow  & 0 & 0\\
      \hline 
      \rowcolor{gray!50}
       0 & 0 & \rightarrow   & \downarrow & 0 & 0\\
      \hline 
       0 & 0 & \rightarrow  & \rightarrow   & 0 & 0\\       
      \hline 
       \hline 
      \rowcolor{gray!50}
       0 & 1 & \downarrow & \downarrow & 0 & 1\\
      \hline 
       0 & 1 & \downarrow & \rightarrow   & 0 & 0\\
      \hline 
      
       0 & 1 &\rightarrow  & \downarrow & 0 & 0\\
      \hline 
       0 & 1 & \rightarrow  &\rightarrow   & 0 & 0\\
      \hline       
     \end{array}$
     \hfil
     $\begin{array}{|c|c|c|c||c|c|}
     \hline
       a & b & s_1 & s_2 & a'& b'\\
      \hline
       1 & 0 & \downarrow & \downarrow & 0 & 0\\
      \hline 
    
       1 & 0 & \downarrow & \rightarrow   & 0 & 0\\
      \hline
       1 & 0 & \rightarrow   & \downarrow & 0 & 0\\
      \hline 
           \rowcolor{gray!50}    
       1 & 0 &\rightarrow   &\rightarrow   & 1& 0\\
      \hline 
       \hline 
     
       1 & 1 & \downarrow & \downarrow & 0 & 1\\
      \hline 
       \rowcolor{gray!50}
       1 & 1 & \downarrow & \rightarrow   & 1 & 1\\
      \hline 
       1 & 1 & \rightarrow  & \downarrow & 0 & 0\\
      \hline 
       1 & 1 &\rightarrow & \rightarrow   & 1 & 0\\
      \hline 
     \end{array}
     $
     \caption{
      Table of combinations between the inputs $a$ and $b$ the selectors $s_1$ and $s_2$, and the possible outputs $a'$ and $b'$. 
      The symbols $\rightarrow$ and $\downarrow $ mean that the selector is an east ($u_i = 1$) or south ($u_i=0$) selector, respectively. 
      The rows marked in gray correspond to the combinations $(\neg b, a)$, in which the crossing gadget acts as a crossing block without deleting signals.}
     \label{tbl:label}
   \end{table} 
\end{proof}

This  crossing gadget implies, roughly, that combining disjunction, conjunction and selector blocks it is possible to simulate a crossing block. Unfortunately, this construction involves two difficulties that we have to handle in order to correctly define our simulation. The first difficulty is the fact that it \emph{ do not always outputs the right values}. More precisely, there exist choices of selectors in which the outputs of our crossing gadget are different from the values that the original crossing block would have outputted. Fortunately this difficulty will not have a real effect in our main result (that is to say the one that says that restricted grid Boolean satisfiability problem is \NPC), because condition (a) of Lemma \ref{lem:selectormetablock} assures that the outputs of each crossing gadget can not produce \emph{false positives}. In other words, the signal outputted by this gadget can not make satisfiable an unsatisfiable vertex. On the other hand, condition (b) of  Lemma \ref{lem:selectormetablock} assures that, if a vertex $v$ is satisfiable, then there exists a choice of selectors in the gadget that produces the correct simulation of the crossing block. 
 
 The second difficulty, is that the input blocks of the simplified crossing gadget are not \emph{aligned} with the output blocks. More precisely, on the one hand the north and west input blocks are, respectively, in coordinates $(0,3)$ and $(3,0)$. On the other hand the east and south blocks are, respectively, in coordinates $(4,7)$ and $(7, 4)$. This implies that if we place two crossing gadgets one adjacent to the other, we will not be able to put the inputs of the second one with the outputs of the first one. This difficulty can also be handled, but the solution is slightly more involved.  
 
 Before giving more details on this issue, we define the gadgets for every other block, all having dimensions $8\times 8$, input blocks in coordinates $(3,0)$ and $(0,3)$, and output blocks in coordinates $(4, 7)$ and $(4, 5)$, as shown in Figure \ref{fig:blockgadgets}. We remark that each gadget only contain conjunction, disjunction, selector and fixed-value blocks. 
 
\begin{figure}  
\centering
\tikzset{every picture/.style={line width=0.75pt}} 

 \caption{Gadgets simulating the disjunction block (top-left), the conjunction block (top-right), the selector block (middle-left), the fixed-value block (middle-right), the west-signal multiplier block (bottom-left) and the north-signal multiplier block (bottom-right).}
 \label{fig:blockgadgets}
\end{figure}

We also define one more gadget called \emph{wire-gadget}. Let $i_1, i_2, j_1, j_2$ such that $0 \leq i_1 \leq i_2 \leq 7$ and $0 \leq j_1 \leq j_2 \leq 7$.  A wire gadget with input in $(i_1, j_1)$ and output in $(i_2, j_2)$ is a pattern of $8\times 8$ blocks, all being disjunction of fixed-value blocks.  This gadget has one input block in coordinates $(i_1, j_1)$ that is either the north or the west side of the pattern (therefore $i_1$ or $j_1$ equals $0$), and one output block in coordinates $(i_2, j_2)$, either in the east or south side of the pattern (hence $i_2$ or $j_2$ equals $7$). The wire gadget consists in a shortest directed path of disjunction blocks that starts from the input block, and finishes in the output block. All the other blocks in the pattern are fixed-value blocks with value $0$.  See Figure \ref{fig:wiregadgets} for an example of the construction. 
 
\begin{figure}[h] 
\centering
\tikzset{every picture/.style={line width=0.75pt}} 
\begin{tikzpicture}[x=0.75pt,y=0.75pt,yscale=-1,xscale=1]
\draw   (80,73.18) -- (103.21,73.18) -- (103.21,96.35) -- (80,96.35) -- cycle ;
\draw   (103.21,73.18) -- (126.42,73.18) -- (126.42,96.35) -- (103.21,96.35) -- cycle ;
\draw   (126.42,73.18) -- (149.63,73.18) -- (149.63,96.35) -- (126.42,96.35) -- cycle ;
\draw   (149.63,73.18) -- (172.84,73.18) -- (172.84,96.35) -- (149.63,96.35) -- cycle ;
\draw   (172.84,73.18) -- (196.05,73.18) -- (196.05,96.35) -- (172.84,96.35) -- cycle ;
\draw   (80,96.35) -- (103.21,96.35) -- (103.21,119.53) -- (80,119.53) -- cycle ;
\draw   (103.21,96.35) -- (126.42,96.35) -- (126.42,119.53) -- (103.21,119.53) -- cycle ;
\draw   (126.42,96.35) -- (149.63,96.35) -- (149.63,119.53) -- (126.42,119.53) -- cycle ;
\draw   (149.63,96.35) -- (172.84,96.35) -- (172.84,119.53) -- (149.63,119.53) -- cycle ;
\draw   (172.84,96.35) -- (196.05,96.35) -- (196.05,119.53) -- (172.84,119.53) -- cycle ;
\draw   (196.05,73.18) -- (219.26,73.18) -- (219.26,96.35) -- (196.05,96.35) -- cycle ;
\draw   (196.05,96.35) -- (219.26,96.35) -- (219.26,119.53) -- (196.05,119.53) -- cycle ;
\draw   (80,119.53) -- (103.21,119.53) -- (103.21,142.71) -- (80,142.71) -- cycle ;
\draw   (103.21,119.53) -- (126.42,119.53) -- (126.42,142.71) -- (103.21,142.71) -- cycle ;
\draw   (126.42,119.53) -- (149.63,119.53) -- (149.63,142.71) -- (126.42,142.71) -- cycle ;
\draw   (149.63,119.53) -- (172.84,119.53) -- (172.84,142.71) -- (149.63,142.71) -- cycle ;
\draw   (172.84,119.53) -- (196.05,119.53) -- (196.05,142.71) -- (172.84,142.71) -- cycle ;
\draw   (196.05,119.53) -- (219.26,119.53) -- (219.26,142.71) -- (196.05,142.71) -- cycle ;
\draw   (80,142.71) -- (103.21,142.71) -- (103.21,165.89) -- (80,165.89) -- cycle ;
\draw   (103.21,142.71) -- (126.42,142.71) -- (126.42,165.89) -- (103.21,165.89) -- cycle ;
\draw   (196.05,165.89) -- (219.26,165.89) -- (219.26,189.06) -- (196.05,189.06) -- cycle ;
\draw   (126.42,142.71) -- (149.63,142.71) -- (149.63,165.89) -- (126.42,165.89) -- cycle ;
\draw   (149.63,142.71) -- (172.84,142.71) -- (172.84,165.89) -- (149.63,165.89) -- cycle ;
\draw   (172.84,142.71) -- (196.05,142.71) -- (196.05,165.89) -- (172.84,165.89) -- cycle ;
\draw   (80,165.89) -- (103.21,165.89) -- (103.21,189.06) -- (80,189.06) -- cycle ;
\draw   (196.05,142.71) -- (219.26,142.71) -- (219.26,165.89) -- (196.05,165.89) -- cycle ;
\draw   (103.21,189.06) -- (126.42,189.06) -- (126.42,212.24) -- (103.21,212.24) -- cycle ;
\draw   (149.63,165.89) -- (172.84,165.89) -- (172.84,189.06) -- (149.63,189.06) -- cycle ;
\draw   (103.21,165.89) -- (126.42,165.89) -- (126.42,189.06) -- (103.21,189.06) -- cycle ;
\draw   (126.42,165.89) -- (149.63,165.89) -- (149.63,189.06) -- (126.42,189.06) -- cycle ;
\draw   (172.84,165.89) -- (196.05,165.89) -- (196.05,189.06) -- (172.84,189.06) -- cycle ;
\draw   (126.42,189.06) -- (149.63,189.06) -- (149.63,212.24) -- (126.42,212.24) -- cycle ;
\draw   (80,189.06) -- (103.21,189.06) -- (103.21,212.24) -- (80,212.24) -- cycle ;
\draw   (149.63,189.06) -- (172.84,189.06) -- (172.84,212.24) -- (149.63,212.24) -- cycle ;
\draw   (196.05,189.06) -- (219.26,189.06) -- (219.26,212.24) -- (196.05,212.24) -- cycle ;
\draw   (172.84,189.06) -- (196.05,189.06) -- (196.05,212.24) -- (172.84,212.24) -- cycle ;
\draw   (103.21,212.24) -- (126.42,212.24) -- (126.42,235.42) -- (103.21,235.42) -- cycle ;
\draw   (80,212.24) -- (103.21,212.24) -- (103.21,235.42) -- (80,235.42) -- cycle ;
\draw   (219.26,73.18) -- (242.46,73.18) -- (242.46,96.35) -- (219.26,96.35) -- cycle ;
\draw   (219.26,96.35) -- (242.46,96.35) -- (242.46,119.53) -- (219.26,119.53) -- cycle ;
\draw   (219.26,119.53) -- (242.46,119.53) -- (242.46,142.71) -- (219.26,142.71) -- cycle ;
\draw   (219.26,142.71) -- (242.46,142.71) -- (242.46,165.89) -- (219.26,165.89) -- cycle ;
\draw   (219.26,189.06) -- (242.46,189.06) -- (242.46,212.24) -- (219.26,212.24) -- cycle ;
\draw   (219.26,165.89) -- (242.46,165.89) -- (242.46,189.06) -- (219.26,189.06) -- cycle ;
\draw   (126.42,212.24) -- (149.63,212.24) -- (149.63,235.42) -- (126.42,235.42) -- cycle ;
\draw   (149.63,212.24) -- (172.84,212.24) -- (172.84,235.42) -- (149.63,235.42) -- cycle ;
\draw   (196.05,212.24) -- (219.26,212.24) -- (219.26,235.42) -- (196.05,235.42) -- cycle ;
\draw   (172.84,212.24) -- (196.05,212.24) -- (196.05,235.42) -- (172.84,235.42) -- cycle ;
\draw   (219.26,212.24) -- (242.46,212.24) -- (242.46,235.42) -- (219.26,235.42) -- cycle ;
\draw   (242.46,212.24) -- (265.67,212.24) -- (265.67,235.42) -- (242.46,235.42) -- cycle ;
\draw   (242.46,189.06) -- (265.67,189.06) -- (265.67,212.24) -- (242.46,212.24) -- cycle ;
\draw   (242.46,165.89) -- (265.67,165.89) -- (265.67,189.06) -- (242.46,189.06) -- cycle ;
\draw   (242.46,142.71) -- (265.67,142.71) -- (265.67,165.89) -- (242.46,165.89) -- cycle ;
\draw   (242.46,119.53) -- (265.67,119.53) -- (265.67,142.71) -- (242.46,142.71) -- cycle ;
\draw   (242.46,96.35) -- (265.67,96.35) -- (265.67,119.53) -- (242.46,119.53) -- cycle ;
\draw   (242.46,73.18) -- (265.67,73.18) -- (265.67,96.35) -- (242.46,96.35) -- cycle ;
\draw   (80,50) -- (103.21,50) -- (103.21,73.18) -- (80,73.18) -- cycle ;
\draw   (103.21,50) -- (126.42,50) -- (126.42,73.18) -- (103.21,73.18) -- cycle ;
\draw   (126.42,50) -- (149.63,50) -- (149.63,73.18) -- (126.42,73.18) -- cycle ;
\draw   (149.63,50) -- (172.84,50) -- (172.84,73.18) -- (149.63,73.18) -- cycle ;
\draw   (172.84,50) -- (196.05,50) -- (196.05,73.18) -- (172.84,73.18) -- cycle ;
\draw   (196.05,50) -- (219.26,50) -- (219.26,73.18) -- (196.05,73.18) -- cycle ;
\draw   (219.26,50) -- (242.46,50) -- (242.46,73.18) -- (219.26,73.18) -- cycle ;
\draw   (242.46,50) -- (265.67,50) -- (265.67,73.18) -- (242.46,73.18) -- cycle ;
\draw   (295.67,73.18) -- (318.88,73.18) -- (318.88,96.35) -- (295.67,96.35) -- cycle ;
\draw   (318.88,73.18) -- (342.09,73.18) -- (342.09,96.35) -- (318.88,96.35) -- cycle ;
\draw   (342.09,73.18) -- (365.3,73.18) -- (365.3,96.35) -- (342.09,96.35) -- cycle ;
\draw   (365.3,73.18) -- (388.51,73.18) -- (388.51,96.35) -- (365.3,96.35) -- cycle ;
\draw   (388.51,73.18) -- (411.72,73.18) -- (411.72,96.35) -- (388.51,96.35) -- cycle ;
\draw   (295.67,96.35) -- (318.88,96.35) -- (318.88,119.53) -- (295.67,119.53) -- cycle ;
\draw   (318.88,96.35) -- (342.09,96.35) -- (342.09,119.53) -- (318.88,119.53) -- cycle ;
\draw   (342.09,96.35) -- (365.3,96.35) -- (365.3,119.53) -- (342.09,119.53) -- cycle ;
\draw   (365.3,96.35) -- (388.51,96.35) -- (388.51,119.53) -- (365.3,119.53) -- cycle ;
\draw   (388.51,96.35) -- (411.72,96.35) -- (411.72,119.53) -- (388.51,119.53) -- cycle ;
\draw   (411.72,73.18) -- (434.93,73.18) -- (434.93,96.35) -- (411.72,96.35) -- cycle ;
\draw   (411.72,96.35) -- (434.93,96.35) -- (434.93,119.53) -- (411.72,119.53) -- cycle ;
\draw   (295.67,119.53) -- (318.88,119.53) -- (318.88,142.71) -- (295.67,142.71) -- cycle ;
\draw   (318.88,119.53) -- (342.09,119.53) -- (342.09,142.71) -- (318.88,142.71) -- cycle ;
\draw   (342.09,119.53) -- (365.3,119.53) -- (365.3,142.71) -- (342.09,142.71) -- cycle ;
\draw   (365.3,119.53) -- (388.51,119.53) -- (388.51,142.71) -- (365.3,142.71) -- cycle ;
\draw   (388.51,119.53) -- (411.72,119.53) -- (411.72,142.71) -- (388.51,142.71) -- cycle ;
\draw   (411.72,119.53) -- (434.93,119.53) -- (434.93,142.71) -- (411.72,142.71) -- cycle ;
\draw   (295.67,142.71) -- (318.88,142.71) -- (318.88,165.89) -- (295.67,165.89) -- cycle ;
\draw   (318.88,142.71) -- (342.09,142.71) -- (342.09,165.89) -- (318.88,165.89) -- cycle ;
\draw   (411.72,165.89) -- (434.93,165.89) -- (434.93,189.06) -- (411.72,189.06) -- cycle ;
\draw   (342.09,142.71) -- (365.3,142.71) -- (365.3,165.89) -- (342.09,165.89) -- cycle ;
\draw   (365.3,142.71) -- (388.51,142.71) -- (388.51,165.89) -- (365.3,165.89) -- cycle ;
\draw   (388.51,142.71) -- (411.72,142.71) -- (411.72,165.89) -- (388.51,165.89) -- cycle ;
\draw   (295.67,165.89) -- (318.88,165.89) -- (318.88,189.06) -- (295.67,189.06) -- cycle ;
\draw   (411.72,142.71) -- (434.93,142.71) -- (434.93,165.89) -- (411.72,165.89) -- cycle ;
\draw   (318.88,189.06) -- (342.09,189.06) -- (342.09,212.24) -- (318.88,212.24) -- cycle ;
\draw   (365.3,165.89) -- (388.51,165.89) -- (388.51,189.06) -- (365.3,189.06) -- cycle ;
\draw   (318.88,165.89) -- (342.09,165.89) -- (342.09,189.06) -- (318.88,189.06) -- cycle ;
\draw   (342.09,165.89) -- (365.3,165.89) -- (365.3,189.06) -- (342.09,189.06) -- cycle ;
\draw   (388.51,165.89) -- (411.72,165.89) -- (411.72,189.06) -- (388.51,189.06) -- cycle ;
\draw   (342.09,189.06) -- (365.3,189.06) -- (365.3,212.24) -- (342.09,212.24) -- cycle ;
\draw   (295.67,189.06) -- (318.88,189.06) -- (318.88,212.24) -- (295.67,212.24) -- cycle ;
\draw   (365.3,189.06) -- (388.51,189.06) -- (388.51,212.24) -- (365.3,212.24) -- cycle ;
\draw   (411.72,189.06) -- (434.93,189.06) -- (434.93,212.24) -- (411.72,212.24) -- cycle ;
\draw   (388.51,189.06) -- (411.72,189.06) -- (411.72,212.24) -- (388.51,212.24) -- cycle ;
\draw   (318.88,212.24) -- (342.09,212.24) -- (342.09,235.42) -- (318.88,235.42) -- cycle ;
\draw   (295.67,212.24) -- (318.88,212.24) -- (318.88,235.42) -- (295.67,235.42) -- cycle ;
\draw   (434.93,73.18) -- (458.14,73.18) -- (458.14,96.35) -- (434.93,96.35) -- cycle ;
\draw   (434.93,96.35) -- (458.14,96.35) -- (458.14,119.53) -- (434.93,119.53) -- cycle ;
\draw   (434.93,119.53) -- (458.14,119.53) -- (458.14,142.71) -- (434.93,142.71) -- cycle ;
\draw   (434.93,142.71) -- (458.14,142.71) -- (458.14,165.89) -- (434.93,165.89) -- cycle ;
\draw   (434.93,189.06) -- (458.14,189.06) -- (458.14,212.24) -- (434.93,212.24) -- cycle ;
\draw   (434.93,165.89) -- (458.14,165.89) -- (458.14,189.06) -- (434.93,189.06) -- cycle ;
\draw   (342.09,212.24) -- (365.3,212.24) -- (365.3,235.42) -- (342.09,235.42) -- cycle ;
\draw   (365.3,212.24) -- (388.51,212.24) -- (388.51,235.42) -- (365.3,235.42) -- cycle ;
\draw   (411.72,212.24) -- (434.93,212.24) -- (434.93,235.42) -- (411.72,235.42) -- cycle ;
\draw   (388.51,212.24) -- (411.72,212.24) -- (411.72,235.42) -- (388.51,235.42) -- cycle ;
\draw   (434.93,212.24) -- (458.14,212.24) -- (458.14,235.42) -- (434.93,235.42) -- cycle ;
\draw   (458.14,212.24) -- (481.35,212.24) -- (481.35,235.42) -- (458.14,235.42) -- cycle ;
\draw   (458.14,189.06) -- (481.35,189.06) -- (481.35,212.24) -- (458.14,212.24) -- cycle ;
\draw   (458.14,165.89) -- (481.35,165.89) -- (481.35,189.06) -- (458.14,189.06) -- cycle ;
\draw   (458.14,142.71) -- (481.35,142.71) -- (481.35,165.89) -- (458.14,165.89) -- cycle ;
\draw   (458.14,119.53) -- (481.35,119.53) -- (481.35,142.71) -- (458.14,142.71) -- cycle ;
\draw   (458.14,96.35) -- (481.35,96.35) -- (481.35,119.53) -- (458.14,119.53) -- cycle ;
\draw   (458.14,73.18) -- (481.35,73.18) -- (481.35,96.35) -- (458.14,96.35) -- cycle ;
\draw   (295.67,50) -- (318.88,50) -- (318.88,73.18) -- (295.67,73.18) -- cycle ;
\draw   (318.88,50) -- (342.09,50) -- (342.09,73.18) -- (318.88,73.18) -- cycle ;
\draw   (342.09,50) -- (365.3,50) -- (365.3,73.18) -- (342.09,73.18) -- cycle ;
\draw   (365.3,50) -- (388.51,50) -- (388.51,73.18) -- (365.3,73.18) -- cycle ;
\draw   (388.51,50) -- (411.72,50) -- (411.72,73.18) -- (388.51,73.18) -- cycle ;
\draw   (411.72,50) -- (434.93,50) -- (434.93,73.18) -- (411.72,73.18) -- cycle ;
\draw   (434.93,50) -- (458.14,50) -- (458.14,73.18) -- (434.93,73.18) -- cycle ;
\draw   (458.14,50) -- (481.35,50) -- (481.35,73.18) -- (458.14,73.18) -- cycle ;

\draw (184.44,131.12) node   {$\lor $};
\draw (138.02,131.12) node   {$\lor $};
\draw (114.81,131.12) node   {$\lor $};
\draw (91.6,131.12) node   {$\lor $};
\draw (91.6,61.59) node   {$0$};
\draw (114.81,61.59) node   {$0$};
\draw (91.6,84.77) node   {$0$};
\draw (114.81,84.77) node   {$0$};
\draw (138.02,61.59) node   {$0$};
\draw (138.02,84.77) node   {$0$};
\draw (184.44,61.59) node   {$0$};
\draw (207.65,61.59) node   {$0$};
\draw (230.86,61.59) node   {$0$};
\draw (254.07,61.59) node   {$0$};
\draw (254.07,84.77) node   {$0$};
\draw (91.6,107.94) node   {$0$};
\draw (114.81,107.94) node   {$0$};
\draw (138.02,107.94) node   {$0$};
\draw (254.07,107.94) node   {$0$};
\draw (91.6,154.3) node   {$0$};
\draw (91.6,177.48) node   {$0$};
\draw (91.6,200.65) node   {$0$};
\draw (91.6,223.83) node   {$0$};
\draw (114.81,223.83) node   {$0$};
\draw (138.02,177.48) node   {$0$};
\draw (138.02,223.83) node   {$0$};
\draw (161.23,177.48) node   {$0$};
\draw (161.23,223.83) node   {$0$};
\draw (207.65,177.48) node   {$0$};
\draw (230.86,177.48) node   {$0$};
\draw (254.07,177.48) node   {$0$};
\draw (207.65,200.65) node   {$0$};
\draw (230.86,154.3) node   {$0$};
\draw (254.07,200.65) node   {$0$};
\draw (207.65,223.83) node   {$0$};
\draw (230.86,223.83) node   {$0$};
\draw (254.07,223.83) node   {$0$};
\draw (161.23,200.65) node   {$0$};
\draw (138.02,200.65) node   {$0$};
\draw (114.81,200.65) node   {$0$};
\draw (114.81,177.48) node   {$0$};
\draw (114.81,154.3) node   {$0$};
\draw (138.02,154.3) node   {$0$};
\draw (161.23,154.3) node   {$0$};
\draw (184.44,107.94) node   {$0$};
\draw (184.44,84.77) node   {$0$};
\draw (207.65,84.77) node   {$0$};
\draw (230.86,84.77) node   {$0$};
\draw (207.65,107.94) node   {$0$};
\draw (230.86,107.94) node   {$0$};
\draw (161.23,131.12) node   {$\lor $};
\draw (376.91,131.12) node   {$\lor $};
\draw (400.12,154.3) node   {$\lor $};
\draw (423.33,154.3) node   {$\lor $};
\draw (376.91,107.94) node   {$\lor $};
\draw (376.91,84.77) node   {$\lor $};
\draw (469.74,154.3) node   {$\lor $};
\draw (376.91,61.59) node   {$\lor $};
\draw (307.28,61.59) node   {$0$};
\draw (330.49,61.59) node   {$0$};
\draw (307.28,84.77) node   {$0$};
\draw (330.49,84.77) node   {$0$};
\draw (353.7,61.59) node   {$0$};
\draw (353.7,84.77) node   {$0$};
\draw (400.12,61.59) node   {$0$};
\draw (423.33,61.59) node   {$0$};
\draw (446.53,61.59) node   {$0$};
\draw (469.74,61.59) node   {$0$};
\draw (469.74,84.77) node   {$0$};
\draw (307.28,107.94) node   {$0$};
\draw (330.49,107.94) node   {$0$};
\draw (353.7,107.94) node   {$0$};
\draw (469.74,107.94) node   {$0$};
\draw (469.74,131.12) node   {$0$};
\draw (307.28,154.3) node   {$0$};
\draw (307.28,177.48) node   {$0$};
\draw (307.28,200.65) node   {$0$};
\draw (307.28,223.83) node   {$0$};
\draw (330.49,223.83) node   {$0$};
\draw (353.7,177.48) node   {$0$};
\draw (353.7,223.83) node   {$0$};
\draw (376.91,177.48) node   {$0$};
\draw (376.91,223.83) node   {$0$};
\draw (423.33,177.48) node   {$0$};
\draw (446.53,177.48) node   {$0$};
\draw (469.74,177.48) node   {$0$};
\draw (423.33,200.65) node   {$0$};
\draw (446.53,200.65) node   {$0$};
\draw (469.74,200.65) node   {$0$};
\draw (423.33,223.83) node   {$0$};
\draw (446.53,223.83) node   {$0$};
\draw (469.74,223.83) node   {$0$};
\draw (376.91,200.65) node   {$0$};
\draw (353.7,200.65) node   {$0$};
\draw (330.49,200.65) node   {$0$};
\draw (330.49,177.48) node   {$0$};
\draw (330.49,154.3) node   {$0$};
\draw (353.7,154.3) node   {$0$};
\draw (400.12,131.12) node   {$0$};
\draw (400.12,107.94) node   {$0$};
\draw (400.12,84.77) node   {$0$};
\draw (423.33,84.77) node   {$0$};
\draw (446.53,84.77) node   {$0$};
\draw (423.33,107.94) node   {$0$};
\draw (423.33,131.12) node   {$0$};
\draw (446.53,107.94) node   {$0$};
\draw (446.53,131.12) node   {$0$};
\draw (446.53,154.3) node   {$\lor $};
\draw (353.7,131.12) node   {$0$};
\draw (330.49,131.12) node   {$0$};
\draw (307.28,131.12) node   {$0$};
\draw (161.23,61.59) node   {$0$};
\draw (161.23,84.77) node   {$0$};
\draw (161.23,107.94) node   {$0$};
\draw (207.65,131.12) node   {$\lor $};
\draw (230.86,131.12) node   {$\lor $};
\draw (254.07,131.12) node   {$\lor $};
\draw (184.44,154.3) node   {$0$};
\draw (207.65,154.3) node   {$0$};
\draw (227.65,197.48) node   {$0$};
\draw (254.07,154.3) node   {$0$};
\draw (184.44,177.48) node   {$0$};
\draw (184.44,200.65) node   {$0$};
\draw (184.44,223.83) node   {$0$};
\draw (400.12,177.48) node   {$0$};
\draw (400.12,200.65) node   {$0$};
\draw (400.12,223.83) node   {$0$};
\draw (376.91,154.3) node   {$\lor $};
\end{tikzpicture}
  \caption{Left: a wire-gadget with input in $(3,0)$ and output in $(3,7)$. Right: a wire-gadget with input in $(0,3)$ and output in $(4,7)$.}
 \label{fig:wiregadgets}
\end{figure}
 
Let $C$ be a grid-embedded Boolean circuit defined over a directed acyclic square grid $G(n)$.  Our meta-blocks will have size $\Delta \times \Delta$, where $\Delta = 8(n+2)$. The block in coordinates $(i,j)$ of $C$ will be simulated by a meta-block constructed as follows (see also Figure \ref{fig:meta-block}): consider a subdivision of the meta-block (observe that a meta-block has dimensions $\Delta \times \Delta = 8(n+2) \times 8(n+2)$) in a grid of dimensions $(n+2)\times (n+2)$, where each cell is a grid of $8\times 8$ blocks. The meta-block will receive the inputs from cell $(8i+3,0)$ (west input) and $(0,8j+3)$ (north input). On the other hand, it will output the values of the simulated block through the block in coordinates $(8(i+1)+3, \Delta)$ (east output) and $(\Delta, 8(j+1)+3)$ (south output). In other words, the west input of  the meta block  is the the block in coordinates $(3,0)$ of cell $(i,0)$; the north input is block in coordinates $(0,3)$ of cell $(0,j)$;  the east output is the block in coordinates $(3,7)$ of cell $(i+1,n+2)$ and the south output is the block in coordinates $(7,3)$ of cell $(n+2, j+1)$. 
 
We start placing on cell $(i,j)$ the gadget that simulates the operator of the block $(i,j)$ (See Figure \ref{fig:blockgadgets}). From cell $(i,0)$ until cell $(i,j-1)$, we place wire gadgets with inputs in $(3,0)$ and output in $(3,7)$. Similarly, from cell $(0,j)$ until cell $(i-1,j)$ we put wire gadgets with inputs in $(0,3)$ and output in $(7,3)$. Roughly speaking, this path of block wires will transmit the inputs received in cells $(i,0)$ and $(0,j)$ to the block gadget in $(i,j)$. 
 
Then, from cell $(i,j+1)$ until $(i,n+1)$ we place wire gadgets with inputs in $(4,0)$ and output in $(4,7)$, and from cell $(i+1,j)$ until $(n+1,j)$ we place wire gadgets with inputs in $(0,4)$ and output in $(7,4)$. Besides, we put wire gadgets in the following positions: 
\begin{itemize}
  \item  In cell $(i, n+2)$ we put a wire gadget with input in $(4,0)$ and output in $(7,3)$
   \item In cell $(i+1, n+2)$ put a wire gadget with input in $(0,3)$ and output in $(3,7)$
   \item In cell $(n+1, j)$ put a wire gadget with input in $(0,4)$ and output in $(3,7)$
   \item In cell $(n+1,j+1)$ put a wire gadget with input in $(3,0)$ and output in $(7,3)$
\end{itemize} Roughly speaking, this sequence of  wire gadgets will transmit the outputs of the gadget in cell $(i,j)$, to the output blocks of the meta-block. Finally, in every cell not yet defined, place a grid of $8\times 8$ fixed-value blocks.

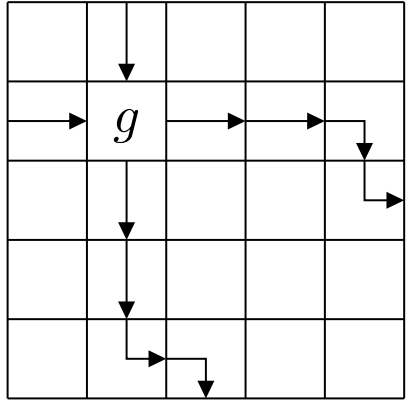
\begin{figure}[h]
 \centering
\tikzset{every picture/.style={line width=0.75pt}} 
\begin{tikzpicture}[x=0.75pt,y=0.75pt,yscale=-1,xscale=1]
\draw    (70,90) -- (70,290) ;
\draw    (70,90) -- (270,90) ;
\draw    (110,90) -- (110,290) ;
\draw    (150,90) -- (150,290) ;
\draw    (190,90) -- (190,290) ;
\draw    (230,90) -- (230,290) ;
\draw    (270,90) -- (270,290) ;
\draw    (70,130) -- (270,130) ;
\draw    (70,170) -- (270,170) ; 
\draw    (70,210) -- (270,210) ;
\draw    (70,250) -- (270,250) ;
\draw    (70,290) -- (270,290) ;
\draw    (150,150) -- (188,150) ;
\draw [shift={(190,150)}, rotate = 180] [fill={rgb, 255:red, 0; green, 0; blue, 0 }  ][line width=0.75]  [draw opacity=0] (8.93,-4.29) -- (0,0) -- (8.93,4.29) -- cycle    ;
\draw    (70,150) -- (108,150) ;
\draw [shift={(110,150)}, rotate = 180] [fill={rgb, 255:red, 0; green, 0; blue, 0 }  ][line width=0.75]  [draw opacity=0] (8.93,-4.29) -- (0,0) -- (8.93,4.29) -- cycle    ;
\draw    (130,170) -- (130,208) ;
\draw [shift={(130,210)}, rotate = 270] [fill={rgb, 255:red, 0; green, 0; blue, 0 }  ][line width=0.75]  [draw opacity=0] (8.93,-4.29) -- (0,0) -- (8.93,4.29) -- cycle    ;
\draw    (130,90) -- (130,128) ;
\draw [shift={(130,130)}, rotate = 270] [fill={rgb, 255:red, 0; green, 0; blue, 0 }  ][line width=0.75]  [draw opacity=0] (8.93,-4.29) -- (0,0) -- (8.93,4.29) -- cycle    ;
\draw    (130,210) -- (130,248) ;
\draw [shift={(130,250)}, rotate = 270] [fill={rgb, 255:red, 0; green, 0; blue, 0 }  ][line width=0.75]  [draw opacity=0] (8.93,-4.29) -- (0,0) -- (8.93,4.29) -- cycle    ;
\draw    (190,150) -- (228,150) ;
\draw [shift={(230,150)}, rotate = 180] [fill={rgb, 255:red, 0; green, 0; blue, 0 }  ][line width=0.75]  [draw opacity=0] (8.93,-4.29) -- (0,0) -- (8.93,4.29) -- cycle    ;
\draw    (230,150) -- (250,150) -- (250,168) ;
\draw [shift={(250,170)}, rotate = 270] [fill={rgb, 255:red, 0; green, 0; blue, 0 }  ][line width=0.75]  [draw opacity=0] (8.93,-4.29) -- (0,0) -- (8.93,4.29) -- cycle    ;
\draw    (250,170) -- (250,190) -- (268,190) ;
\draw [shift={(270,190)}, rotate = 180] [fill={rgb, 255:red, 0; green, 0; blue, 0 }  ][line width=0.75]  [draw opacity=0] (8.93,-4.29) -- (0,0) -- (8.93,4.29) -- cycle    ;
\draw    (130,250) -- (130,270) -- (148,270) ;
\draw [shift={(150,270)}, rotate = 180] [fill={rgb, 255:red, 0; green, 0; blue, 0 }  ][line width=0.75]  [draw opacity=0] (8.93,-4.29) -- (0,0) -- (8.93,4.29) -- cycle    ;
\draw    (150,270) -- (170,270) -- (170,288) ;
\draw [shift={(170,290)}, rotate = 270] [fill={rgb, 255:red, 0; green, 0; blue, 0 }  ][line width=0.75]  [draw opacity=0] (8.93,-4.29) -- (0,0) -- (8.93,4.29) -- cycle    ;
\draw (130,152.5) node[scale=2]   {$g$};
\end{tikzpicture}
 \caption{A meta-block simulating a block with operator $g$, which is placed in the cell $(1,1)$ of grid-embedded Boolean circuit $C$ defined over a directed acyclic square grid $G(3)$. The meta block consists in a grid of $5\times 5$ gadgets. The arrow represent the wire gadgets, white cells represent fixed-value gadgets, and the cell denoted $g$ represents the gadget of the block with Boolean operator $g$.}
 \label{fig:meta-block}
\end{figure} 
  
 Using this construction we obtain a way of simulating any grid-embedded Boolean circuit with one that only has conjunction, disjunction, fixed-value and selector blocks. In Figure \ref{fig:gridsimulationmetablocks} we give an example of a simulation of a grid-embedded circuit $C$ defined over a directed acyclic square grid $G(3)$. Moreover this simulating grid-embedded Boolean circuit can be constructed in polynomial time in the size of the simulated circuit. 
  
\begin{figure}
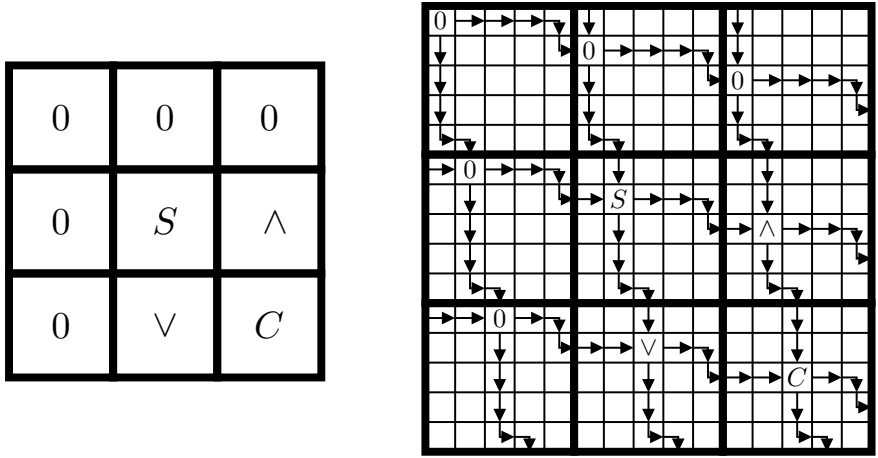

\tikzset{every picture/.style={line width=0.75pt}} 

 
 \caption{A representation of the simulation of an arbitrary grid-embedded Boolean circuit $C$ with a restricted set of blocks. The circuit $C$ is depicted in the left side of the figure, and is defined over a directed acyclic square grid $G(3)$. In the right side of the figure is the representation of the restricted grid-embedded Boolean circuit $\tilde{C}$ that simulates $C$ using only disjunction, conjunction, selector and fixed-value blocks. }
 \label{fig:gridsimulationmetablocks}
 \end{figure}
Let us call \RGCSAT the problem \GCSAT when the input grid-embedded Boolean circuit is restricted to have only disjunction, conjunction, fixed-value and selector blocks. 

 \begin{theorem}\label{theo:RGCSATisNPC}  
\RGCSAT is \NPC.
\end{theorem}

\begin{proof}
We reduce \GCSAT to \RGCSAT. Let $(C,v)$ be an instance of \GCSAT, where $C$ is a grid-embedded Boolean circuit and $v $ is a block with coordinates $(p,q)$. The reduction constructs a grid-embedded Boolean circuit $\tilde{C}$ using only disjunction, conjunction, fixed-value and selector blocks according to the constructions giving above. Let $\tilde{v}$ be the block in coordinates $(7,4)$ (i.e. the south output block) of the cell $(p,q)$ of the meta block in coordinates $(p,q)$ representing block $v$. From the constructions and arguments given above, it is direct that $\tilde{v}$ is satisfiable in $\tilde{C}$ if and only if $v$ is satisfiable in $C$. We deduce that $\RGCSAT$ is \NPC.
\end{proof}

\subsection{The complexity of rule $S22$}

In this subsection we show that for rule $S22$ problem \AsyncStability is \NPC. We remark that as we have already stated in previous section \AsyncStability is in \NP. Thus, it suffices to show that it is \NP-hard.  The proof consists of a series of patterns in the two-dimensional grid that, iterated asynchronously, according to rule $S22$, simulate disjunction, conjunction, selector and fixed-value blocks. Then, the \NP-Completeness follow from Theorem \ref{theo:RGCSATisNPC}.

As we mentioned in the introduction, in \cite{StabilityFTCA} it is shown that rule $S22$ is capable of simulating monotone Boolean circuits under the synchronous updating scheme.  More precisely, it is shown that certain patterns in the two-dimensional grid can simulate disjunction, conjunction, crossing, signal-multiplier and fixed-value blocks, and therefore the synchronous version of the problem is \PtC. Unfortunately, these gates are not robust under the asynchronous update of rule $S22$, in the sense that in certain updating schemes these patterns incorrectly produce simulated true-signals in cases that the simulated block should output a false signal (we call this situation a \emph{false positive}). 

Nevertheless, we show that it is possible to simulate disjunction, conjunction selector and fixed-value blocks by four patterns of dimensions $10\times 10$. We call these patterns disjunction, conjunction, selector and fixed-value pattern, respectively. We exhibit these patterns in Figure \ref{fig:gatesS22}. 

\begin{figure}
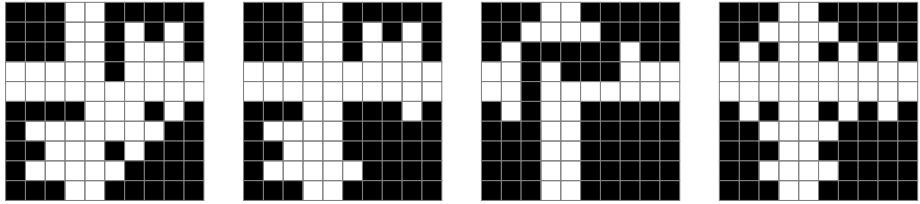

\centering


\caption{Patterns of rule $S22$ that simulate disjunction (left), conjunction (middle-left), selector (middle-right) and fixed-value (right) blocks. Black cells represent cells in state $1$, while white cells represent cells in state $0$.}
\label{fig:gatesS22}
\end{figure}

On each pattern, we identify a pair of cells in the north edge, namely $n_1 = (0,3)$ and $n_2=(0,4)$, and a pair of cells in the west edge, namely $w_1 = (3,0)$ and $w_2 = (4,0)$. We call these cells \emph{input cells}.  Our simulation considers that when a block receives true-signal through the north input (respectively through the west input), then in the corresponding simulating pattern $n_1$ or $n_2$ (respectively $w_1$ or $w_2$) will become active. Similarly, we identify two cells in the bottom edge, namely $s_1 = (9,3)$ and $s_2 = (9,4)$ and a pair of cells in the east edge, namely $e_1 = (3,9)$ and $w_2 = (4,9)$. When a block outputs a true-signal through the south output (respectively the east output), then in the corresponding simulating pattern cell $s_1$ or $s_2$ (respectively $e_1$ or $e_2$) will become active. 

In the following lemma, we show that the simulating patterns of disjunction, conjunction and fixed-value blocks have the advantage that they do not produce \emph{false positives} under any updating scheme. 

\begin{lemma}\label{lem:robustsym}
Let $x$ be the conjunction, disjunction or fixed-value pattern. Let $y$ be a configuration of the square grid $G(10)$ that is equal to the conjunction, disjunction or fixed-value pattern in every cell except $n_1$, $n_2$, $w_1$, $w_2$, $s_1$, $s_2$, $e_1$ and $e_2$. If $s_1$, $s_2$, $e_1$ and $e_2$ are inactive in $y$, then
\begin{itemize}
\item[a)] For every updating scheme $\sigma$ of $G(10)$ and every $t\geq 0$, 
$$S22^{\sigma(t)}(y)_{s_1} \leq g(y_{n_1}  \vee  y_{n_2} , y_{w_1}  \vee y_{w_2} )$$
$$S22^{\sigma(t)}(y)_{s_2} \leq g(y_{n_1}  \vee y_{n_2} , y_{w_1}  \vee y_{w_2} )$$
$$S22^{\sigma(t)}(y)_{e_1} \leq g(y_{n_1}  \vee y_{n_2} , y_{w_1}  \vee y_{w_2} )$$
$$S22^{\sigma(t)}(y)_{e_1} \leq g(y_{n_1}  \vee y_{n_2} , y_{w_1}  \vee y_{w_2} )$$
\item[b)] There exists an updating scheme $\sigma$ of $G(10)$ and $t>0$ such that:
$$S22^{\sigma(t)}(y)_{s_1} = g(y_{n_1}  \vee  y_{n_2} , y_{w_1}  \vee y_{w_2} )$$
$$S22^{\sigma(t)}(y)_{s_2} = g(y_{n_1}  \vee y_{n_2} , y_{w_1}  \vee y_{w_2} )$$
$$S22^{\sigma(t)}(y)_{e_1} = g(y_{n_1}  \vee y_{n_2} , y_{w_1}  \vee y_{w_2} )$$
$$S22^{\sigma(t)}(y)_{e_1} = g(y_{n_1}  \vee y_{n_2} , y_{w_1}  \vee y_{w_2} )$$

\end{itemize}
where $g(p,q) = p\vee q$ if $x$ is the disjunction pattern, $g(p,q) = p \wedge q$ if $x$ is the conjunction pattern, and $g(p,q) = 0$ if $x$ is the fixed-value pattern.

Moreover, if $n_1$, $n_2$, $w_1$, $w_2$ are inactive, then every cell in the pattern different than the input or output cells are stable.
\end{lemma}

\begin{proof}

Observe first that the lemma straightforwardly holds when $y_{n_1} = y_{n_2} = y_{w_1} = y_{w_2} = 0$, because the disjunction, conjunction and fixed-value patterns are defined as fixed points for rule $S22$ (the number of active neighbors of each cell is different than $2$ on each cell). Therefore, in this case $(a)$ and $(b)$ hold when $x$ is any of the three patterns.   We now assume that one of $n_1, n_2, w_1$ or $w_2$ is initially active. By the symmetry of the pattern, we can assume without loss of generality, $w_1$ or $w_2$ is active.

First, observe that if $x$ is the fixed-value pattern, then the configuration on all  cells with exception of the inputs and output are in a fixed-point (see Figure \ref{fig:gatesS22}). As $g(y_{n_1}  \vee y_{n_1} , y_{w_1}  \vee y_{w_2} ) = 0$ we deduce that $(a)$ and $(b)$ hold.  This still holds when $y_{n_1} = y_{n_2} =  y_{w_1} = y_{w_2} = 1$. We deduce that, when $x$ is the fixed-value pattern, (a) and (b) hold for every combination of states of cells $n_1$, $n_2$, $w_1$ and $w_2$. 

Now suppose that $x$ is a disjunction pattern. As $g(y_{n_1}  \vee y_{n_1} , y_{w_1}  \vee y_{w_2} ) = 1$, we obtain that in this case $(a)$ trivially holds.  On the other hand, it is not hard to construct an updating scheme that activates the output gates on this case. In Figure \ref{fig:conjunctionS22dyn} is shown an example of an updating scheme that activates the output cells, when $w_2$ is active. We deduce that $(b)$ holds when $x$ is the disjunction pattern. 

Therefore, (a) and (b) hold when $w_1$ or $w_2$ are active. By the symmetry of the pattern, the same is true when $n_1$ or $n_2$ are active. We deduce that, when $x$ is the disjunction pattern, (a) and (b) hold  
 for every combination of states of cells $n_1$, $n_2$, $w_1$ and $w_2$. 

\begin{figure}[h]
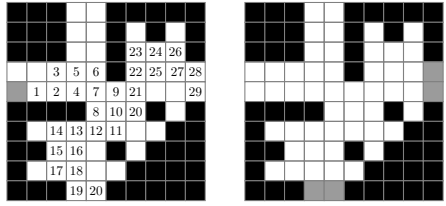

\centering

\caption{Left: Example of an updating scheme of the disjunction pattern that produces the correct output value in the output cells, when $w_2$ is active. The numbers represent the order in which the cells are updated. Unnumbered cells are updated in an arbitrary order after updating cell numbered $29$ ($e_2$). Right: a representation of the cells (in gray) that can become active when a signal comes through the output cells. }
\label{fig:disjunctionS22dyn}
\end{figure}

Suppose then that $x$ is the conjunction gadget. First, if $n_1$ and $n_2$ are inactive, then  $g(y_{n_1}  \vee y_{n_1} , y_{w_1}  \vee y_{w_2} ) = g(1,0) = 0$. In this case, observe that cells in coordinates $(3,3), (3,4), (4,3)$ and $(4,4)$ have all initially $4$ inactive neighbors. On the other hand, trying all possible updating scheme of cells $(3,1), (3,2), (4,1)$ and $(4,2)$, we obtain that cells $(3,3)$ and $(3,4)$ have at most one active neighbors, while all the other inactive cells in the pattern will remain with the same number of active neighbors. We deduce that $s_1, s_2, e_1$ and $e_2$ remain inactive for every updating scheme of the pattern. Therefore, (a) and (b) hold when $w_1$ or $w_2$ are active and $n_1$ and $n_2$ are inactive. By symmetry of the pattern (a) and (b) also hold in the complementary case, i.e., when $n_1$ or $n_2$ are active and $w_1$ and $w_2$ are inactive. 

Suppose now that one of $n_1, n_2$ is active and one of  $w_1, w_2$ is active. Then   $g(y_{n_1}  \vee y_{n_1} , y_{w_1}  \vee y_{w_2} ) = g(1,1) = 1$. In this case (a) trivially holds. On the other hand, we can exhibit an updating scheme that activates the output gates in latter case. In Figure \ref{fig:conjunctionS22dyn} an example of an updating scheme that activates the output when $n_2$ and $w_2$ are active cells is shown. Thus, (b) hold.

We deduce that, when $x$ is the conjunction pattern, (a) and (b) hold for every combination of states of cells $n_1$, $n_2$, $w_1$ and $w_2$. 

\begin{figure}[h]
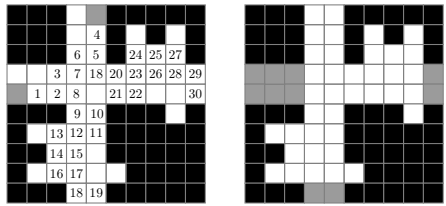

\centering

\caption{Left: Example of an updating scheme of the disjunction pattern that produces the correct output value in the output cells, when $w_2$ is active. The numbers represent the order in which the cells are updated. Unnumbered cells are updated in an arbitrary order after updating cell numbered $30$ ($e_2$). Right: a representation of the cells (in gray) that can become active when a signal comes through the output cells, and/or when only one signal comes through the input cells. }
\label{fig:conjunctionS22dyn}
\end{figure}

Finally, we remark that, when $s_1, s_2, w_1, w_2$ are inactive then when $x$ is any of the three patterns, then cells $(8,3), (8,4), (3,8)$ and $(4,8)$ have at least three inactive neighbors that are not output cells. Therefore, the number of active neighbors of every inactive cell in the pattern (eventually different to the output cells) is different than $2$ (see Figures \ref{fig:disjunctionS22dyn} and \ref{fig:conjunctionS22dyn}). Therefore, when  $s_1, s_2, w_1, w_2$ are inactive every cell in the pattern different than the input or the output cells is stable. 
\end{proof}

Roughly, Lemma \ref{lem:robustsym} states three useful properties of the disjunction, conjunction and fixed-value patterns. First, they do no produce \emph{false positives}. Second, if they should output a true signal, there is an updating scheme that produces it. Third, if the signal comes through the wrong direction, for example through one of the output cells, then it is impossible that such a signal enters to the pattern and changes one of the input cells or the output cells in the other side.  

The selector pattern have a different behavior. Observe that in the selector pattern, the dynamic of every cell in coordinates $(i,j)$ such that $i\geq 3$ and $j\geq 3$ is independent of the values of the states of the input cells. 
 The following lemma states that each updating scheme of the selector patterns simulate one of the choices of the selector block, namely the east-selector or the south-selector blocks.

\begin{lemma}\label{lem:selectorS22}
Let $x$ be the selector pattern and let  $\sigma$ be an updating scheme of the cells in $x$. Then either $\{s_1, s_2\}$ are stable and at least one of $\{n_1, n_2\}$ is iterated, or $\{n_1, n_2\}$ are stable and at least one of $\{s_1, s_2\}$ is iterated. 
\end{lemma}

\begin{proof}

Let $\sigma$ be any updating scheme of the cells in $x$. Let $v_1$ and $v_2$ be the cells in coordinates $(4,5)$ and $(4,6)$, respectively. Observe that every cell different than $v_1$ or $v_2$ has a number of active neighbors different than $2$, and therefore they wont change if they are updated before $v_1$ and $v_2$. Suppose that $v_1$ is iterated before $v_2$. Then, after that time-step $v_2$ has three active neighbors and becomes stable, implying that all cells in columns greater than $6$ are stable. Similarly, if $v_2$ is iterated before $v_1$, then $v_1$ is stable and all cells in columns smaller than $4$ that belong to the connected component of inactive cells containing $v_1$ (i.e. cells with coordinates $(i,j)$, with $i\geq 4$ and $j \in \{3,4\}$) are stable. 
On the other hand, if $v_1$ is iterated, then it is not hard to see that any updating scheme of the pattern eventually iterates $s_1$ or $s_2$. Similarly, if $v_2$ is iterated any updating scheme of $x$ eventually iterates $e_1$ or $e_2$. In Figure \ref{fig:selectorS22dyn} we represent two examples of updating schemes activating output cells, one when $v_1$ is iterated and the other when $v_2$ is the one that becomes active. 
\end{proof}

\begin{figure}[h]
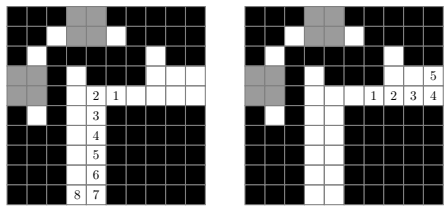

\centering

\caption{Left: Example of an updating scheme of the selector pattern that produces an output value in the southern output cells (south-selector). Right: Example of an updating scheme of the selector pattern that produces an output value in the eastern output cells (east-selector). The numbers represent the order in which the cells are updated. Unnumbered cells can be updated in any order after that the numbered cells are updated.}
\label{fig:selectorS22dyn}
\end{figure}

Roughly, Lemma \ref{lem:selectorS22} indicates that every updating scheme of the selector pattern simulates either a south-selector (when $v_1$ is iterated before $v_2$), or an east-selector (when $v_2$ is iterated before $v_1$). Therefore, for each updating $\sigma$ scheme of the selector pattern, we call $\sigma$ an \emph{east-selector updating scheme} if $\sigma$ updates $v_2$ before $v_1$. Otherwise, we call $\sigma$ an \emph{south-selector updating scheme}. Observe that Figure \ref{fig:selectorS22dyn} shows that there exists at least one east-selector updating scheme, and at least one south-selector updating scheme. 

We are now ready to state the main result of this section. 

\begin{theorem}
For rule $S22$, the problem \AsyncStability is \NPC.
\end{theorem}

\begin{proof}
In order to show our result we reduce \RGCSAT to \AsyncStability for rule $S22$. Let $(C,v)$ be an instance of \RGCSAT, defined over a directed square grid $G(n)$. Let $\tilde{x}$ be the configuration of active and active cells on the square grid $G(10n)$, constructed replacing each block $b$ of $C$ with the  pattern of the corresponding type given by Figure \ref{fig:gatesS22}. For each block $b$ of $C$, call $x_b$ the set of cells in the pattern representing $b$, and call $s_b$ the output cell $s_1$ of $x_b$ (observe that it is possible to choose any other output cell).  We claim that $s_b$ is unstable if and only if $b$ is satisfiable. 

Suppose first that $b$ is satisfiable, and let $u$ be a truth-assignment of $C$ satisfying $b$. Remember that a truth-assignment of $C$ is a choice over the selectors blocks, picking them as south or east-selectors. Let $\sigma$ the update scheme that first updates each selector pattern, in such a way that $\sigma$ restricted to the pattern is a east or south-selector updating scheme depending on the choices given by $u$. After updating every selector pattern, we define $\sigma$  to update each cell inside a conjunction or disjunction pattern according to the order given by Figures \ref{fig:disjunctionS22dyn} and \ref{fig:conjunctionS22dyn} and we define it in way such that, it continues updating each pattern in the same order than the circuit is evaluated. From Lemma \ref{lem:robustsym}  (b), we obtain that necessarily $s_b$ is unstable. 

Suppose now that $v_b$ is unstable, and let $\sigma$ an updating scheme that iterates it. For each selector $b$ of $C$, consider the restriction of $\sigma$ to the cells of $x_b$, and call it $\sigma_b$. Now we consider the truth assignment $u$ such that $u_b = 1$ if $\sigma_b$ is a east-selector updating scheme and $u_b = 0$ otherwise. We claim that $C(u)_b = 1$. Indeed, the choice of $\sigma$ assures that there exists a $t>0$ such that $S22^{\sigma(t)} (x)_{s_b} = 1$. On the other hand, Lemma \ref{lem:robustsym} (a) implies that the truth-value of $b$ is lower-bounded by   $S22^{\sigma(t)} (x)_{s_b}$. Therefore, we deduce that $b$ is satisfiable.

We conclude that for rule $S22$  \AsyncStability is \NPC.
\end{proof}

\section{Concluding Remarks}\label{sec:conclu}
In this paper we have analyzed the dynamics of asynchronous freezing cellular automata by studying the computational complexity of \AsyncStability problem which consist in answering whether there exists an update scheme that iterates a cell in asynchronous freezing cell automata.

First we have shown that in the one-dimensional case, the algorithm of \cite{goles:hal-01294144} for the synchronous case can be extended for the asynchronous update scheme implying that the problem is in $\NL$.  

As our main aim was to essentially understand what makes a rule simple or complex, we wondered whether the efficient verification algorithm found for the one-dimensional FCA can be extended to the two-dimensional case with two states, both in the triangular grid and square grid topologies.

We have found that, even when we were able to exhibit efficient parallel algorithms to solve \AsyncStability for the vast majority of life-like freezing rules, there exist some rules for which the problem is \NPC. Specifically,  we have found that in the triangular grid \AsyncStability is in $\NC$ for all life-like FCA, as well as all the rules in the square grid with the exception of the rule $S22$, which is \NPC. Moreover, we remark that, unlike the one dimensional case for which it was possible to exhibit one algorithm in order to solve the latter problem for all the rules, we introduced different approaches to solve \AsyncStability: i) infiltration approach and ii) monotony approach (by studying the rules that we have called monotone-like). In this regard, we observe that the infiltration approach is intrinsically interesting as it can be used to extend the obtained results to rules that are not necessarily Life-like rules. More precisely, consider the set $I_F$ associated to a given rule $F$. Then, by using the infiltration approach one can derive that all the rules such that $1 \in I_F$ are \NC for the triangular grid as we can always determine whether the objective cell is in $V_+$.  However, for the square grid it is not possible to straightforwardly use the latter approach to solve \AsyncStability. For example, if $I_F = \{1,4\}$, we do not know how predict the state of the objective cell in the case in which, initially, it has two active neighbors. Contrarily, as the number of neighbors is just three in the triangular grid, the problem can be solved directly as the objective cell is in $V_+$.

Our classification of simple life-like rules contain two groups, the \emph{infiltration rules}, which are the rules $F$ for which $1 \in \mathcal{I}_F$; and the \emph{monotone-like rules}, which are the rules $F$ for which $N-1 \in \mathcal{I}_F$, where $N$ is the size of the neighborhood. This sets contain all life-like rules defined over the triangular grid, and almost all life-like rules in the square grid, with exception of rule $S22$, which results to be \emph{complex} (i.e. \AsyncStability is \NPC). In the light of the obtained results, a future work could consider the asynchronous dynamics of Life-like FCA in two-dimensions considering the Moore neighborhood. Obviously by the same arguments given in this paper,  \AsyncStability is in \NC for all the  rules such that $1 \in \mathcal{I}_F$ and $N-1 \in \mathcal{I}_F$. An interesting research question are: how complex is this problem in the other cases?. Are all other rules \NPC?  We remark that the well-known \emph{Life-without-death} belongs to this set of rules. 

Finally, we remark that the \NP-completeness of the rule $S22$ is closely related with the capability of the dynamics to simulate three specific gadgets: the conjunction, disjunction and the selector blocks. Moreover, the constructions satisfy a specific topological restriction, which is the directed-grid oriented with south-east edges.  As consequence of the latter, we have shown that it is possible to simulate crossings and ultimately simulate SAT.  We consider that this simulation framework is of independent interest and could be used in the study of other asynchronous CA (including, for instance, Life-like FCA defined with the Moore neighborhood).


\section{Acknowledgements}

This work has been partially supported by: CONICYT via  PAI + Convocatoria Nacional Subvenci\'on a la Incorporaci\'on en la Academia A\~no 2017 + PAI77170068 and FONDECYT  11190482	(P.M.), CONICYT via PFCHA / DOCTORADO NACIONAL/2018 - 21180910 + PIA AFB 170001 (M.R.W) and ECOS C16E01 (E.G and M.R.W.); and CONICYT via Programa Regional STIC-AmSud (CoDANet) c\'od. 19-STIC-03 (E.G. and P.M.).  Additionally, D.M. thanks the support of Universidad Adolfo Ib\'a\~nez.

\bibliographystyle{elsarticle-num}
\bibliography{biblio}

\begin{thebibliography}{10}
\expandafter\ifx\csname url\endcsname\relax
  \def\url#1{\texttt{#1}}\fi
\expandafter\ifx\csname urlprefix\endcsname\relax\def\urlprefix{URL }\fi
\expandafter\ifx\csname href\endcsname\relax
  \def\href#1#2{#2} \def\path#1{#1}\fi

\bibitem{wolfram1983statistical}
S.~Wolfram, Statistical mechanics of cellular automata, Reviews of modern
  physics 55~(3) (1983) 601.

\bibitem{martin1984algebraic}
O.~Martin, A.~M. Odlyzko, S.~Wolfram, Algebraic properties of cellular
  automata, Communications in mathematical physics 93~(2) (1984) 219--258.

\bibitem{wolfram1984universality}
S.~Wolfram, Universality and complexity in cellular automata, Physica D:
  Nonlinear Phenomena 10~(1-2) (1984) 1--35.

\bibitem{wolfram1984computation}
S.~Wolfram, Computation theory of cellular automata, Communications in
  mathematical physics 96~(1) (1984) 15--57.

\bibitem{wolfram1985undecidability}
S.~Wolfram, Undecidability and intractability in theoretical physics, Physical
  Review Letters 54~(8) (1985) 735.

\bibitem{adamatzky2015actin}
A.~Adamatzky, R.~Mayne, Actin automata: Phenomenology and localizations,
  International Journal of Bifurcation and Chaos 25~(02) (2015) 1550030.

\bibitem{lehotzky2019cellular}
D.~Lehotzky, G.~K. Zupanc, Cellular automata modeling of stem-cell-driven
  development of tissue in the nervous system, Developmental neurobiology
  (2019).

\bibitem{deveaux2019defining}
W.~Deveaux, K.~Hayashi, K.~Selvarajoo, Defining rules for cancer cell
  proliferation in trail stimulation, NPJ systems biology and applications
  5~(1) (2019) 5.

\bibitem{hoekstra2010simulating}
A.~G. Hoekstra, J.~Kroc, P.~M. Sloot, Simulating complex systems by cellular
  automata, Springer, 2010.

\bibitem{torrens2001cellular}
P.~M. Torrens, D.~O'Sullivan, Cellular automata and urban simulation: where do
  we go from here? (2001).

\bibitem{hegselmann1998understanding}
R.~Hegselmann, A.~Flache, Understanding complex social dynamics: A plea for
  cellular automata based modelling, Journal of Artificial Societies and Social
  Simulation 1~(3) (1998) 1.

\bibitem{10.1007/11786986_13}
T.~Neary, D.~Woods, P-completeness of cellular automaton rule 110, in:
  M.~Bugliesi, B.~Preneel, V.~Sassone, I.~Wegener (Eds.), Automata, Languages
  and Programming, Springer Berlin Heidelberg, Berlin, Heidelberg, 2006, pp.
  132--143.

\bibitem{conf/focs/Banks70}
E.~R. Banks,
  \href{http://dblp.uni-trier.de/db/conf/focs/focs70.html#Banks70}{Universality
  in cellular automata}, in: SWAT (FOCS), IEEE Computer Society, 1970, pp.
  194--215.
\newline\urlprefix\url{http://dblp.uni-trier.de/db/conf/focs/focs70.html#Banks70}

\bibitem{di2008computational}
P.~Di~Lena, L.~Margara, Computational complexity of dynamical systems: the case
  of cellular automata, Information and Computation 206~(9-10) (2008)
  1104--1116.

\bibitem{cannataro1995parallel}
M.~Cannataro, S.~Di~Gregorio, R.~Rongo, W.~Spataro, G.~Spezzano, D.~Talia, A
  parallel cellular automata environment on multicomputers for computational
  science, Parallel Computing 21~(5) (1995) 803--823.

\bibitem{cornforth2003artificial}
D.~Cornforth, D.~G. Green, D.~Newth, M.~Kirley, Do artificial ants march in
  step? ordered asynchronous processes and modularity in biological systems,
  in: Proceedings of the eighth international conference on Artificial life,
  MIT Press, 2003, pp. 28--32.

\bibitem{fates2006fully}
N.~Fates, {\'E}.~Thierry, M.~Morvan, N.~Schabanel, Fully asynchronous behavior
  of double-quiescent elementary cellular automata, Theoretical Computer
  Science 362~(1-3) (2006) 1--16.

\bibitem{fates2004experimental}
N.~A. Fat{\`e}s, M.~Morvan, An experimental study of robustness to asynchronism
  for elementary cellular automata, arXiv preprint nlin/0402016 (2004).

\bibitem{SCHONFISCH1999123}
B.~Schönfisch, A.~de~Roos,
  \href{http://www.sciencedirect.com/science/article/pii/S0303264799000258}{Synchronous
  and asynchronous updating in cellular automata}, Biosystems 51~(3) (1999) 123
  -- 143.
\newblock \href {https://doi.org/https://doi.org/10.1016/S0303-2647(99)00025-8}
  {\path{doi:https://doi.org/10.1016/S0303-2647(99)00025-8}}.
\newline\urlprefix\url{http://www.sciencedirect.com/science/article/pii/S0303264799000258}

\bibitem{kitagawa1974cell}
T.~Kitagawa, Cell space approaches in biomathematics, Mathematical Biosciences
  19~(1-2) (1974) 27--71.

\bibitem{robert2012discrete}
F.~Robert, Discrete iterations: a metric study, Vol.~6, Springer Science \&
  Business Media, 2012.

\bibitem{goles:hal-01294144}
E.~Goles, N.~Ollinger, G.~Theyssier,
  \href{https://hal.archives-ouvertes.fr/hal-01294144}{{Introducing Freezing
  Cellular Automata}}, in: {Cellular Automata and Discrete Complex Systems,
  21st International Workshop (AUTOMATA 2015)}, Vol.~24 of TUCS Lecture Notes,
  Turku, Finland, 2015, pp. 65--73.
\newline\urlprefix\url{https://hal.archives-ouvertes.fr/hal-01294144}

\bibitem{Greenlaw:1995}
R.~Greenlaw, H.~Hoover, W.~Ruzzo, Limits to Parallel Computation:
  {P-completeness} Theory, Oxford University Press, Inc., New York, NY, USA,
  1995.

\bibitem{arora2009computational}
S.~Arora, B.~Barak, Computational complexity: a modern approach, Cambridge
  University Press, 2009.

\bibitem{Cook:1971:CTP:800157.805047}
S.~A. Cook, \href{http://doi.acm.org/10.1145/800157.805047}{The complexity of
  theorem-proving procedures}, in: Proceedings of the Third Annual ACM
  Symposium on Theory of Computing, STOC '71, ACM, New York, NY, USA, 1971, pp.
  151--158.
\newblock \href {https://doi.org/10.1145/800157.805047}
  {\path{doi:10.1145/800157.805047}}.
\newline\urlprefix\url{http://doi.acm.org/10.1145/800157.805047}

\bibitem{toffoli1987cellular}
T.~Toffoli, N.~Margolus, Cellular automata machines: a new environment for
  modeling, MIT press, 1987.

\bibitem{ConwaysLife}
M.~Gardner, {The fantastic combinations of John Conway's new solitaire game
  ``life''}, Scientific American 223 (1970) 120--123.

\bibitem{Durand1999}
B.~Durand, Z.~R{\'o}ka, \href{https://doi.org/10.1007/978-94-015-9153-9_2}{The
  Game of Life: Universality Revisited}, Springer Netherlands, Dordrecht, 1999,
  pp. 51--74.
\newblock \href {https://doi.org/10.1007/978-94-015-9153-9_2}
  {\path{doi:10.1007/978-94-015-9153-9_2}}.
\newline\urlprefix\url{https://doi.org/10.1007/978-94-015-9153-9_2}

\bibitem{Berlekamp1982}
E.~Berlekamp, J.~Conway, R.~Guy, Winning Ways for your Mathematical Plays,
  Vol.~2, 1982.

\bibitem{RePEc:wop:safiwp:97-05-044}
D.~Griffeath, C.~Moore,
  \href{http://EconPapers.repec.org/RePEc:wop:safiwp:97-05-044}{Life without
  death is {P-Complete}}, Working papers, Santa Fe Institute (1997).
\newline\urlprefix\url{http://EconPapers.repec.org/RePEc:wop:safiwp:97-05-044}

\bibitem{StabilityMajority}
E.~Goles, P.~Montealegre-Barba, The complexity of the asynchronous prediction
  of the majority automata, Information \& Computation.

\bibitem{StabilityFTCA}
E.~Goles, D.~Maldonado, P.~Montealegre-Barba, N.~Ollinger, On the complexity of
  the stability problem of binary freezing totalistic cellular automata,
  Information \& Computation.

\bibitem{goles:hal-00914603}
E.~Goles, P.~Montealegre-Barba, I.~Todinca,
  \href{https://hal.inria.fr/hal-00914603}{{The complexity of the bootstraping
  percolation and other problems}}, {Theoretical Computer Science} 504 (2013)
  73--82.
\newblock \href {https://doi.org/10.1016/j.tcs.2012.08.001}
  {\path{doi:10.1016/j.tcs.2012.08.001}}.
\newline\urlprefix\url{https://hal.inria.fr/hal-00914603}

\bibitem{0022-3719-12-1-008}
J.~Chalupa, P.~L. Leath, G.~R. Reich,
  \href{http://stacks.iop.org/0022-3719/12/i=1/a=008}{Bootstrap percolation on
  a bethe lattice}, Journal of Physics C: Solid State Physics 12~(1) (1979)
  L31.
\newline\urlprefix\url{http://stacks.iop.org/0022-3719/12/i=1/a=008}

\bibitem{RePEc:wop:safiwp:96-08-060}
C.~Moore,
  \href{http://EconPapers.repec.org/RePEc:wop:safiwp:96-08-060}{Majority-vote
  cellular automata, ising dynamics, and p-completeness}, Working papers, Santa
  Fe Institute (1996).
\newline\urlprefix\url{http://EconPapers.repec.org/RePEc:wop:safiwp:96-08-060}

\bibitem{JaJa:1992:IPA:133889}
J.~J\'{a}J\'{a}, An Introduction to Parallel Algorithms, Addison Wesley Longman
  Publishing Co., Inc., Redwood City, CA, USA, 1992.

\bibitem{kari1994reversibility}
J.~Kari, Reversibility and surjectivity problems of cellular automata, Journal
  of Computer and System Sciences 48~(1) (1994) 149--182.

\bibitem{amoroso1972decision}
S.~Amoroso, Y.~N. Patt, Decision procedures for surjectivity and injectivity of
  parallel maps for tessellation structures, Journal of Computer and System
  Sciences 6~(5) (1972) 448--464.

\bibitem{Sutner2009ModelCO}
K.~Sutner, Model checking one-dimensional cellular automata, J. Cellular
  Automata 4 (2009) 213--224.

\bibitem{Delorme20113866}
M.~Delorme, J.~Mazoyer, N.~Ollinger, G.~Theyssier,
  \href{http://dx.doi.org/10.1016/j.tcs.2011.02.023}{Bulking {I}: an abstract
  theory of bulking}, Theoret. Comput. Sci. 412~(30) (2011) 3866--3880.
\newblock \href {https://doi.org/10.1016/j.tcs.2011.02.023}
  {\path{doi:10.1016/j.tcs.2011.02.023}}.
\newline\urlprefix\url{http://dx.doi.org/10.1016/j.tcs.2011.02.023}

\bibitem{Delorme20113881}
M.~Delorme, J.~Mazoyer, N.~Ollinger, G.~Theyssier,
  \href{http://dx.doi.org/10.1016/j.tcs.2011.02.024}{Bulking {II}:
  classifications of cellular automata}, Theoret. Comput. Sci. 412~(30) (2011)
  3881--3095.
\newblock \href {https://doi.org/10.1016/j.tcs.2011.02.024}
  {\path{doi:10.1016/j.tcs.2011.02.024}}.
\newline\urlprefix\url{http://dx.doi.org/10.1016/j.tcs.2011.02.024}

\bibitem{Goldschlager1977}
L.~M. Goldschlager, \href{http://doi.acm.org/10.1145/1008354.1008356}{The
  monotone and planar circuit value problems are log space complete for {P}},
  SIGACT News 9~(2) (1977) 25--29.
\newblock \href {https://doi.org/10.1145/1008354.1008356}
  {\path{doi:10.1145/1008354.1008356}}.
\newline\urlprefix\url{http://doi.acm.org/10.1145/1008354.1008356}

\end{thebibliography}

\end{document}